\documentclass{article}
\usepackage{graphicx,amsmath,amssymb,tikz,amsthm} 
\usepackage{arydshln}
\usepackage{mathtools}
\usepackage{rotating}
\usepackage[linktocpage=true,
  colorlinks=true, 
  pdfborder={0 0 0},
  linkcolor=blue,
  citecolor=red,
  filecolor=yellow,
  urlcolor=blue,
  bookmarks,
  pdfauthor={},
]{hyperref}
\usetikzlibrary{quantikz}
\usepackage{authblk}

\newtheorem{theorem}{Theorem}

\newtheorem{lemma}{Lemma}
\newtheorem{remark}{Remark}

\newtheorem{definition}{Definition}

\title{A probabilistic imaginary-time evolution quantum algorithm for advection-diffusion equation: Explicit gate-level implementation and comparisons to quantum\\ linear system algorithms}

\author[1,2]{Xinchi Huang\footnote{Email: huangxc@g.ecc.u-tokyo.ac.jp}}
\author[1,2]{Hirofumi Nishi}
\author[1,2]{Taichi Kosugi}
\author[1,2]{Yoshifumi Kawada}
\author[1,2,3,4]{Yu-ichiro Matsushita}
\affil[1]{Department of Physics, The University of Tokyo, Tokyo 113-0033, Japan}
\affil[2]{Quemix Inc., Taiyo Life Nihombashi Building, 2-11-2, Nihombashi Chuo-ku, Tokyo 103-0027, Japan}
\affil[3]{Quantum Materials and Applications Research Center, National Institutes for Quantum Science and Technology (QST), 2-12-1 Ookayama, Meguro-ku, Tokyo 152-8550, Japan}
\affil[4]{Laboratory for Materials and Structures, Institute of Innovative Research, Tokyo Institute of Technology, Yokohama 226-8503, Japan}

\date{\today}

\begin{document}

\maketitle

\begin{abstract}
Simulating differential equations on classical computers becomes an intractable problem if the grid size is extremely large. Quantum computers are believed to achieve a possibly exponential speedup in the matrix operation. In this paper, we propose a quantum algorithm for solving the advection-diffusion-reaction equation by employing a novel approximate probabilistic imaginary-time evolution (PITE) operator. 
First, the effectiveness of the proposed approximate PITE operator is justified by the theoretical evaluation of the error. 
Next, we construct the explicit quantum circuit to realize the imaginary-time evolution of the Hamiltonian coming from the advection-diffusion equation, whose gate complexity is logarithmic regarding the size of the discretized Hamiltonian matrix. 
Compared to the existing algorithms for the quantum linear system problem, our algorithm achieves an exponential speedup regarding the matrix size at the cost of a worse dependence on the error bound. 
Moreover, numerical simulations using gate-based quantum emulator for 1D/2D examples are also provided to verify our algorithm. 
Finally, we extend our algorithm to the coupled system of advection-diffusion equations to show the prospects for practical applications. 
\end{abstract}

\section{Introduction}
\label{sec:1}

Analysis of partial differential equations (PDEs) is one of the central research topics in mathematical physics since it helps to quantitatively describe the physical phenomena in electrodynamics, fluid dynamics, viscoelastic mechanics, quantum mechanics, etc. \cite{LiQin2012}
Accompanying the establishment of the mathematical theory, there is great progress in the derivations of the numerical solutions with plenty of methods, e.g. \cite{LeVeque1992, Pinder2018}. 
However, large-scale numerical simulation for PDEs is still an intractable problem because of the unpractical execution time and memory capacity limitations. 

One prospect to overcome this problem is to use quantum computers, which are expected to achieve possibly exponential speedup in many applications \cite{Babbush.2023, Childs.2003, Lee.2023, Yuan2020}. 
Harrow, Hassidim, and Lloyd \cite{HHL09} provided a pioneer work in solving linear systems of equations by quantum computers, which is now called the HHL algorithm. This work exploited the potential exponential speedup of the computational complexity of the quantum solver (compared to the classical ones) regarding the size of the system, assuming the oracles for encoding the matrix and the vector. 
Since then, there have been many follow-up works \cite{Ambainis2012, Clader.2013, Berry2014, BCOW17, Childs.2017, Kieferova.2019, Childs.2020, Lin.2020, Childs.2021, An.2022, Costa.2022, Fang.2023, Krovi2023, Berry.2024}, known as the quantum linear system algorithms (QLSAs), on this topic since linear PDEs can be discretized to linear systems using well-known numerical methods. Moreover, the QLSAs are recently developed to solve nonlinear differential equations using the method of Carleman linearization \cite{Krovi2023, Liu.2021}. 

Currently the QLSAs that give the best asymptotic behavior have the query complexity (queries to the oracles for the matrices and vectors) of $O(\kappa\mathrm{log} (1/\varepsilon))$, where $\kappa$ and $\varepsilon$ are the condition number of the matrix and the error bound, respectively, by applying a recent technique called quantum eigenstate filtering (QEF) \cite{Lin.2020, Costa.2022}. 
According to the authors in \cite{Costa.2022}, the above dependence fulfills the asymptotic lower bound regarding $\kappa$ and $\varepsilon$ that was established in the unpublished work by Harrow and Kothari. By a general assumption that there exist efficient gate implementations of the above matrix and vector oracles (up to the complexity $O(\mathrm{polylog} N)$), the best-order QLSAs have the gate complexity of $O(\kappa\mathrm{polylog}N\, \mathrm{log}(1/\varepsilon))$ with the matrix size $N $ (without counting the repetitions of the circuit for the readout). On the other hand, the best-order conventional classical algorithm for solving general linear systems is regarded as the sparse conjugate gradient (CG) method, which has the computational complexity of $O(\sqrt{\kappa}N\,\mathrm{log}(1/\varepsilon))$ \cite{Shewchuk1994, Saad2003}. 
\begin{table}[htb]
\centering
\caption{Comparison of the computational complexity of the quantum algorithm (without the readout) and classical algorithm regarding matrix size $N$ and error bound $\varepsilon$. The polynomial dependence on dimension $d$ is omitted. }
\label{sec1:tab1}
\scalebox{0.85}[0.85]{
\begin{tabular}{l|ccc}
\hline
& General linear systems & General second-order PDEs & Second-order PDEs \\
&                                 & (without spectral infor.)        & (with known spectral infor.) \\
\hline
QLSAs e.g. \cite{Lin.2020, Costa.2022, Krovi2023} & $O(\kappa\mathrm{polylog}N\, \mathrm{log}(1/\varepsilon))$ & $O(N^{2/d}\mathrm{polylog}N\, \mathrm{log}(1/\varepsilon))$ & $O(\mathrm{polylog}N\, \mathrm{log}(1/\varepsilon))$ \\
Sparse CG e.g. \cite{Shewchuk1994, Saad2003} & $O(\sqrt{\kappa}N\,\mathrm{log}(1/\varepsilon))$ & $O(N^{(d+1)/d}\,\mathrm{log}(1/\varepsilon))$ & $O(N\,\mathrm{log}(1/\varepsilon))$ \\
\hline
\end{tabular}
}
\end{table}
As for the $d$-dimensional second-order linear PDEs (e.g. the Poisson equations), it is well-known that the condition number scales as $O(dN^{2/d})$, which implies that the QLSAs (without counting the readout) would give a polynomial speedup regarding $N$, compared to the classical sparse CG method, see Table \ref{sec1:tab1}. Moreover, if one knows the exact spectral information of the differential operator/the discretized matrix (i.e., all the eigenvalues and eigenstates), then QLSA with an exponential speedup in $N$ is possible (see \cite{An.2023} proposed for the time evolution equations). The reason can be understood as that we can prepare the optimal preconditioning to reduce the condition number to $O(1)$ if we have the spectral information. However, it is not practical to know the exact spectral information. Without the knowledge of such information, there are also methods to reduce the condition number, see e.g. \cite{Clader.2013, Grote.1997, Chow2000} and the references therein, but the performance is case-by-case, and to find/encode a good preconditioner needs additional classical/quantum computations. 

The best-order QLSAs mentioned above have a key assumption that the matrix and vector oracles can be efficiently implemented within the gate complexity $O(\mathrm{polylog} N)$. This assumption holds for some well-structured matrices with limited numbers of distinct values \cite{Camps.2024, Sunderhauf.2024}, but is not clarified for the matrices coming from general PDEs with spatial varying coefficients, e.g., a function potential (actually the gate complexity is $\Omega(N)$ by the existing proposal in \cite{Sunderhauf.2024}). Moreover, in practical applications, the analytical relation between the error bound $\varepsilon$ and the matrix size $N$ is unknown (sometimes a rough order can be estimated, but one never knows the prefactor). One conventional way in engineering is to gradually increase $N$ by numerical simulations until the limit of the computational resources, and check whether the calculation is convergent. Thus, the scaleup in the matrix size is more important in some applications while the error bound is suitably fixed. In this paper, we mainly focus on the quantum acceleration in the matrix size and we propose a new quantum algorithm for the advection-diffusion-reaction equation  (see Eq.~\eqref{sec4:eq-gov}), which is an important PDE model for practical applications. Moreover, we avoid the matrix oracles and give the explicit gate implementation of the proposed quantum algorithm. Furthermore, the gate complexity of our algorithm achieves an exponential quantum speedup concerning $N$ even without the spectral information, at the cost of increasing the dependence on $1/\varepsilon$ to a power one (see Table \ref{sec3:tab1}). In detail, we adopt the method using the imaginary-time evolution (ITE) operator. 

For a given (time-independent) Hamiltonian $\mathcal{H}$, the solution to the first-order time evolution equation:
$$
\partial_t u(t) + \mathcal{H} u(t) = 0, \quad t\in (0,T], \quad u(0) = u_0,
$$
can be directly expressed using the ITE operator as follows:
$$
u(t) = \exp\left(-t\mathcal{H}\right) u_0, \quad t\in (0,T].
$$
Here, we consider the homogeneous problem (without source term) for simplicity. The quantum circuit for the inhomogeneous problem can be constructed similarly (i.e., we discretize the additional time integral), but we need to pay more attention to the total error estimate. We will address it in a future work. Under the above formulation, our target is to find a quantum circuit implementing a good approximation of the non-unitary ITE operator after the Hamiltonian $\mathcal{H}$ is discretized into a matrix $H_N$ whose size is $N\times N$. 
Although the quantum computers can treat only unitary operations, we know that any non-unitary operations can be embedded into a larger unitary matrix using ancillary qubits, which is known as block encoding (see e.g. Chapter 6 in \cite{Lin2022} and the references therein). The non-unitary operation is thus realized by measuring the ancillary qubits in some desired states, and we call the probability of deriving the desired states the success probability. Here, we refer to the block encoding of the ITE operator with only one ancillary qubit as the probabilistic imaginary-time evolution (PITE) operator according to the previous works \cite{Kosugi.2022, Nishi.2023, Brinet.2024}. The main contributions of this paper are shown in Table \ref{sec1:tab2} (the complexity of several methods are estimated for the advection-diffusion-reaction equations), and the paper is organized as follows:
\begin{sidewaystable}[htbp]
\centering
\caption{Comparison among PITE algorithms and classical/quantum algorithms for linear systems in the example of advection-diffusion-reaction equations. The complexity regarding two important parameters: matrix size $N$ and error bound $\varepsilon$, is addressed without counting the costs of preparation and readout of the statevectors. $d$ denotes the spatial dimension. }
\label{sec1:tab2}
\scalebox{0.73}[0.73]{\renewcommand\arraystretch{1.2}
\begin{tabular}{l|ccccc|p{7cm}}
\hline
  & Variations & Basic oracles & Query complexity$^{[a]}$ & Oracle complexity & Total complexity & Comments \\
  &                &                 & /Iterations & /Complexity per iter. &  &  \\
\hline
Sparse CG & $^\ast\ ${\bf BE} & N/A & $O(N^{1/d}(1/\varepsilon)\mathrm{log}(1/\varepsilon))$ & $O(N)$ & $\tilde O(N^{1+1/d}(1/\varepsilon))$ & Sparse CG is a conventional iterative method \\
e.g. \cite{Shewchuk1994} &  &  &  &  &  & for solving linear systems. Due to the $N$-dep- \\
& EM & N/A & $\tilde O(N^{1/d}\mathrm{polylog}(1/\varepsilon))$ & $\tilde O(N^{1+2/d}\mathrm{log}(1/\varepsilon))$ & $\tilde O(N^{1+3/d}\mathrm{polylog}(1/\varepsilon))$ & endence in condition number, preconditioning is needed for large matrices. \\
\hline
QLSAs, e.g. & QPE(=HHL) & RTEs for matrix & $\tilde O(N^{2/d}(1/\varepsilon))$ & $\tilde O(N^{2/d}(1/\varepsilon))$ & $\tilde O(N^{4/d}(1/\varepsilon)^2)$ & QSVT-based QLSA is one of the best-order \\
\cite{HHL09, Childs.2017, Krovi2023} & + BE & $A=I+\Delta\tau H_N$ &  &  &  & quantum algorithms for linear systems. The \\
 &  & $\Delta\tau=O(\varepsilon)$ &  &  &  & oracle complexity is not yet clarified if there \\
& QSVT & Block encoding & $O(N^{2/d}(1/\varepsilon)\mathrm{log}(1/\varepsilon))$ & $\Omega(\mathrm{polylog} N)^{[b]}$ & $\tilde O(N^{2/d}(1/\varepsilon))$ & exists a self-reaction potential. Because of the \\
 & + BE & for normalized &  &  &  & large condition number, QLSAs are not appli- \\
 &  & matrix $A$ &  &  &  & cable for large matrices, and hence precondi- \\
& $^\ast\ ${\bf QSVT} & Block encoding & $\tilde O(N^{2/d}\mathrm{polylog}(1/\varepsilon))^{[c]}$ & $\Omega(\mathrm{polylog} N)^{[b]}$ & $\tilde O(N^{2/d}\mathrm{polylog}(1/\varepsilon))$ & tioning is indispensable, see Sect.~\ref{sec:5}. \\
 & {\bf + EM} & for a normalized &  &  &  &  \\
 &  & enlarged matrix &  &  &  &  \\
\hline
PITE & APITE\cite{Kosugi.2022} & RTE for matrix & $O((1/\varepsilon)\mathrm{exp}(1/\varepsilon))$ & $O\left(\mathrm{polylog} N (1/\varepsilon)^{o(1)}\right)$ & $\tilde O\left(\mathrm{polylog} N \mathrm{exp}(1/\varepsilon)\right)$ & AAPITE algorithm is applicable for large ma- \\
algorithms & + ST1 & $H_N$ &  &  &  & trices, see Sect.~\ref{sec:4} and can be extended to cou- \\
(proposal) &  &  &  &  &  & pled systems, see Sect.~\ref{sec:6}. The oracle com- \\
 & $^\ast\ ${\bf AAPITE} & RTEs for matrices & $O((1/\varepsilon))$ & $O\left(\mathrm{polylog} N (1/\varepsilon)^{o(1)}\right)$ & $\tilde O\left(\mathrm{polylog} N (1/\varepsilon)^{1+o(1)}\right)$ & plexity relies on the FSM with periodic BCs. \\
 & {\bf + ST1} & $D_N^{(1)}, \sqrt{D_N^{(2)}}, \sqrt{V_N}$ &  &  &  & Thus, it is not directly applicable for the cases \\
 &  & see Sect.~\ref{sec:3} &  &  &  & with complicated boundaries. \\
\hline
\end{tabular}
}
\begin{flushleft}
\footnotesize BE: backward Euler,\ EM: enlarged matrix, see the matrix $L$ in \cite{Krovi2023},\ ST1: first-order Suzuki-Trotter formula,\ BC: boundary condition. In this paper, $O, \tilde O$ and $\Omega$ denote the Landau notations \cite{Knuth1976}. 

$\ast$ denotes the best-order algorithm in each block of algorithms. 

$[a]$ The $(1/\varepsilon)$ dependence in the query complexity comes from the time discretization error for AAPITE or BE. It can be reduced to $(1/\varepsilon)^{1/p}$ for some integer $p$ if $p$th-order time discretization methods and $p$-th order Suzuki-Trotter formula are used.  

$[b]$ According to Camps et al. \cite{Camps.2024}, the oracle complexity is in general $O(N)$ if there is a self-reaction potential function. On the other hand, using arithmetic operations with additional ancillary qubits, it is possible to derive the oracle complexity $O(\mathrm{polylog} N\, \mathrm{polylog}(1/\varepsilon))$ for an approximate implementation. 

$[c]$ Recent techniques in \cite{Lin.2020, Costa.2022} suggested a two-step algorithm to improve the $\varepsilon$ dependence to $\log(1/\varepsilon)$ in the query complexity by first taking a relatively large $\varepsilon_0>0$ and then applying the QEF to achieve the desirable error bound $\varepsilon$. 
\end{flushleft}
\end{sidewaystable}
\begin{itemize}
\item In Sect.~\ref{sec:2}, we propose a new approximate PITE operator called the alternative approximate PITE (AAPITE), which overcomes the crucial problem (vanishing success probability) of the previous ones \cite{Kosugi.2022, Nishi.2023} as a quantum differential equation solver. To justify the effectiveness of the new approximate PITE operator, the theoretical estimations on the $\ell^2$-error between the AAPITE and the ITE operators as well as the success probability are established, see Eqs.~\eqref{sec2:eq-err2}, \eqref{sec2:eq-sucprob}. 

\item In Sect.~\ref{sec:3}, based on the real-space grid method, a special Fourier spectral method (FSM), for the first-quantized Hamiltonian simulation, we apply the AAPITE to address the ITE operator for the advection-diffusion-reaction equation. 
By an explicit gate implementation of the real-time evolution (RTE) of the discretized Hamiltonian, we conclude that the gate complexity is logarithmic in the size of the discretized matrix (see Table \ref{sec3:tab1}).
Compared to the QLSAs in e.g. \cite{HHL09, Childs.2017, Krovi2023}, the number of ancillary qubits is much smaller and is independent of the desired error bound $\varepsilon$, see Table \ref{sec3:tab1}.
Owing to the minimal use of ancillary qubits, we believe our proposed algorithm fits the early fault-tolerant quantum computers (eFTQC), for which limited qubits (one to several hundred logical qubits) are allowed.

\item In Sect.~\ref{sec:4}, we demonstrate numerical simulations of our quantum algorithm for a 1D example and a 2D example with absorption potential in the center of the domain using a quantum emulator Qiskit \cite{Qiskit23}. Moreover, we check the numerical dependence on the parameters of the AAPITE as well as the coefficients of the differential equation in the $\ell^2$-errors between the quantum solutions and the analytical solution (or the exact solution by matrix operations).

\item In Sect.~\ref{sec:5}, we compare our quantum algorithm with a previous work using either a specific HHL algorithm or a variational quantum algorithm (VQA) based on the FDMs \cite{Ingelmann.2024}. The result indicates the theoretical advantage of the FSM, over the FDMs.
In addition, we find that our algorithm based on the AAPITE outperforms the one based on the original approximate PITE in \cite{Kosugi.2022}.
Moreover, compared to the QLSAs, our algorithm is more practical and has an exponential advantage in the scaling of the matrix size at the expense of a worse dependence on the error bound. 

\item In Sect.~\ref{sec:6}, we extend our quantum algorithm for a single advection-diffusion equation to a coupled system of equations, and we provide the explicit quantum circuit in the specific case of two equations. Moreover, we mention that the proposed quantum algorithm can be applied to nonlinear systems if time-step-wise measurement/tomography of the statevectors is allowed. Although the simulation of the nonlinear systems in the above way is currently expensive due to the repeated measurements and the lack of efficient readout, we provide the numerical results of Turing Pattern formulation and Burgers' equation for prospects. 
\end{itemize}
\section{Probabilistic imaginary-time evolution (PITE)}
\label{sec:2}

Let $n\in \mathbb{N}$ and suppose we are given a Hamiltonian $\mathcal{H}\in \mathbb{C}^{2^n\times 2^n}$. 
Our target is the implementation of the ITE operator $m_0 \mathrm{exp}\left(-\Delta\tau \mathcal{H}\right)$ for a given time step $\Delta\tau>0$ and $m_0\in (0,1]$. 
There is a recent way to approximate the ITE operator using the quantum singular value transformation (QSVT) developed in \cite{Gilyen.2019}. However, this technique requires access to the block encoding of $\mathcal{H}$ whose implementation is efficient if $\mathcal{H}$ is well-structured \cite{Camps.2024, Sunderhauf.2024}, but requires $O(N^2)$ operations in the worst scenario. Here, we apply a conventional way is to introduce one ancillary qubit and implement the following unitary operation on $n+1$ qubits (provided that $\mathcal{H}$ is a positive semi-definite Hermite operator): 
\begin{align*}
\mathcal{M}(\Theta) := 
\begin{pmatrix}
\cos\Theta & \mathrm{i}\sin\Theta \\
\mathrm{i}\sin\Theta & \cos\Theta
\end{pmatrix},
\end{align*}
where $\Theta := \arccos(m_0 \mathrm{exp}\left(-\Delta\tau \mathcal{H}\right))$. 
Then, for an arbitrary input quantum state $\ket{\psi}$, we have
\begin{align*}
\mathcal{M}(\Theta)(\ket{0}\otimes \ket{\psi}) = \ket{0}\otimes \cos\Theta\ket{\psi} + \ket{1} \otimes \mathrm{i}\sin\Theta\ket{\psi}.
\end{align*}
By post-selecting the ancillary qubit to be $\ket{0}$, we obtain the desired non-unitary operation:
$$
\ket{\psi} \longrightarrow \frac{\mathrm{exp}\left(-\Delta\tau\mathcal{H}\right)\ket{\psi}}{\|\mathrm{exp}\left(-\Delta\tau\mathcal{H}\right)\ket{\psi}\|}.
$$
Here, we need the assumption that $\mathcal{H}$ is positive semi-definite and Hermitian, so that $\mathcal{H}$ is diagonalizable and all the eigenvalues are real-valued and non-negative. 
This guarantees that all the eigenvalues of $\mathrm{exp}\left(-\Delta\tau \mathcal{H}\right)$ lie in $(0,1]$, and hence $\mathcal{M}(\Theta)$ is well-defined. 
The quantum circuit for the PITE based on the cosine function is given in Fig.~\ref{sec2:fig1}. Measuring the ancillary qubit, the PITE operator is implemented with the success probability $m_0^2\|\mathrm{exp}\left(-\Delta\tau \mathcal{H}\right)\ket{\psi}\|^2$.
\begin{figure}
\centering
\resizebox{12cm}{!}{
\includegraphics[keepaspectratio]{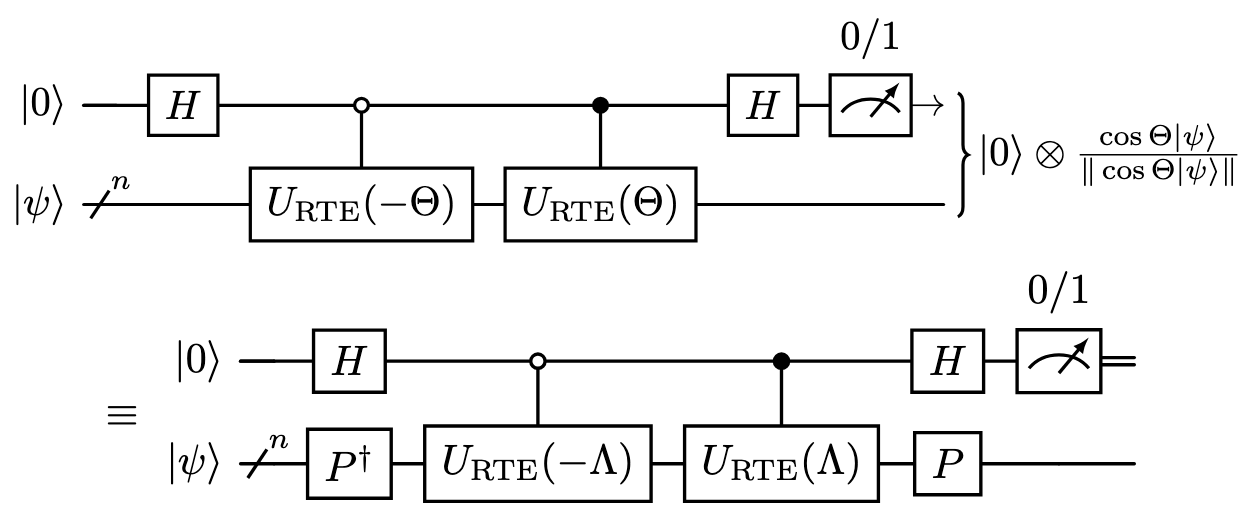}
}
\caption{Quantum circuits for PITE based on the cosine function. Here, $U_{\text{RTE}}(\Theta)$ denotes the RTE operator for the operator $\Theta=\arccos\left(m_0 \mathrm{exp}\left(-\Delta\tau \mathcal{H}\right)\right)$ defined by $U_{\text{RTE}}(\Theta) := \mathrm{exp}\left(-\mathrm{i}\Theta\right)$. 
The upper circuit is equivalent to the lower one provided that $\Theta$ is diagonalized by $P^\dag \Theta P = \Lambda$ for some unitary matrix $P$. }
\label{sec2:fig1}
\end{figure}
Although most previous works \cite{Kosugi.2022, Nishi.2023, Brinet.2024} consider the unitary operation: 
\begin{align*}
\widetilde{\mathcal{M}}(\Theta) := 
\begin{pmatrix}
\cos\Theta & -\sin\Theta \\
\sin\Theta & \cos\Theta
\end{pmatrix},
\end{align*}
instead of $\mathcal{M}(\Theta)$, these two implementations of $\mathrm{exp}\left(-\Delta\tau \mathcal{H}\right)$ based on $\mathcal{M}(\Theta)$ and $\widetilde{\mathcal{M}}(\Theta)$ are equivalent under the following unitary operations (i.e., two single-qubit gates): 
\begin{align*}
\mathcal{M}(\Theta) =
(Q\otimes I^{\otimes n}) \widetilde{\mathcal{M}}(\Theta) (Q^\dag\otimes I^{\otimes n}),
\end{align*}
where 
\begin{align*}
Q := \frac{1}{\sqrt{2}}
\begin{pmatrix}
1+\mathrm{i} & 0 \\
0 & -1+\mathrm{i}
\end{pmatrix}.
\end{align*}
Moreover, thanks to our assumption on $\mathcal{H}$, there exist a unitary operator $P$ and a diagonal matrix $\Lambda$ such that 
$$
\mathcal{M}(\Theta) = (I \otimes P)\mathcal{M}(\Lambda)(I \otimes P^\dag), \quad
\widetilde{\mathcal{M}}(\Theta) = (I \otimes P)\widetilde{\mathcal{M}}(\Lambda)(I \otimes P^\dag). 
$$
Then, the exact implementation of the PITE operator is demonstrated by the lower circuit in Fig.~\ref{sec2:fig1}, in which the RTE operators for the diagonal unitary matrices $U_{\text{RTE}}(\pm \Lambda)$ can be realized by using $2^n$ phase rotation gates and $2^n-2$ CNOT gates with depth $2^n$ \cite{Zhang.2024}. 
On the other hand, Mangin-Brinet et al. \cite{Brinet.2024} used the quantum circuit proposed in \cite{Mottonen.2004}, and gave the detailed circuit for $\widetilde{\mathcal{M}}(\Lambda)$ by $2^n$ $y$-axis rotation gates and $2^n$ CNOT gates with depth $2^{n+1}$. 
In general, we need to calculate the unitary $P$ using the eigenvalue decomposition and implement it by the quantum gates, which is not an easy task. There are some kinds of $\mathcal{H}$ such that $P$ is known and can be (approximately) implemented efficiently within depth $O(\mathrm{poly}\, n)$. For example, $P=I$ if $\mathcal{H}$ itself is diagonal, and $P$ is the shifted/centered quantum Fourier transform (QFT) if $\mathcal{H}$ is generated from the kinetic energy in the first-quantized formalism \cite{Kosugi.2022, Ollitrault.2020}, equivalently, $\mathcal{H}$ is the discretized matrix of the Laplacian with periodic boundary condition. 
To achieve the quantum advantage regarding $N$, we consider an approximate PITE operator whose implementation is efficient for the Hamiltonian that appeared in the differential equations. 

\subsection{Alternative approximate PITE}
\label{subsec:2-1}

Kosugi et al. \cite{Kosugi.2022} proposed an approximate PITE circuit with a restriction that the normalization factor $m_0$ does not equal one. 
By a shift of energy, this approximate PITE circuit and its derivatives are efficient in the problem of ground state preparation \cite{Kosugi.2022, Nishi.2023, Kosugi.2023}. 
However, such approximate PITE circuits are inefficient for realizing the ITE because they have an exponentially vanishing success probability.  
The details on the approximate PITE \cite{Kosugi.2022} are reconsidered in \ref{appA}. 
Here, we propose a new approximate PITE with $m_0=1$, which we call the alternative approximate PITE (AAPITE), to avoid such a problem. 

Let $m_0 = 1$. We rewrite $\Theta = \arccos\left(\mathrm{exp}\left(-(\sqrt{\Delta\tau \mathcal{H}})^2\right)\right)$. 
The notation $\sqrt{\Delta\tau \mathcal{H}}$ makes sense since we assume that $\mathcal{H}$ is a positive semi-definite Hermite operator.
Denote $g(x) := \arccos(\mathrm{exp}\left(-x^2\right))$, $x\ge 0$. 
By the Taylor expansion around $x_0\ge 0$ up to the first order, we obtain
$$
g(x) = g(x_0) + g^\prime(x_0)(x-x_0) + O(|x-x_0|^2).
$$
We calculate directly or use L'Hospital's rule, and we obtain $\lim_{x\to 0} g(x) = 0$ and $\lim_{x\to 0+} g^\prime(x) =\sqrt{2}$.
Taking $x_0=0$, we have the approximation:
\begin{align*}
\Theta = g(\sqrt{\Delta\tau \mathcal{H}}) = \sqrt{2}\sqrt{\Delta\tau \mathcal{H}} + O(\Delta\tau) \approx \sqrt{2\Delta\tau\mathcal{H}}. 
\end{align*}
In other words, 
$$
\mathrm{exp}\left(-\Delta\tau \mathcal{H}\right) = \cos(g(\sqrt{\Delta\tau \mathcal{H}})) \approx \cos(\sqrt{2\Delta\tau \mathcal{H}}). 
$$
Then, we obtain the quantum circuit for the AAPITE in Fig.~\ref{sec2:fig2}. 
\begin{figure}
\centering
\resizebox{12cm}{!}{
\includegraphics[keepaspectratio]{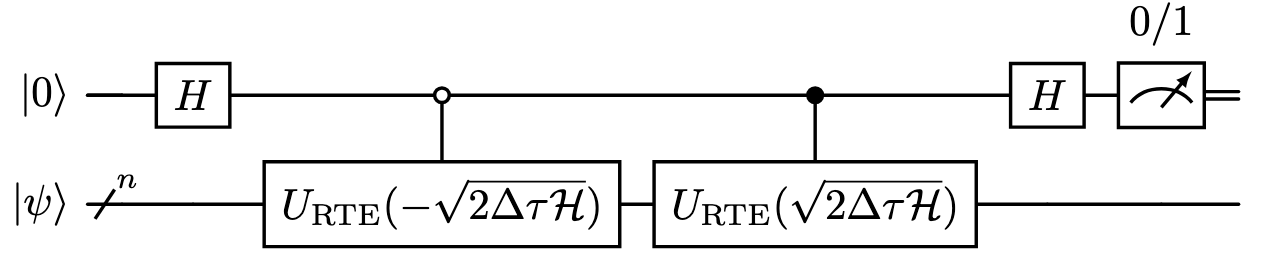}
}
\caption{A quantum circuit for the alternative approximate PITE.}
\label{sec2:fig2}
\end{figure}

We mention the complexity for a single PITE step. The main costs lie in the part of the controlled RTE operations related to $\sqrt{2\Delta\tau \mathcal{H}}$ for a given Hermite operator $\mathcal{H}$. 
For such controlled RTE operations, we should expect efficient implementations in gate complexity $O(\mathrm{polylog} N)$, which yields an exponential improvement compared to the exact implementation. In the next section, we provide an efficient implementation of the RTE of the half-order operator specified for a Hamiltonian dynamics in real space. Here, we stick to the justification of the AAPITE and establish the theoretical error estimate and the success probability. 

\subsection{Error estimate and success probability}
\label{subsec:2-2}

For $K\in \mathbb{N}$, we consider the operation of applying the AAPITE for $K$ times. 
We define the $\ell^2$-error after normalization by
\begin{align*}
\mathrm{Err} := \left\|\frac{\cos^K(\sqrt{2\Delta\tau \mathcal{H}})\ket{\psi}}{\|\cos^K(\sqrt{2\Delta\tau \mathcal{H}})\ket{\psi}\|}-\frac{\mathrm{exp}\left(-K\Delta\tau \mathcal{H}\right)\ket{\psi}}{\|\mathrm{exp}\left(-K\Delta\tau \mathcal{H}\right)\ket{\psi}\|}\right\|,
\end{align*}
where $\ket{\psi}$ is an initial state. 
Introduce also the $\ell^2$-error before the normalization 
$$
\widetilde{\mathrm{Err}} := \left\|\cos^K(\sqrt{2\Delta\tau \mathcal{H}})\ket{\psi}-\mathrm{exp}\left(-K\Delta\tau \mathcal{H}\right)\ket{\psi}\right\|.
$$
Then, we have the following error estimate. 
\begin{theorem}
\label{thm:aap-err-esti}
Let $N, K\in\mathbb{N}$, $\Delta\tau\in \mathbb{R}_{>0}$, and $\mathcal{H}\in \mathbb{C}^{N\times N}$ be a positive semi-definite Hermite matrix. For an arbitrarily given quantum state $\ket{\psi}$, there exists a constant $C>0$, depending only on $\mathcal{H}$ and $\ket{\psi}$, such that the following error estimate
$$
\widetilde{\mathrm{Err}} \le CK(\Delta\tau)^2,
$$
holds. Moreover, by the triangle inequality, we have
\begin{align}
\label{sec2:eq-err2}
\mathrm{Err} \le 2C K(\Delta\tau)^2/\|\mathrm{exp}\left(-K\Delta\tau \mathcal{H}\right)\ket{\psi}\|. 
\end{align}
\end{theorem}
Here, $C$ depends on the eigensystem of $\mathcal{H}$, hence there is an implicit dependence on $N$. In the application, we take $\mathcal{H}$ the discretized matrix of the differential operator $-\nabla^2$ or a bounded operator $V(x)I$. The proof is provided in \ref{prf:thm1}, and we show that $C$ depends only on the variation of the function $V(x)$ in the bounded cases, while $C$ has a uniform bound regarding $N$ in the Laplacian case provided that the initial condition is suffiently smooth.  

For a given error bound $\varepsilon>0$, we can choose $\Delta\tau=O(\varepsilon)$ so that the error bound is fulfilled, that is, 
\begin{equation}
\label{sec2:eq-errbound}
\widetilde{\mathrm{Err}}\le \varepsilon \|\mathrm{exp}\left(-K\Delta\tau \mathcal{H}\right)\ket{\psi}\|/2, \ \mbox{and} \ \mathrm{Err} \le 2\widetilde{\mathrm{Err}}/\|\mathrm{exp}\left(-K\Delta\tau \mathcal{H}\right)\ket{\psi}\| \le \varepsilon. 
\end{equation}
Equivalently, we have the choice $K=\Omega(\varepsilon^{-1})$.

Next, we discuss the success probability using the error estimate. 
We note that the quantum state after applying the AAPITE operator for $j$ times is 
$$
\cos^j\left(\sqrt{2\Delta\tau \mathcal{H}}\right)\ket{\psi}/\left\|\cos^j\left(\sqrt{2\Delta\tau \mathcal{H}}\right)\ket{\psi}\right\|, \quad j=1,\ldots,K,
$$
and the success probability is 
$$
\mathbb{P}_1(\ket{0}) = \prod_{j=1}^K \left\|\cos\left(\sqrt{2\Delta\tau \mathcal{H}}\right)\frac{\cos^{j-1}\left(\sqrt{2\Delta\tau \mathcal{H}}\right)\ket{\psi}}{\left\|\cos^{j-1}\left(\sqrt{2\Delta\tau \mathcal{H}}\right)\ket{\psi}\right\|}\right\|^2 
= \left\|\cos^K\left(\sqrt{2\Delta\tau \mathcal{H}}\right)\ket{\psi}\right\|^2.
$$
By the definition of $\widetilde{\mathrm{Err}}$ and the triangle inequalities, we obtain
\begin{align*}
\|\mathrm{exp}\left(-K\Delta\tau \mathcal{H}\right)\ket{\psi}\| - \widetilde{\mathrm{Err}} 
\le \|\cos^K(\sqrt{2\Delta\tau \mathcal{H}})\ket{\psi}\| 
\le \|\mathrm{exp}\left(-K\Delta\tau \mathcal{H}\right)\ket{\psi}\| + \widetilde{\mathrm{Err}}.
\end{align*}
Together with Eq.~\eqref{sec2:eq-errbound}, this yields the estimate of the success probability as follows:
\begin{equation}
\label{sec2:eq-sucprob}
\left(\frac{2-\varepsilon}{2}\right)^2 \|\mathrm{exp}\left(-K\Delta\tau \mathcal{H}\right) \ket{\psi}\|^2 \le \mathbb{P}_1(\ket{0}) \le \left(\frac{2+\varepsilon}{2}\right)^2 \|\mathrm{exp}\left(-K\Delta\tau \mathcal{H}\right) \ket{\psi}\|^2. 
\end{equation}
By noting that $(2-\varepsilon)/2$ and $(2+\varepsilon)/2$ are both of order $\Omega(1)$ with respect to $\varepsilon$, the success probability is proportional to $\|\mathrm{exp}\left(-K\Delta\tau \mathcal{H}\right) \ket{\psi}\|^2$, which depends only on $\ket{\psi}$ and $K\Delta\tau$. Besides, the proportional constant tends to $1$ as the error bound $\varepsilon$ tends to zero. 

\section{Application to advection-diffusion equation}
\label{sec:3}

In this section, we investigate the quantum circuit for an ITE operator $\mathrm{exp}\left(-T\mathcal{H}\right)$ for which $\mathcal{H}$ is a given Hamiltonian matrix corresponding to the following operator appeared in the advection-diffusion-reaction equation: 
\begin{align}
\label{sec3:eq-Ham}
\hat H = - a\nabla^2 + \mathbf{v}\cdot \nabla + V, 
\end{align}
where $a>0$ is a positive constant, $\mathbf{v}=(v_1, v_2, \ldots, v_d)^\mathrm{T}\in \mathbb{R}^d$ is a constant vector, and $V=V(\mathbf{x})\in \mathbb{R}$ is a real-valued potential function defined on the cubic domain $[0,L]^d$. Here, $d\in \mathbb{N}$ denotes the spatial dimension and $L>0$ is the length of the domain in each dimension. 
Let $n\in\mathbb{N}$, $N_0 = 2^n$, and $N = N_0^d$. In terms of the real-space grid method (a FSM-based method), we introduce the grid points, which are described by a $d$-dimensional vector $\mathbf{l} = (l_1,l_2,\ldots,l_d)^\mathrm{T}\in [\tilde N_0]^d := \{0, 1, \ldots, N_0-1\}^d$: 
\begin{align}
\label{sec3:eq-x}
\mathbf{p}_{\mathbf{l}} = \left(l_1 \frac{L}{N_0}, l_2 \frac{L}{N_0}, \ldots, l_d \frac{L}{N_0}\right)^\mathrm{T}\in \mathbb{R}^d.
\end{align}
Then, the above Hamiltonian can be discretized into an $N\times N$ matrix: 
\begin{align}
\label{sec3:eq-HN}
H_N := F_N \left(D_N^{(2)}-\mathrm{i}D_N^{(1)}\right) F_N^\dag + V_N,
\end{align}
where
\begin{align*}
&F_N = \frac{1}{N_0^{d/2}} \sum_{\mathbf{k}\in [\tilde N_0]^d}\sum_{\mathbf{l}\in [\tilde N_0]^d} \mathrm{exp}\left(\mathrm{i}\frac{2\pi}{N_0}\left(\mathbf{k}-\frac{N_0}{2}\right)\cdot\mathbf{l}\right) \ket{\mathbf{l}} \bra{\mathbf{k}},\\
&D_N^{(1j)} = \sum_{\mathbf{k}\in [\tilde N_0]^d} v_j (2\pi/L) \left(k_j-\frac{N_0}{2}\right) \ket{\mathbf{k}} \bra{\mathbf{k}}, \quad j=1,\ldots,d, \quad D_N^{(1)} = \sum_{j=1}^d D_N^{(1j)},\\ 
&D_N^{(2)} = \sum_{\mathbf{k}\in [\tilde N_0]^d} a (2\pi/L)^2\left|\mathbf{k}-\frac{N_0}{2}\right|^2 \ket{\mathbf{k}} \bra{\mathbf{k}}, \\
&V_N = \sum_{\mathbf{l}\in [\tilde N_0]^d} V(\mathbf{p}_{\mathbf{l}}) \ket{\mathbf{l}} \bra{\mathbf{l}}
= \sum_{\mathbf{l}\in [\tilde N_0]^d} V\left(\frac{L}{N_0}\mathbf{l}\right) \ket{\mathbf{l}} \bra{\mathbf{l}}.
\end{align*}
For the convenience, we call $D_N^{(2)}$ the kinetic part in the following context. 
Here, we use the notation in a quantum fashion (see e.g. \cite{NC10}) where $\ket{\mathbf{k}}=\ket{k_d}\otimes\cdots\otimes\ket{k_1}$ ($\bra{\mathbf{k}}=\bra{k_d}\otimes\cdots\otimes\bra{k_1}$) is the tensor product of column vectors (row vectors). $F_N$ and $F_N^\dag$ are the $d$-dimensional shifted QFT and its inverse, respectively. According to the definition of $F_N$, the $F_N$ operation on $dn$ qubits is equivalent to $d$ operations of the one-dimensional shifted QFT on each $n$ qubits. 
For completeness, we provide the theoretical derivation of the real-space grid method in \ref{appB}.

\subsection{Gate complexity}
\label{subsec:3-1}

We consider the ITE operator $\mathrm{exp}\left(-T H_N\right)$, where $T>0$ is a given target time. 
If we directly use the AAPITE circuit, then we need to know at least a lower bound of the eigenvalues of $H_N$, and give an implementation of the RTE operator $U_{\text{RTE}}(\sqrt{2\Delta\tau H_N})$, which can be inefficient as the operator $H_N$ is not diagonal (even not sparse). 
Although one can use the idea of diagonalization in \cite{An.2023}, the computational cost of the diagonalization itself can be more expensive than the PDE solver unless we have already known analytically the precise spectral information of the operator. To fit with the needs in practical applications, we avoid using the precise spectral information by introducing some approximations. 
Here, we choose $K\in \mathbb{N}$, $\Delta\tau = T/K$, and denote $V_0 := \min_{\mathbf{x}\in [0,L]^d} V(\mathbf{x})$. The first-order Suzuki-Trotter formula yields
\begin{align*}
\mathrm{exp}\left(-T H_N\right) &= \left(\mathrm{exp}\left(-\Delta\tau H_N\right)\right)^K \approx \left(\mathrm{exp}\left(-\Delta\tau F_N (D_N^{(2)}-\mathrm{i}D_N^{(1)}) F_N^\dag\right) \mathrm{exp}\left(-\Delta\tau V_N\right)\right)^K \\
&= \mathrm{exp}\left(-V_0 T\right)\left(F_N \mathrm{exp}\left(\mathrm{i}\Delta\tau D_N^{(1)}\right)\mathrm{exp}\left(-\Delta\tau D_N^{(2)}\right)F_N^\dag \mathrm{exp}\left(-\Delta\tau (V_N-V_0 I)\right)\right)^K. 
\end{align*}
Then, we can apply the AAPITE circuit to $\mathcal{H} = D_N^{(2)}$ and $\mathcal{H} = V_N-V_0 I$, respectively. 
Using this approximation, the RTE operators $U_{\text{RTE}}\left(\sqrt{2\Delta\tau D_N^{(2)}}\right)$ and $U_{\text{RTE}}\left(\sqrt{2\Delta\tau (V_N-V_0 I)}\right)$ can be implemented with gate complexity $O(\mathrm{polylog}\, N)$ since $D_N^{(2)}$ and $V_N-V_0 I$ are both diagonal matrices with known underlying functions. The quantum circuit based on the AAPITE is constructed by Fig.~\ref{sec3:Fig1} in each time step. 
\begin{figure}
\centering
\resizebox{15cm}{!}{
\includegraphics[keepaspectratio]{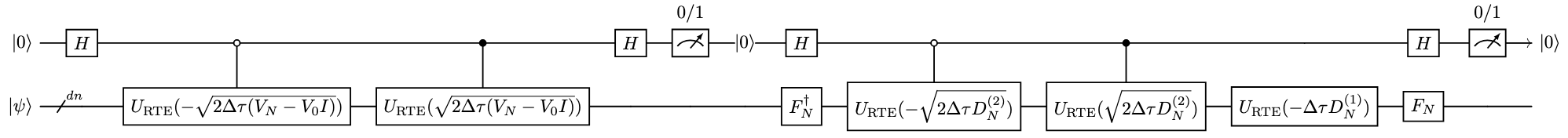}
}
\caption{Quantum circuit in each time step for the ITE for Hamiltonian $H_N$. Here, two AAPITE circuits are employed. The first one is for the potential part, and the second is for the kinetic part. }
\label{sec3:Fig1}
\end{figure}
Using the well-known strategies in Hamiltonian simulation, we can approximate the phases of the diagonals by piecewise polynomials, whose gate complexity is $O(\mathrm{poly}\, n)$. 
In the case of $d = 1$, the operator $U_{\text{RTE}}\left(\sqrt{2\Delta\tau D_N^{(2)}}\right)$ can be further efficiently implemented. 
Note that the eigenvalues of $\sqrt{2\Delta\tau D_N^{(2)}}$ are $\lambda_{k,D_N^{(2)}} = \sqrt{2a\Delta\tau} (2\pi/L)|k-N_0/2|$, $k=0,\ldots,N_0-1$, which is a piecewise linear function of $k$ with only two sub-intervals. Thus, it can be implemented with gate complexity $O(n)$ by the LIU method proposed in \cite{Huang.2024p}. If the potential function $V$ can be efficiently approximated by piecewise linear or quadratic function, then the total gate complexity of a single PITE step is $O(n^2)$.  
On the other hand, in the case of $d\ge 2$, we can first derive a squared distance register $\ket{|\mathbf{k}-N_0/2|^2}_{2n-1+\lceil\log_2 d\rceil}$, and then $U_{\text{RTE}}\left(\sqrt{2\Delta\tau D_N^{(2)}}\right)$ is implemented by applying a unitary diagonal matrix corresponding to a square root function to the squared distance register. Using piecewise $p$-th degree polynomial approximation for the square root function, the gate complexity is $O(n^p)$ \cite{Huang.2024p}. 
For completeness, we provide the explicit gate implementation of such a RTE operator related to the Laplace operator in \ref{appG}.  
If the potential function $V$ are sufficiently smooth and can be efficiently approximated, e.g. localized potential $V=V(|\mathbf{x}-\mathbf{x_0}|)$ for some point $\mathbf{x_0}\in \mathbb{R}^d$, then we can again introduce a squared distance register and implement $U_{\text{RTE}}\left(\sqrt{2\Delta\tau (V_N-V_0 I)}\right)$ with gate complexity $O(\mathrm{poly}\, n)$. 

There is another way to implement a unitary diagonal matrix by using the fixed-point arithmetic operations and the phase kickback strategy \cite{Kassal.2008, Jones.2012, Babbush.2018, Haner.2018, Sanders.2020}. Although the idea is again a piecewise polynomial approximation, this method uses plenty of ancillary qubits and gives a quantum circuit constructed mainly by the Clifford gates and the Toffoli gates that are suitable for the current error correction algorithms. In this paper, we adopt the method in \cite{Huang.2024p} (based on rotation gates) to minimize the use of ancillary qubits.  

\subsection{Error estimate and success probability}
\label{subsec:3-2}

We estimate the error between the AAPITE operator and the ITE operator. By denoting the operator:
$$
P := F_N \mathrm{exp}\left(\mathrm{i}\Delta\tau D_N^{(1)}\right)\cos\left(\sqrt{2\Delta\tau D_N^{(2)}}\right)F_N^\dag \cos\left(\sqrt{2\Delta\tau \tilde V_N}\right),
$$
we define the non-unitary operator of our algorithm in one time step: 
\begin{align*}
U_{\text{AAPITE}}(\Delta\tau)\ket{\psi} := \frac{P\ket{\psi}}{\left\|P\ket{\psi}\right\|}, 
\end{align*}
for any quantum state $\ket{\psi}$. 
Here and henceforth, we denote the shifted $H_N$ and $V_N$ by $\tilde H_N = H_N - V_0 I$ and $\tilde V_N = V_N - V_0 I$ so that the real parts of the eigenvalues of $\tilde H_N$ and $\tilde V_N$ are both non-negative.
A direct calculation yields
\begin{align*}
U_{\text{AAPITE}}^K(\Delta\tau)\ket{\psi} = \frac{P^K\ket{\psi}}{\left\|P^K\ket{\psi}\right\|}.
\end{align*}
For a given initial state $\ket{\psi_{(0)}}$, and $K\in \mathbb{N}$, we let $T=K\Delta\tau$ and define the $\ell^2$-error by
\begin{align*}
\tilde{E}_{\text{tot}} := \Bigg\|\mathrm{exp}\left(-T H_N\right)\ket{\psi_{(0)}} - \mathrm{exp}\left(-TV_0\right)P^K \ket{\psi_{(0)}} \Bigg\|.
\end{align*}
By the triangle inequality, the normalized $\ell^2$-error defined by
$$
E_{\text{tot}} := \left\|\frac{\mathrm{exp}\left(-T H_N\right)\ket{\psi_{(0)}}}{\left\|\mathrm{exp}\left(-T H_N\right)\ket{\psi_{(0)}}\right\|} - U_{\text{AAPITE}}^K(\Delta\tau) \ket{\psi_{(0)}} \right\|,
$$
satisfies
$$
E_{\text{tot}} \le 2\tilde{E}_{\text{tot}}/\left\|\mathrm{exp}\left(-T H_N\right)\ket{\psi_{(0)}}\right\|. 
$$
Combining the error estimate for the AAPITE in Sect.~\ref{subsec:2-2} with the error estimate for the Suzuki-Trotter formula, we obtain
\begin{equation}
\label{sec3:eq-err}
\tilde{E}_{\text{tot}} \le \tilde C T\mathrm{exp}\left(-TV_0\right) \Delta\tau, \quad \mbox{and hence} \quad E_{\text{tot}} \le \frac{2\tilde CT\Delta\tau}{\|\mathrm{exp}\left(-T \tilde H_N\right)\ket{\psi_{(0)}}\|} = O(\Delta\tau).  
\end{equation}
Here, $\tilde C$ is a constant which is independent of $T$ and $N$. The independence/uniformness on $N$ is crucial for deriving an exponential speedup regarding $N$. The details of the derivation are provided in \ref{prf:tot-err-esti}, and the numerical confirmation is provided in \ref{subsec:4-3-2}--\ref{subsec:4-3-4}. 
To achieve a given error bound $\varepsilon>0$, we take $\Delta\tau = O(\varepsilon)$, or equivalently $K=\Omega(\varepsilon^{-1})$. 
We omit the $T$ dependence because we have also $T$ dependence in the denominator of the error estimate and the effect of the denominator could be dominant since we consider a dissipative system. Of course, we can give the loose overhead that $\Delta\tau=O\left(\varepsilon\|\mathrm{exp}\left(-T\tilde H_N\right) \ket{\psi_{(0)}}\|/T\right)$, but we see from the numerical examples in Sect.~\ref{sec:4} that such a theoretical dependence on $T$ is far from sharp. 

Next, we consider the total success probability. According to our construction based on the AAPITE circuits, the success probability is given by
\begin{align*}
\mathbb{P}_{\text{tot}}(\ket{0}) = \left\|P^K\ket{\psi_{(0)}}\right\|^2.
\end{align*}
By the triangle inequalities and Eq.~\eqref{sec3:eq-err}, we have
\begin{align*}
\left\|\mathrm{exp}\left(-T \tilde H_N\right)\ket{\psi_{(0)}}\right\| - \tilde C T\Delta\tau \le \sqrt{\mathbb{P}_{\text{tot}}(\ket{0})} \le \left\|\mathrm{exp}\left(-T \tilde H_N\right)\ket{\psi_{(0)}}\right\| + \tilde C T\Delta\tau,
\end{align*}
which implies 
$$
\mathbb{P}_{\text{tot}}(\ket{0}) \approx \left\|\mathrm{exp}\left(-T \tilde H_N\right)\ket{\psi_{(0)}}\right\|^2 = \mathrm{exp}\left(2T V_0\right) \left\|\mathrm{exp}\left(-T H_N\right)\ket{\psi_{(0)}}\right\|^2,
$$
for any small $\Delta\tau$. 
\begin{remark}
We mention that $\left\|\mathrm{exp}\left(-T H_N\right)\ket{\psi_{(0)}}\right\|^2$ approximates the squared $L^2$-norm of the solution to the continuous equation $(\partial_t + \hat H) u(t) = 0$ with an $L^2$-normalized initial condition as $N\to \infty$. In other words, we have $\left\|\mathrm{exp}\left(-T H_N\right)\ket{\psi_{(0)}}\right\|^2$ approaches $\|u(T)\|_{L^2}^2/\|u(0)\|_{L^2}^2$ as $N$ tends to infinity. This implies that the success probability converges to some positive constant depending only on the equation and $T$ as $N\to \infty$ and $\Delta\tau\to 0$. 
Moreover, the factor $\mathrm{exp}\left(2TV_0\right)$ is necessary, otherwise the operator norm $\|\mathrm{exp}\left(-TH_N\right)\|$ can exceed one, which yields a contradiction that the success probability is larger than one. Furthermore, it is more efficient (higher success probability) if we shift the original differential operator by a constant potential so that its minimal eigenvalue (real part) is almost zero.   
\end{remark}

\subsection{Quantum resources estimation}

We sum up the required quantum resources for the ITE operator $e^{-T\mathcal{H}}$ using the AAPITE with the first-order Suzuki-Trotter formula in Table \ref{sec3:tab1} by focusing on the number of two-qubit gates (CNOT gates) and ancillary qubits. Due to our implementation of the real-time evolution operator \cite{Huang.2024p}, the circuit depth is at most the same order as the CNOT count.  
\begin{table}[htb]
\centering
\caption{Resources requirement for the ITE operator using the AAPITE with the first-order Suzuki-Trotter formula.}
\label{sec3:tab1}
\scalebox{0.95}[0.95]{
\begin{tabular}{l|ccc}
\hline
& CNOT count (Case 1) & CNOT count (Case 2) & Ancilla count \\
\hline
$d=1$ & $O\left((n^2+n^p)\varepsilon^{-1}\right)$ & $O\left((n^2+n^p\varepsilon^{-1/(p+1)})\varepsilon^{-1}\right)$ & $2$ \\
$d\ge 2$ & $O\left(d(n^3+n^p\varepsilon^{-1/(p+1)})\varepsilon^{-1}\right)$ & $O\left(d(n^3+n^p\varepsilon^{-1/(p+1)})\varepsilon^{-1}\right)$ & $2n+2+\lceil\log_2 d\rceil$ \\
\hline
\end{tabular}
}
\end{table}
Here, Case 1 denotes the case that the potential function itself is a (finite-interval) piecewise polynomial of degree $p\in \mathbb{N}$, while Case 2 denotes the case that the potential function is a piecewise $C^{(p+1)}([0, L]^d)$ function, so that it can be well approximated by a piecewise polynomial of degree $p$. Moreover, $\varepsilon$ denotes the error bound, and we need additional gate complexity $O(\varepsilon^{-1/(p+1)})$ in Case 2 that comes from the approximate error regarding the potential function. 
For $d\ge 2$, since we approximate a square root function for the kinetic part $D_N^{(2)}$, we need the additional gate complexity even for Case 1. 
Furthermore, for $d=1$, we only need one ancillary qubit for the PITE and another for the potential function. For $d\ge 2$, we need the other $2n+\lceil\log_2 d\rceil$ qubits for constructing the squared distance register. Finally, if we employ a $q$th-order AAPITE (higher order AAPITE operators are discussed in \ref{subsec:F-1}) as well as a $q$th-order Suzuki-Trotter formula \cite{Janke.1992}, then the dependence $\varepsilon^{-1}$ can be reduced to $\varepsilon^{-1/q}$ at the cost that the prefactor has a dependence on $q$. Although the asymptotic order becomes better as $q$ increases, to achieve a given error bound there exists an optimal $q$ due to the trade-off of the prefactor and the order of the error bound. 

To clarify the result, we give several comments as follows. 

\noindent - First, as a well-used setting for the QLSAs, the error bound in Table \ref{sec3:tab1} describes the $\ell^2$-norm of the difference between the normalized ITE operator and the normalized approximate operator based on AAPITE. Besides, we include only the quantum resources for the approximate ITE of the given Hamiltonian, while the resources for the statevector preparation and readout are not counted. Some comments on the statevector preparation/readout are provided in \ref{appC}. 

\noindent - Second, although we only consider the differential operator \eqref{sec3:eq-Ham} corresponding to advection-diffusion-reaction equations in this paper, the AAPITE algorithm can be readily extended to any time evolution PDEs including the convection equations, the (time-dependent) L\'ame equations, etc. The algorithm itself works for the cases with general boundary conditions as long as one finds an efficient implementation of the RTEs for the square root operator $\sqrt{H_N}$. In this paper, we restrict ourselves to the periodic boundary condition and propose also the explicit and efficient quantum circuits for the realization of such RTEs. 

\noindent - Third, in the best-order QLSAs (e.g. \cite{BCOW17, Krovi2023, Berry.2024}), the authors constructed an enlarged matrix to obtain a logarithmic dependence on $\varepsilon^{-1}$ as shown in Table \ref{sec1:tab2}, which is better than our power dependence on $\varepsilon^{-1}$. However, the dependence on $N^{1/d}=2^n$ therein is polynomial, which comes from the condition number of the discretized Laplace operator. Our algorithm achieves an exponential speedup in $N$ by suitably approximating the ITE operator at the cost of a worse dependence on the error bound $\varepsilon$. 
We discuss in more detail in Sect.~\ref{subsec:5-4}. 

\noindent - Fourth, in \cite{Montanaro.2016}, the authors proved that to readout the norm of the solution as well as some physical quantities related to the solution, $O(\varepsilon^{-1})$ repetitions of the quantum circuit are indispensable. Therefore, it is, to some extent, reasonable to consider a quantum algorithm for the calculation part with a power dependence on $\varepsilon^{-1}$ because the total complexity of any end-to-end quantum algorithm eventually has a power dependence regarding $\varepsilon^{-1}$ (see also the second comment in Sect.~\ref{subsec:5-5}).  
In this paper, keeping in mind the application for the eFTQC, we minimize the use of ancillary qubits and mainly focus on the quantum speedup regarding the grid size $N$. 

\section{Numerical examples}
\label{sec:4}

Considering the Hamiltonian $\hat H$ given in Eq.~\eqref{sec3:eq-Ham}, a direct application of the AAPITE operator is to solve the following advection-diffusion-reaction equation:
\begin{align}
\label{sec4:eq-gov}
\partial_t u(\mathbf{x},t) - a\nabla^2 u(\mathbf{x},t) + \mathbf{v}\cdot\nabla u(\mathbf{x},t) + V(\mathbf{x}) u(\mathbf{x},t) = 0, \ \mathbf{x}\in [0,L]^d, \, 0<t\le T,
\end{align}
where $d\in \mathbb{N}$, $L$ and $T$ are given positive constants which denote the spatial dimension, the length of the spatial domain, and the final time, respectively. Moreover, $a\in \mathbb{R}_{>0}$ and $\mathbf{v}\in \mathbb{R}^d$ are given diffusion coefficient and advection coefficient, respectively, while $V$ is a given zeroth-order (potential) term describing the self-reaction. 
In the following context, we consider the periodic boundary condition with the initial condition:
\begin{align}
\label{sec4:eq-init}
u(\mathbf{x},0) = u_0(\mathbf{x}), \quad \mathbf{x}\in [0, L]^d,
\end{align}
where $u_0$ is a sufficiently smooth function. 
In this paper, the numerical results are simulated by Qiskit \cite{Qiskit23}, a quantum gate-based emulator. Although we mention some efficient techniques, e.g. \cite{Moosa.2023, KDNTM23} for the (approximate) preparation of the initial state in \ref{appC}, to avoid discussing the additional errors from encoding the initial condition and obtaining the solution classically, we realize the quantum statevector preparation and readout by the Qiskit functions: \texttt{set\_statevector()} and \texttt{get\_statevector()}, respectively, and the success probability is calculated by the squared $\ell^2$-norm of partial statevector. 

\subsection{1D advection-diffusion equation}
\label{subsec:4-1}

First, we consider a simple example in 1D. 
Let $d=1$, $L=1$, $T=0.1$, $a=0.5$, $\mathbf{v}=5$, and $V = 0$. 
Moreover, we assume $u_0(x) = \sin(\pi x)$, $x\in [0,1]$. 
Actually, in this case, we can solve Eq.~\eqref{sec4:eq-gov} analytically using the Fourier series and obtain
\begin{align}
\label{sec4:eq-AS}
u(x,t) = \frac{2}{\pi} + \sum_{k=1}^\infty -\frac{4}{\pi(4k^2-1)}\mathrm{exp}\left(-2\pi^2 k^2 t\right)\cos(2\pi k(x-5t)), \ x\in [0,1], \ t\in [0,0.1].
\end{align}
As a reference solution, we introduce the truncated solution with an even number $M$:
\begin{align*}
u_M(x,t) := &\frac{2}{\pi} + \sum_{k=1}^{M/2-1} -\frac{4}{\pi(4k^2-1)}\mathrm{exp}\left(-2\pi^2 k^2 t\right)\cos(2\pi k(x-5t)) \\
& - \frac{2}{\pi(M^2-1)}\mathrm{exp}\left(-\pi^2 M^2 t/2\right) \mathrm{exp}\left(-\mathrm{i}M\pi(x-5t)\right), \ x\in [0,1], \ t\in [0,0.1],
\end{align*}
which is the true solution to Eq.~\eqref{sec4:eq-gov} with a truncated initial condition: 
\begin{align*}
u_{0,M}(x) &= \sum_{k=-M/2}^{M/2-1} \int_0^1 u_0(\xi) \mathrm{exp}\left(-\mathrm{i}2\pi k\xi\right)d\xi\, \mathrm{exp}\left(\mathrm{i}2\pi kx\right) \\
&= \sum_{k=-M/2}^{M/2-1} -\frac{2}{\pi(4k^2-1)} \mathrm{exp}\left(\mathrm{i}2\pi kx\right), \quad x\in [0,1]. 
\end{align*}
Although such a truncated solution may be unphysical due to its complex values coming from the third term for any $t>0$, no problem will occur in our discussion below, and it converges to the analytical solution exponentially fast as $M\to \infty$.
Furthermore, we call $u_N(t) := (u_N(p_0,t),\ldots,u_N(p_{N-1},t))^\mathrm{T}\in \mathbb{C}^N$ the $N$-truncated solution for any $t\in [0,T]$ where $p_l$, $l=0,\ldots, N-1$ is the grid point defined by Eq.~\eqref{sec3:eq-x}.   
Let $\tilde u_N(t)\in \mathbb{C}^N$ be the ``$N$-quantum solution" using the AAPITE operator. 
The ``$N$-quantum solution" is a vector of length $N$ approximating the solution to Eq.~\eqref{sec4:eq-gov} at grid points, see \ref{appC} for the detailed derivation. 

\noindent \underline{\bf Visualization of the quantum solution}
We plot the proposed $N$-quantum solution and the analytical solution at simulation time $t=0.01,0.05,0.1$ in Fig.~\ref{sec4:Fig0}. Here, we take the real part of the $M$-truncated solution with large $M=2000$ as the analytical solution, and the classical numerical solutions by the FDM in space and the Euler methods in time are also provided for the comparison of the results. 
\begin{figure}[htb]
\centering
\resizebox{15cm}{!}{
\includegraphics[keepaspectratio]{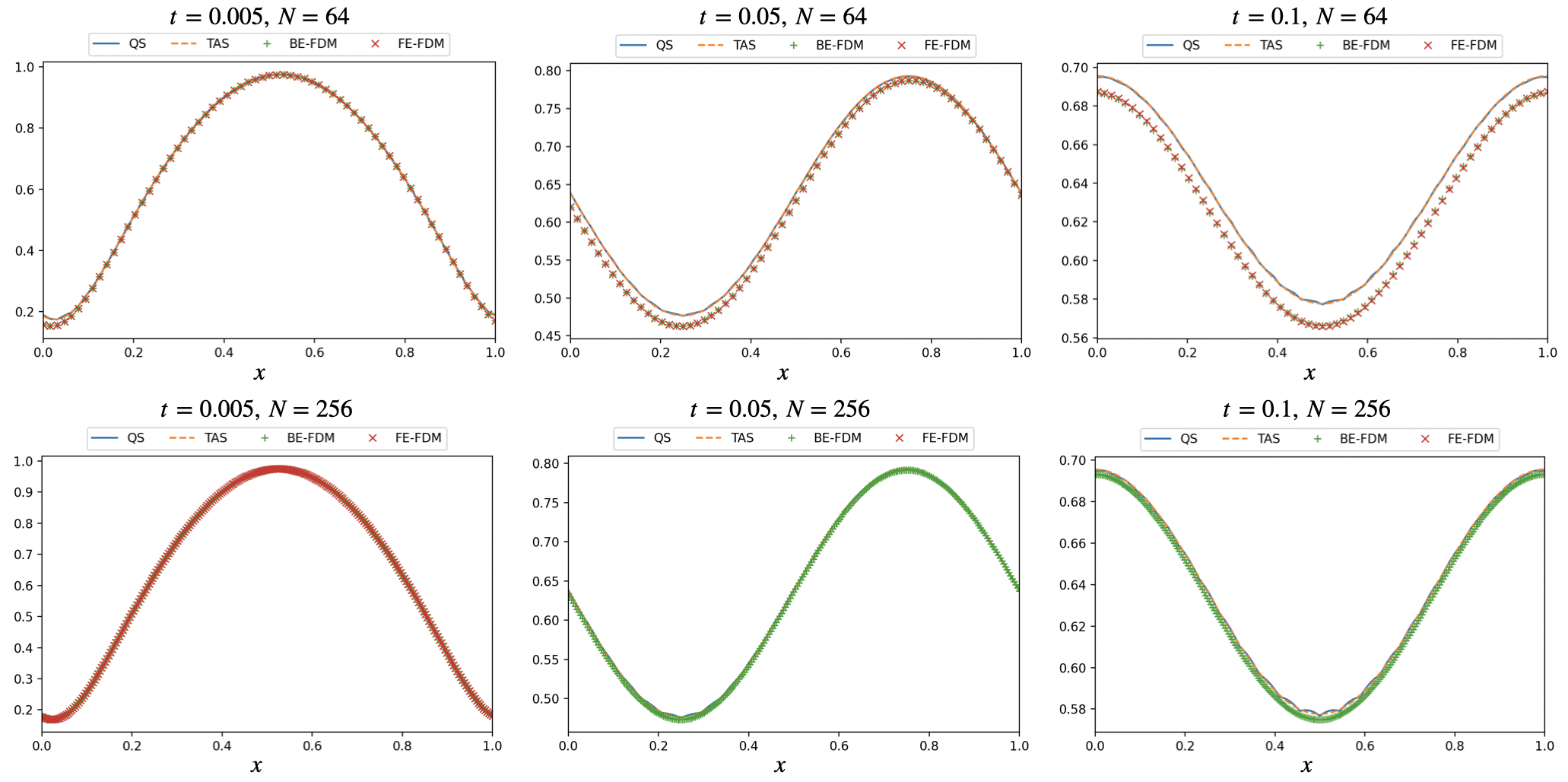}
}
\caption{Plot of solutions at time $t=0.005,0.05,0.1$. The $N$-quantum solution (QS), the (truncated) analytical solution (TAS), the classical solutions by the backward Euler FDM (BE-FDM) and the forward Euler FDM (FE-FDM) are illustrated by line, dashed line, plus marker, and cross marker, respectively. In the first row, we take $N=64$ while we take $N=256$ in the second row. Moreover, we choose the time step $\Delta\tau=0.0005$ for the quantum solution, and the time step $\Delta\tau = t/10^{3}$ for the classical FDM solutions.} 
\label{sec4:Fig0}
\end{figure}
We find that the proposed quantum solution approaches the analytical solution much better than the classical FDM solutions, especially for a small grid parameter. 
The reason is that our quantum solution is based on a special FSM whose $N$ describes the number of bases, and thus, it has a better precision than the FDM methods with the same grid parameter $N$. 
Besides, the quantum solutions in Fig.~\ref{sec4:Fig0} exhibit arc-like shapes, which owes to the approximation error (errors of relatively high frequency parts contribute the most since the Fourier coefficients with large indices are small), and can be relieved by taking smaller $\Delta\tau$. 
Here, we plot the classical FDM solutions to indicate the best possible performance of the QLSAs based on the backward Euler FDM (e.g. \cite{Ingelmann.2024}). It is clear from Fig.~\ref{sec4:Fig0} that for the PDEs with periodic boundary conditions, it is better to apply a quantum algorithm based on FSM, and we compare AAPITE algorithm with the ones in \cite{Ingelmann.2024} in Sect.~\ref{subsec:5-2}.  

\noindent \underline{\bf Error of the AAPITE} 
In this example, there is no Suzuki-Trotter error because $V\equiv 0$. Thus, we illustrate the approximation error of the AAPITE by plotting the $\ell^2$-error between the $N$-quantum solution and the $N$-truncated solution in Fig.~\ref{sec4:Fig1}. The error decreases, and the final success probability converges to some constant around $0.814$ as the time step $\Delta\tau$ becomes smaller. The theoretical linear dependence on $\Delta\tau$ is confirmed numerically in \ref{subsec:4-3-3}.
\begin{figure}[htb]
\centering
\resizebox{8cm}{!}{
\includegraphics[keepaspectratio]{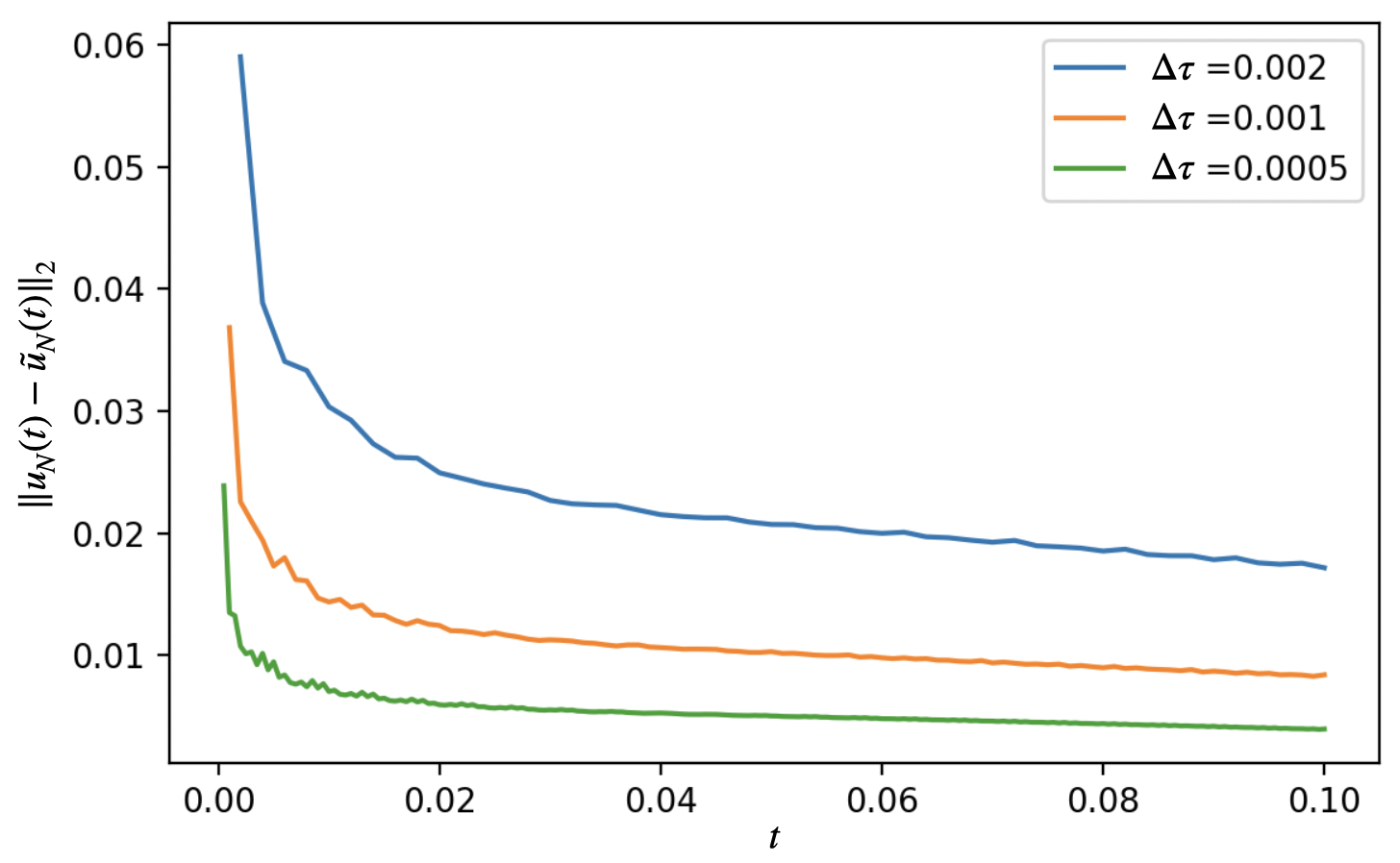}
}
\caption{The $\ell^2$-error between the $N$-quantum solution using the AAPITE operator and the $N$-truncated solution as a time-varying function. Here, $N=64$ and three different time steps of $\Delta\tau = 0.002, 0.001, 0.0005$ are considered in the implementations of the AAPITE operator. The final success probabilities of the AAPITE circuit are $0.81387$, $0.81395$, and $0.81400$, respectively.}
\label{sec4:Fig1}
\end{figure}

\noindent \underline{\bf Numerical dependence on diffusion/advection coefficient}
We consider the numerical dependence on the diffusion and advection coefficients $a$ and $\mathbf{v}$. The $\ell^2$-errors between the $N$-quantum solutions and the $N$-truncated solutions are shown in Fig.~\ref{sec4:Fig1b}. 
This indicates that the error increases as the diffusion coefficient $a$ becomes larger, but it does not depend on the advection coefficient $\mathbf{v}$. The dependence on $a$ can be understood by the following change of variables: 
\begin{align*}
\partial_{\tilde t} u(\mathbf{x},\tilde t) - \nabla^2 u(\mathbf{x},\tilde t) + \mathbf{\tilde v}\cdot\nabla u(\mathbf{x},\tilde t) = 0, \quad \mathbf{x}\in [0,L]^d, \, \tilde t>0,
\end{align*}
where we take $\tilde t = at$ to make the diffusion coefficient $1$ and $\mathbf{\tilde v}=a^{-1}\mathbf{v}$. Then, keeping $K$ invariant, we have $\widetilde{\Delta\tau} = a\Delta\tau$. This yields a linear dependence of the diffusion coefficient $a$ since the error overhead depends linearly on $\widetilde{\Delta\tau}$ (Sect.~\ref{subsec:3-2}). 
As a result, we find that the $\ell^2$-error depends only on $a$, but is independent of $\mathbf{v}$. 
It is well-known that the numerical solution is not stable for large advection coefficients for the conventional FDMs. Here, our algorithm is stable with respect to $\mathbf{v}$ because we use the FSM.  
\begin{figure}[ht]
\centering
\resizebox{15cm}{!}{
\includegraphics[keepaspectratio]{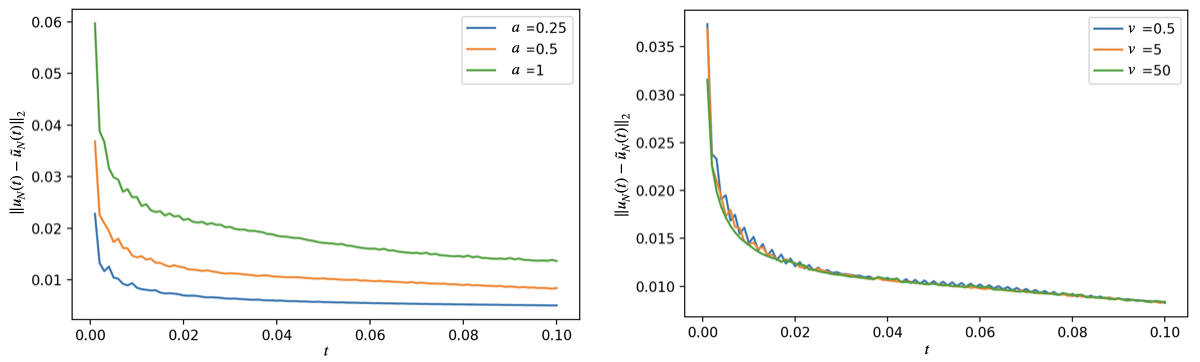}
}
\caption{The $\ell^2$-error between the $N$-quantum solution using the AAPITE operator and the $N$-truncated solution as a time-varying function. Here, $N=64$ and $\Delta\tau=0.001$ are fixed. 
In the left subplot, three different diffusion coefficients $a = 0.25, 0.5, 1$ are considered with fixed $\mathbf{v}=5$, while in the right subplot, three different advection coefficients $\mathbf{v} = 0.5, 5, 50$ are considered with fixed $a=0.5$. }
\label{sec4:Fig1b}
\end{figure}

\subsection{2D advection-diffusion-reaction equation}
\label{subsec:4-2}

Next, we consider a more involved example where we put an absorption potential at the center of the domain. More precisely, we let $d=2$, $L=2\pi$, $T=0.1$, $a=0.5$, $\mathbf{v}=(20, 0)^\mathrm{T}$, and 
$$
V(x_1,x_2) = 
\left\{
\begin{aligned}
& 10, && |x_1-\pi|\le \pi/2 \text{ and } |x_2-\pi|\le \pi/2,\\
& 0, && \text{ otherwise. }
\end{aligned}
\right.
$$
Moreover, we consider the initial condition of the Gaussian function:
$$
u_0(\mathbf{x}) = A\, \mathrm{exp}\left(-\frac{|\mathbf{x}-\mathbf{x}_0|^2}{2\sigma^2}\right),
$$
with $\mathbf{x}_0=(\pi/2,\pi/2)^\mathrm{T}$, $\sigma=0.5$ and $A=1$. 

\noindent \underline{\bf Visualization of the quantum solution}
We consider the $(2^4,2^7)$-quantum solutions, that is, the solutions on a $2^7\times 2^7$ grid calculated by the Fourier coefficients on a $2^4\times 2^4$ grid, see \ref{appC} for the precise definition. The solutions at several fixed time are illustrated by the colormap plots in Fig.~\ref{sec4:Fig2}. 
The initial source moves in $x_1$-direction as expected, and the solution at the center part of the domain decreases rapidly owing to the center absorption potential. 
\begin{figure}[htb]
\centering
\resizebox{14cm}{!}{
\includegraphics[keepaspectratio]{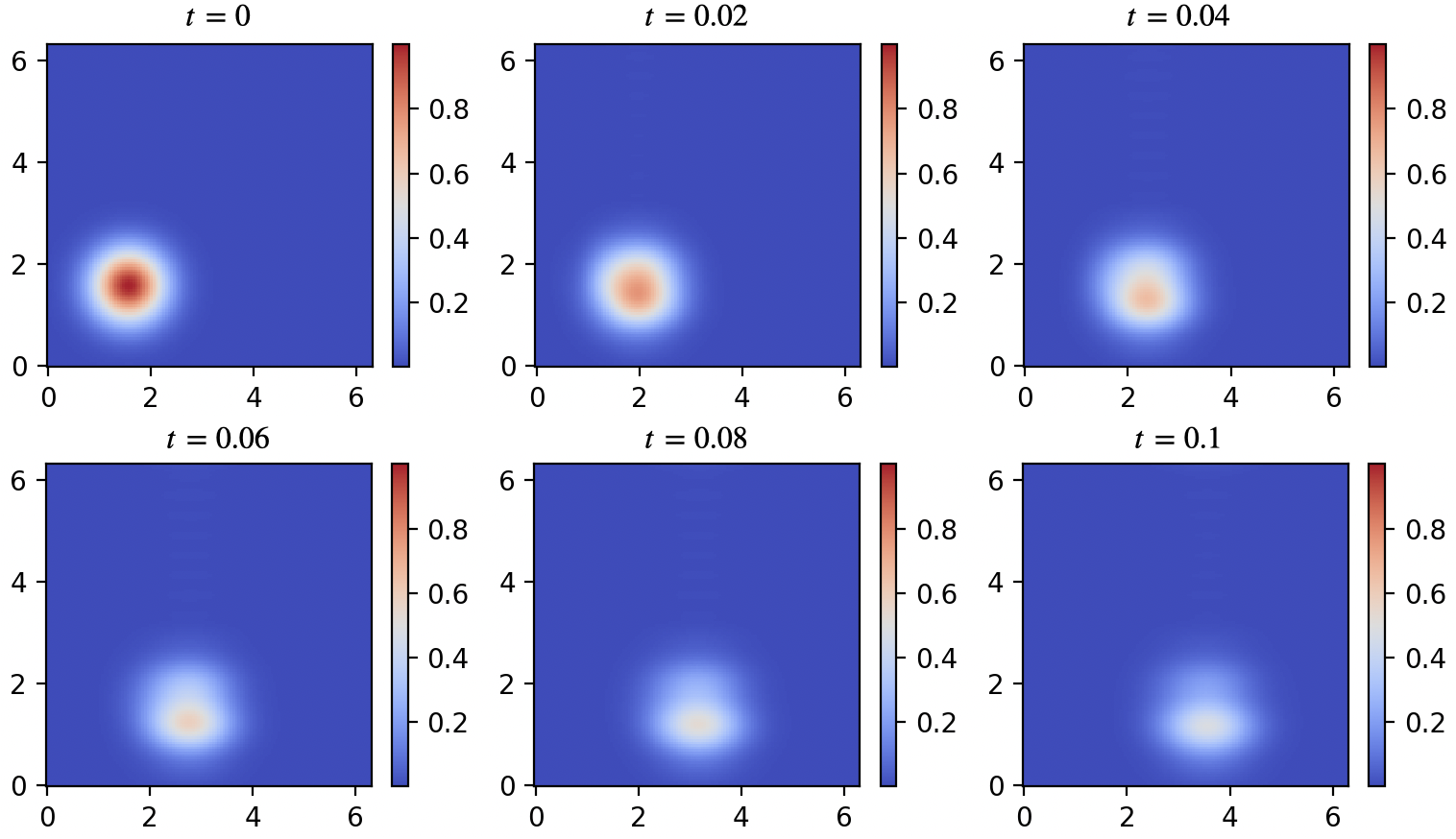}
}
\caption{Time evolution of the $(2^4,2^7)$-quantum solution using the AAPITE operator. The quantum solutions at $t=0, 0.02, 0.04, 0.06, 0.08, 0.1$ are visualized using the colormap plots. In each subplot, the $x$ and $y$ axes describe the location in $x_1$-direction and $x_2$-direction, respectively. The result is simulated by Qiskit, and we take $\Delta\tau = 0.0025$ with a second-order Suzuki-Trotter formula. }
\label{sec4:Fig2}
\end{figure}

\noindent \underline{\bf Error of the Suzuki-Trotter formula and the AAPITE}
With the presence of the spatially varying potential $V$, there is no analytical solution formula. To demonstrate the performance of the AAPITE, we choose 
\begin{align*}
\psi(t) := \mathrm{exp}\left(-t H_N\right)\psi_{\text{init}}, \quad \psi(t) \equiv \sum_{\mathbf{l}\in [\tilde N_0]^2}\psi_{\mathbf{l}}(t)\ket{\mathbf{l}}, 
\end{align*}
and 
$$
u_{N_0,N_f}(t) := \sum_{\mathbf{\tilde l}\in [\tilde N_0]^2} \psi_{\mathbf{\tilde l}}(t)g_N(\mathbf{p}_{\mathbf{l}};\mathbf{p}_{\mathbf{\tilde l}}), \quad \mathbf{l}\in [\tilde N_f]^2,
$$
as the reference solution, where $N_0=2^4$, $N_f=2^7$, $g_N$ is the two-dimensional pixel function defined in Eq.~\eqref{appB:eq10} in \ref{subsec:appB-2}, and $\psi_{\text{init}}$ is the initial state given by
\begin{align*}
\psi_{\text{init}} = \sum_{\mathbf{l}\in [\tilde N_0]^2} 
\left(\frac{1}{L N_0}\sum_{\mathbf{k}\in [\tilde N_0]^2} \left(\int_{[0,L]^2} u_0(\mathbf{x}) \mathrm{exp}\left(-\mathrm{i}\frac{2\pi}{L} \left(\mathbf{k}-\frac{N_0}{2}\right)\cdot\mathbf{x}\right) \mathrm{d}\mathbf{x}\right) \mathrm{exp}\left(\mathrm{i}\frac{2\pi}{N_0}\left(\mathbf{k}-\frac{N_0}{2}\right)\cdot\mathbf{l}\right)\right)\ket{\mathbf{l}}.
\end{align*}
This reference solution describes the exact solution to the ordinary differential system after the discretization in space and converges to the true solution rapidly as $N=N_0^2\to \infty$. Fig.~\ref{sec4:Fig3} shows the $\ell^2$-error between the $(2^4, 2^7)$-quantum solutions and the reference solution, which implies that the quantum solution converges to the reference solution, and the final probability success converges to a constant around $0.262$ as $\Delta\tau\to 0$.
\begin{figure}[htb]
\centering
\resizebox{8cm}{!}{
\includegraphics[keepaspectratio]{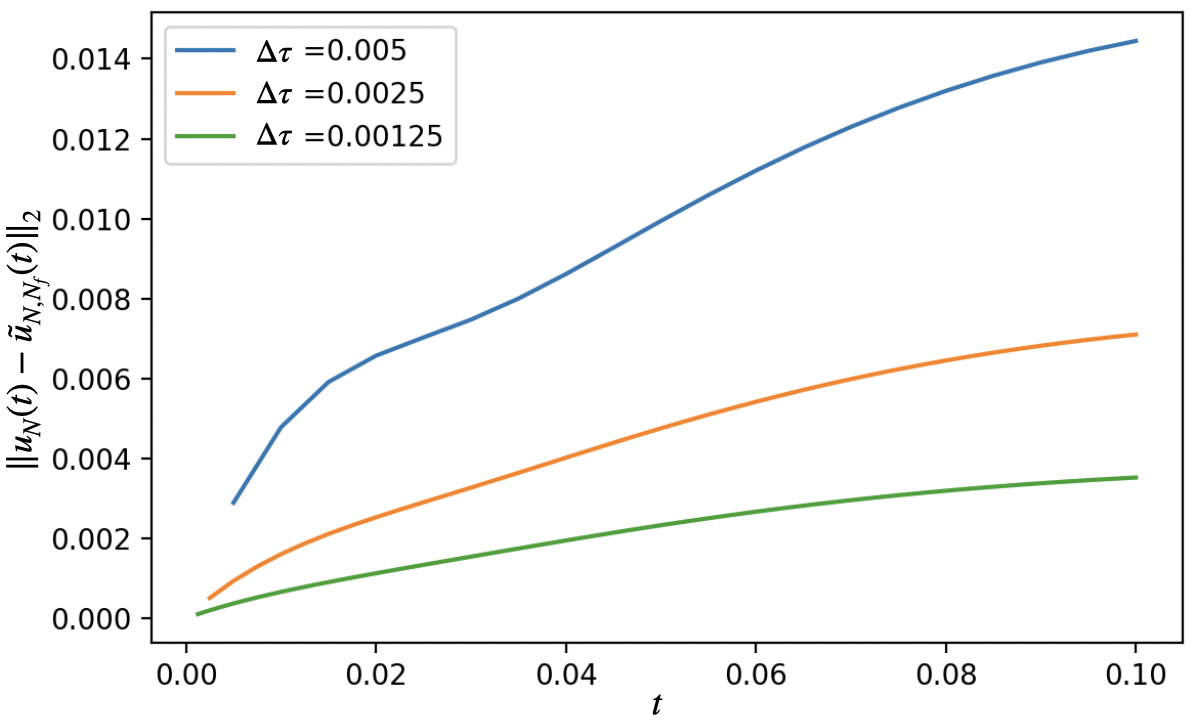}
}
\caption{The $\ell^2$-error between the $(2^4,2^7)$-quantum solution using the AAPITE operator and the reference solution as a time-varying function. Here, three different time steps $\Delta\tau = 0.005, 0.0025, 0.00125$ are considered in the AAPITE operator, and a second-order Suzuki-Trotter formula is employed. The final success probabilities are $0.26172$, $0.26187$, $0.26194$, respectively. 
}
\label{sec4:Fig3}
\end{figure}
This time, the error comes from both the Suzuki-Trotter formula and the approximation of the AAPITE operator. 
We can confirm that the approximation error depends on the time step linearly, while the Suzuki-Trotter error has the dependence on the time step according to its splitting order. The numerical details on each error regarding the time step are provided in \ref{subsec:4-3-3}, \ref{subsec:4-3-4}.  
Here, we observe that different from the decreasing errors (regarding simulation time $t$) in the simple case of vanishing potential (see Fig.~\ref{sec4:Fig1}), the errors in Fig.~\ref{sec4:Fig3} increase as $t$ becomes larger. This is consistent with our theoretical overhead of $\tilde{E}_{\text{tot}}$ in Sect.~\ref{subsec:3-2} that the error may increase for a long-time simulation when the potential function is nonzero. On the other hand, Fig.~\ref{sec4:Fig3} also implies that the linear dependence on $T$ in the theoretical overhead \eqref{sec3:eq-err} is not sharp.

\section{Comparisons and discussion}
\label{sec:5}

In this section, we compare our algorithm based on the AAPITE operator with some existing quantum/hybrid algorithms. 

\subsection{HHL algorithm and a VQA}
\label{subsec:5-2}

First, we compare our quantum algorithm with a previous work for a one-dimensional advection-diffusion equation \cite{Ingelmann.2024}, where the authors used two quantum algorithms. One of them is a QLSA based on the HHL algorithm and the other is a VQA. 
According to Fig.~6 in \cite{Ingelmann.2024}, the simulation results of both the QLSA and the VQA are comparable in the mean squared error (MSE). 

Here, we check the MSE of our quantum solution under the setting of \cite{Ingelmann.2024}. 
Let $d=1$, $L=1$, $a=1$, $\mathbf{v}=10$, $T=0.04$, and $V=0$. Moreover, the initial condition is set to be the shifted Dirac delta function $\delta(x-L/2)$ which can be rewritten in the following Fourier series:
$$
u_0(x) = \frac{1}{L} + \sum_{k=1}^\infty \frac{2}{L}\cos(k\pi)\cos\left(\frac{2k\pi}{L}x\right), \ x\in [0,L].
$$
Then, we define the truncated analytical solution as follows: 
\begin{align*}
u_{N_{\text{trun}}}(x,t) := \frac{1}{L} + \sum_{k=1}^{N_{\text{trun}}} \frac{2}{L}\cos(k\pi)\cos\left(\frac{2k\pi}{L}(x-vt)\right)\mathrm{exp}\left(-a\left(\frac{2k\pi}{L}\right)^2 t\right), \ x\in [0,L], \ t>0,
\end{align*}
with a parameter $N_{\text{trun}}\in \mathbb{N}$. We take sufficiently large $N_{\text{trun}}=1000$ above and regard it as the reference solution, which is almost an analytical solution because of the exponentially decay of the high frequency (i.e., large $k$) parts. 

First, we plot the proposed quantum solutions with different parameters at two time points $t=0.005$ and $t=0.04$ in Fig.~\ref{appD:Fig3}. 
\begin{figure}[htb]
\centering
\resizebox{15cm}{!}{
\includegraphics[keepaspectratio]{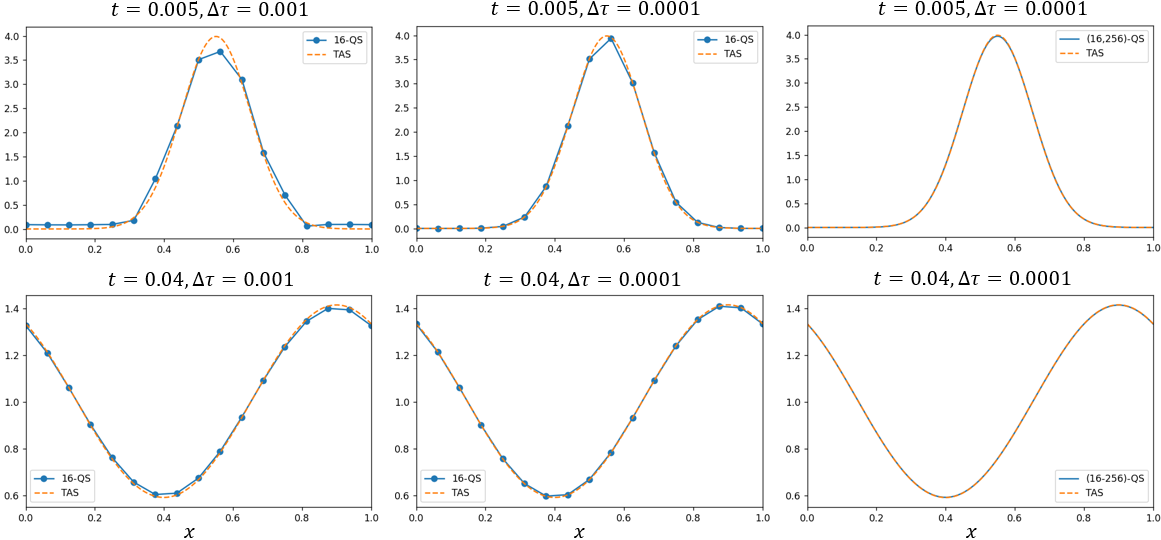}
}
\caption{Proposed quantum solutions with different time steps $\Delta\tau=0.001,0.0001$ at $t=0.005,0.04$. The quantum solutions at $t=0.005$ are shown in the first row, while the solutions at $t=0.04$ are shown in the second row. The first column gives the $2^4$-quantum solutions with a relatively large time step $\Delta\tau=0.001$, the second column gives the $2^4$-quantum solutions with a smaller $\Delta\tau=0.0001$, and the third column gives the post-processed $(2^4,2^8)$-quantum solutions with $\Delta\tau=0.0001$ for comparison. In each subplot, the reference solution (TAS) is plotted in a dashed line. }
\label{appD:Fig3}
\end{figure}
All the quantum solutions fit the analytical solution well. 
Focusing on the first two columns, we find that if a smaller time step $\Delta\tau$ is chosen, then the quantum solution approximates the analytical solution even better. The third column in Fig.~\ref{appD:Fig3} gives the $(2^4,2^8)$-quantum solutions, which indicates that our quantum solution using only $4$ qubits (one ancillary qubit is not counted) provides good approximation to the analytical solution on a $2^8$ grid. 

Next, we provide a quantitative analysis by plotting the MSEs between the quantum solutions and the reference solution. Again, we choose two different time steps $\Delta\tau=0.001$ and $\Delta\tau=0.0001$. The results are demonstrated in Fig.~\ref{appD:Fig4}.
\begin{figure}[htb]
\centering
\resizebox{15cm}{!}{
\includegraphics[keepaspectratio]{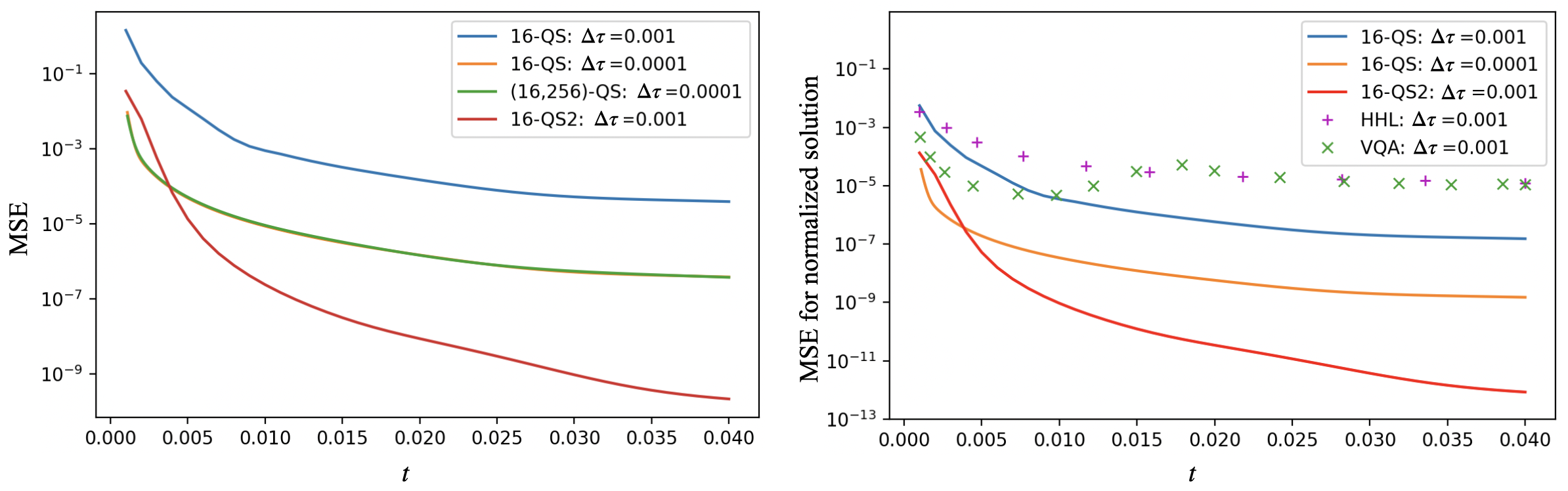}
}
\caption{MSEs between the quantum solutions and the reference solution. The results by several quantum solutions are plotted in different colors, while $4+1$ qubits are used to derive the quantum solutions in all cases. Here, $16$-QS2 indicates a second-order alternative approximate PITE, which implements the operation in Eq.~\eqref{appD:eq-2oPITE}. 
The left subplot describes the MSEs between the quantum solutions and the reference solution. The right one describes the MSEs between the normalized quantum solutions and the normalized reference solution using our proposed quantum solution and the HHL algorithm and the VQA in \cite{Ingelmann.2024}. The data of the HHL algorithm and the VQA are extracted from Fig.~6 in \cite{Ingelmann.2024} by the application PlotDigitizer \cite{PlotD}.}
\label{appD:Fig4}
\end{figure}
The MSEs are relatively large near the initial time, and they decrease as $t$ becomes larger. 
Comparing the blue line with the orange one, we observe that the MSE has a quadratic dependence on the time step $\Delta\tau$, which can be seen by the discussion in the next paragraph. Note that the orange line almost coincides with the green line in the left subplot of Fig.~\ref{appD:Fig4}. This implies that the post-processing step in \ref{subsec:C-3} enables us to obtain the solutions in a finer grid while keeping almost the same MSE. 
To compare with the results in \cite{Ingelmann.2024}, we provide also the MSEs between the normalized quantum solutions and the normalized reference solution. Here, we choose the same normalized solutions as those in \cite{Ingelmann.2024}, that is, we divide the quantum/reference(analytical) solution by the grid parameter $N$ so that the $\ell^2$-norm of the initial state is normalized to $1$. By choosing the same $\Delta\tau=0.001$, we have a comparable result to Fig.~6 in \cite{Ingelmann.2024} (see the blue line in the right subplot in Fig.~\ref{appD:Fig4}). More precisely, we have similar MSE as the HHL and one-order worse MSE than the VQA near $t=0$, while we obtain two-order better MSE than both algorithms near $t=0.04$. Besides, we can achieve further better precision by choosing small $\Delta\tau$ for fixed $N=16$ (see the orange line), and obtain a precision of $10^{-9}$ at $t=0.04$, which is a huge improvement than the HHL algorithm and the VQA in \cite{Ingelmann.2024}. 
To avoid misleading, we emphasize that here we use the MSE for the normalized solution only to compare with the result in \cite{Ingelmann.2024}. This error is not physical because it gives much smaller values than it performs. One should use either the $\ell^2$-error for the normalized solution or the MSE for the original solution (i.e., the left subplot in Fig.~\ref{appD:Fig4}). 
 
The reason why our algorithm has better performance in precision is as follows. 
Concerning the well-known theoretical error estimate of the Euler FDMs, the $\ell^2$-error scales as $O(\Delta\tau)+O(N^{-2})$, which means the MSE is $O(N^{-1}(\Delta\tau)^2)+O(N^{-5})$. 
On the other hand, by the theoretical estimates in Sect.~\ref{subsec:3-2} for both the approximation error and the Suzuki-Trotter error, as well as the theoretical bound of the discretization error (see Lemma 1 in \cite{Childs.2022}), the $\ell^2$-error for the proposed algorithm is $O(\Delta\tau)+O((N/2)^{-N/2})$. This implies that the MSE is $O(N^{-1}(\Delta\tau)^2)+O(N^{-1}(N/2)^{-N})$, whose second part has a more rapid decay in $N$. These theoretical estimates explain the essential superiority (regarding $N$) of our algorithm over the ones in \cite{Ingelmann.2024} based on the Euler FDMs. Numerically, we also refer to the third column in Fig.~\ref{sec4:Fig0} that one needs to take a larger $N$ for the Euler FDMs.  

Moreover, we can further consider the second-order AAPITE operator defined in Eq.~\eqref{appD:eq-2oPITE} in \ref{subsec:F-1}. 
The quantum circuit for the second-order AAPITE operator can be similarly implemented as the proposed (first-order) AAPITE operator because the kinetic part is diagonal in the Fourier domain and the potential part is diagonal. The only difference is that we need to implement a third-degree polynomial function, which can be precisely realized by the polynomial phase gates described in \cite{Kosugi.2023, Huang.2024p} using $O(\mathrm{polylog}\, N)$ quantum gates with depth $O(\mathrm{polylog}\, N)$. 
The MSE for the $16$-quantum solution using the second-order AAPITE is denoted by the red line in Fig.~\ref{appD:Fig4}. Owing to the higher-order dependence on $\Delta\tau$, the method based on the second-order AAPITE greatly outperforms the previous work. 
Combining such a high-order AAPITE operator (see \ref{subsec:F-1}) with a high-order Suzuki-Trotter formula, we can derive a quantum solution with $\ell^2$-error of $O((\Delta\tau)^p)+O((N/2)^{-N/2})$ for any $p\ge 2$ at the cost of larger gate count/circuit depth (but still of the order $O(\mathrm{polylog}\, N)$). 

Finally, we discuss the required computational resources. 
For the total number of qubits, we use only $1+\log_2 N$, that is, $\log_2 N$ qubits for the spatial discretization and another ancillary qubit for the AAPITE (we do not need more qubits because $d=1$ and $V=0$). This is the same as that for the VQA and smaller than that for the HHL algorithm, that is, $1+n_q+\log_2 N$ where $n_q=2\log_2 N +O(1)$ (\cite{HHL09}) is the number of ancillary qubits for the QPE. 
As for the execution time, since our algorithm needs much less classical computations, the total execution time for simulating the above example (by Qiskit with $400$ time steps) requires less than $3$ minutes on a laptop (MacBook Pro, Z16R0004VJ/A, FY2022, Apple M2, 8 Cores, 16GB), much faster than the VQA (a half hour to several hours, see Table C.1 in \cite{Ingelmann.2024}). The circuit depth of our algorithm in each time step is $O(\log N)$, which is much smaller than that for the HHL algorithm because the controlled rotation part needs $O(N^2)$ single-qubit/two-qubit operations. 
Therefore, we conclude that our algorithm outperforms the ones in \cite{Ingelmann.2024} (the classical/quantum computational cost is much smaller, and we obtain a higher precision at the final time). 
 
Although the quantum algorithm based on the FSM is not flexible in boundary conditions compared to the other discretization methods including the FDMs, the finite element methods (FEMs), etc., it is compatible with the first-quantized Hamiltonian simulation so that the corresponding quantum circuits for solving the PDEs can be efficiently and explicitly constructed. Thus, we suggest the FSM for the spatial discretization if the boundary condition is periodic or is open (i.e., it does not influence the solution in the interested domain).

\subsection{Previous approximate PITE}
\label{subsec:5-3}

Second, we compare the AAPITE with the original approximate PITE in \cite{Kosugi.2022} and its derivative in \cite{Nishi.2023} (variable-time-step approximate PITE or simply VS-APITE) that takes varying time steps $\{\Delta\tau_j\}_{j=1}^K$, using a more involved case of $V\not= 0$. 
Let $d=1$, $L=1$, $T=0.1$, $a=0.5$, $\mathbf{v}=5$, and $u_0(x)=\sin(\pi x)$, $x\in [0,1]$. Moreover, we consider the piecewise constant potential function:
\begin{align}
\label{sec4:eq-V}
V(x)=
\left\{
\begin{aligned}
&10, && |x-1/2|\le 1/4, \\
&0, && \text{otherwise},
\end{aligned}
\right.
\end{align}
which has large absorption around the center of the domain. 
In such a case with a nonzero potential, there is no analytical formula for the solution in general. 
We first check the classical solutions by the backward Euler FDM with different grid parameters $N=64,256,1024,4096$ to find that the solutions have small changes after $N\ge 1024$, and hence, we take the one with $N=4096$ as the reference solution. Then, we demonstrate the quantum solutions using the original approximate PITE (APITE) \cite{Kosugi.2022}, the variable-time-step approximate PITE \cite{Nishi.2023}, and the AAPITE under the grid $N=64$ in Fig.~\ref{sec5:Fig1}. 
\begin{figure}[tb]
\centering
\resizebox{15cm}{!}{
\includegraphics[keepaspectratio]{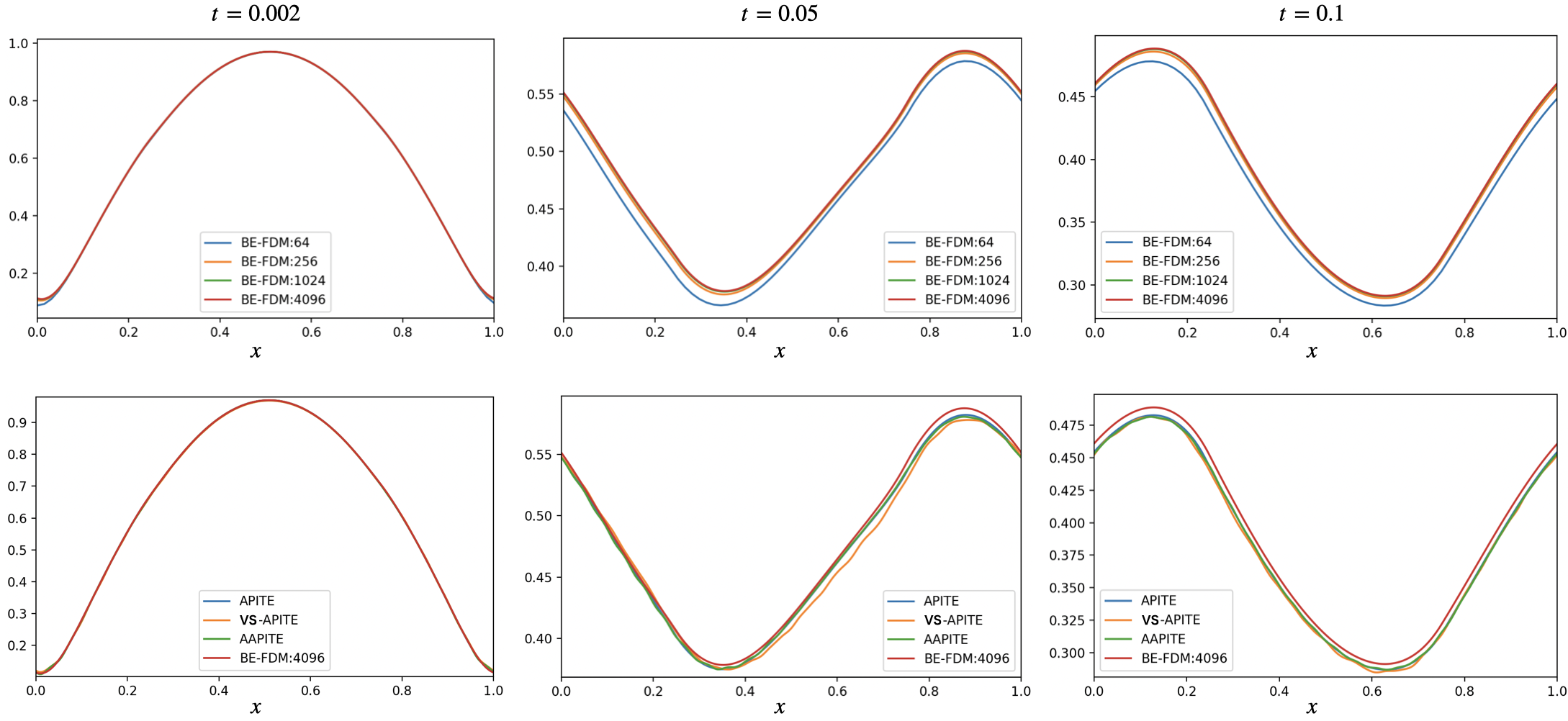}
}
\caption{Plot of solutions at time $t=0.002,0.05,0.1$. In the first row, we plot the classical solutions by the backward Euler FDM with different grid parameters $N=64,256,1024,4096$. 
In the second row, we plot the $(2^6,2^8)$-quantum solutions using the approximate PITE \cite{Kosugi.2022} (APITE), the variable-time-step approximate PITE \cite{Nishi.2023} (VS-APITE), and the alternative approximate PITE (AAPITE). A reference solution by the backward Euler FDM with large parameter $N=4096$ is also given in red. 
Here, the first-order Suzuki-Trotter formula is employed for the quantum solutions, and we choose the time step $\Delta\tau=0.00005$ for APITE and $\Delta\tau=0.0005$ for AAPITE, the linear-scheduled time steps $\Delta\tau_j\in [0.0001,0.0009]$ for VS-APITE, and the time step $\Delta\tau = t/10^{3}$ for the classical FDM solutions. Moreover, a relatively large parameter $m_0=0.9$ is chosen for both APITE and VS-APITE. }
\label{sec5:Fig1}
\end{figure}
From the plots of the quantum solutions, we find that all the quantum solutions approximate the reference solution. 
The APITE solution provides a good approximation for small $\Delta\tau = 0.00005$, but we confirmed the numerical instability as we take a slightly larger $\Delta\tau = 0.0001$ since $\Delta\tau$ should be smaller than $\pi/(2s_0\lambda_N)$ by the estimate in \ref{subsec:A-1}. On the contrary, the quantum solution with a linear-scheduled VS-APITE \cite{Nishi.2023} (whose time step varies linearly from $\Delta\tau_{\text{min}}=0.0001$ to $\Delta\tau_{\text{max}}=0.0009$) demonstrates a normal approximation, and the AAPITE solution has already approached the reference solution with a relatively large $\Delta\tau=0.0005$. 
The gap between the quantum solutions and the reference solution comes from the discretization error and it shrinks as we take larger $N$. 
On the other hand, the success probability of the previous approximate PITE algorithms is extremely small ($\le 10^{-17}$ for VS-APITE and $\le 10^{-183}$ for APITE at $t=0.1$) compared to the proposed AAPITE $(\approx 10^{-1})$, which is demonstrated in Fig.~\ref{sec5:Fig2}. 
\begin{figure}[htb]
\centering
\resizebox{9cm}{!}{
\includegraphics[keepaspectratio]{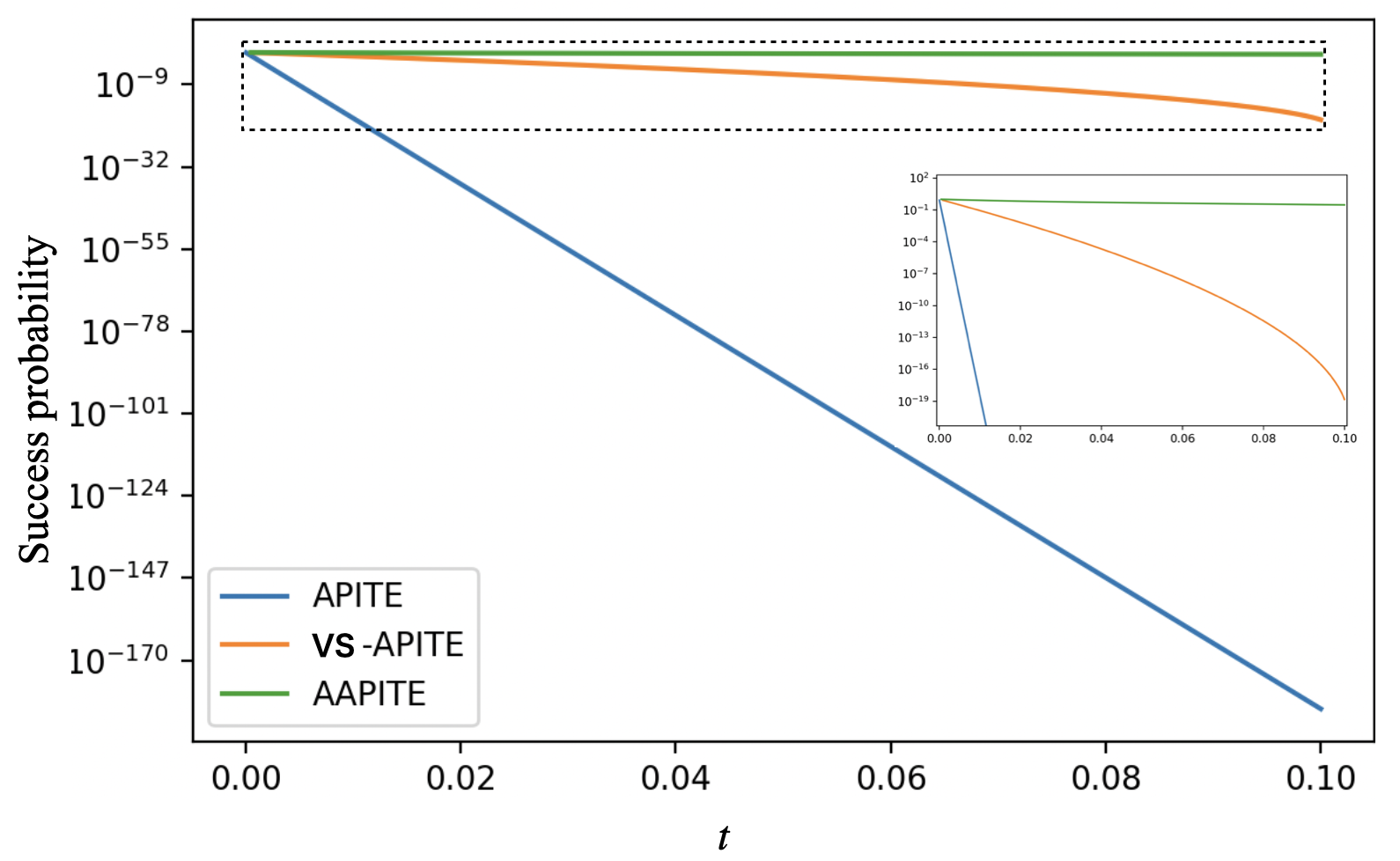}
}
\caption{Success probabilities regarding simulation time $t$ for the quantum solutions by the original approximate PITE (APITE), the variable-time-step approximate PITE (VS-APITE) and the alternative approximate PITE (AAPITE). A zoom-in plot is given to show that the success probability is larger than $10^{-1}$ ($\approx 0.29889$) for AAPITE. }
\label{sec5:Fig2}
\end{figure}

To provide a more quantitative comparison, we choose an exact PITE solution in a similar way as that in Sect.~\ref{subsec:4-2} as the reference solution (so that the discretization error is excluded), and we compare the $\ell^2$-errors between the quantum solutions by these three approximate PITE algorithms and the reference solution under several different choices of $\Delta\tau$ and $m_0$ in Fig.~\ref{sec5:Fig3}. For APITE, we find the $\ell^2$-error becomes larger for larger $m_0$, and we check that the quantum solution is unavailable due to the numerical instability if $m_0\ge 0.92$ and $\Delta\tau\ge 0.0001$. For VS-APITE, we can choose $m_0=0.99$ and the minimum time step $\Delta\tau_{\text{min}}\ge 0.0001$, but larger $m_0$ or larger maximum time step $\Delta\tau_{\text{max}}$ leads to larger $\ell^2$-error. For AAPITE, we find that the quantum solution gives a precision of $10^{-2}$ even when we take $\Delta\tau = 0.0005$, and the $\ell^2$-error further decreases for smaller $\Delta\tau$. In the right subplot of Fig.~\ref{sec5:Fig3}, we confirm numerically the crucial advantage of the AAPITE over the previous approximate PITEs. The success probability for the AAPITE algorithm almost remains constant provided that $\Delta\tau$ is sufficiently small, and this corresponds to the theoretical estimate in Sect.~\ref{subsec:3-2}. Therefore, we conclude the superiority of the AAPITE compared to the previous ones as a quantum solver for the PDEs. 
\begin{figure}[htb]
\centering
\resizebox{15cm}{!}{
\includegraphics[keepaspectratio]{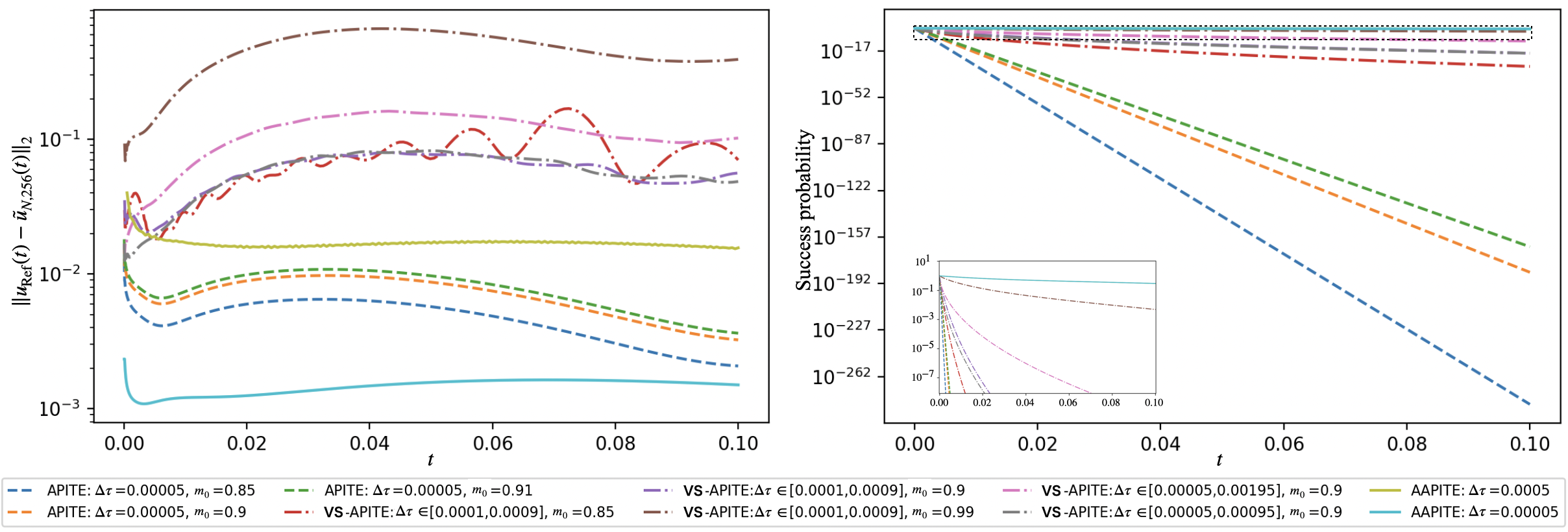}
}
\caption{$\ell^2$-errors between the quantum solutions and the reference solution and corresponding success probabilities under several different choices of $\Delta\tau$ and $m_0$. }
\label{sec5:Fig3}
\end{figure}

\subsection{A best-order QLSA}
\label{subsec:5-4}

Many previous papers apply the QSVT-based method \cite{Gilyen.2019} to solve the quantum linear system problem (QLSP) to obtain the current best asymptotic complexity $O(\kappa\mathrm{polylog}\, N \log(1/\varepsilon))$, e.g. \cite{Lin.2020, Costa.2022, Krovi2023}. Here, $N$ is the matrix size, $\kappa$ is the condition number, and $\varepsilon$ is the error bound of the $\ell^2$-error for the normalized solution. For the time evolution PDEs discussed in this paper, we first discretize the spatial differential operator by some discretization methods to obtain an ordinary differential linear system:
\begin{equation}
\label{appF:eq-gov}
\partial_t \mathbf{v}(t) = -H_N \mathbf{v}(t), \quad t>0, 
\end{equation}
for example, $H_N$ is the discretized matrix in \eqref{sec3:eq-HN}. Without loss of generality, we assume that the minimal real part of the eigenvalues is nonnegative. Then, there are usually two ways to transform the problem into a QLSP. 
The first one is the time-sequential strategy using e.g. the backward Euler (BE) method. The target linear system is
$$
A\mathbf{\tilde{v}}_j := (I+\Delta\tau H_N)\mathbf{\tilde{v}}_j = \mathbf{\tilde{v}}_{j-1},
$$
for each $j=1,2,\ldots$, and we solve the linear system for each time step. Here, $\Delta\tau>0$ is the time step, and $\mathbf{\tilde{v}}_j\approx \mathbf{v}(j\Delta\tau)$ denotes the approximate solution after each time step. 
The other one is the one-shot strategy (e.g. \cite{BCOW17, Krovi2023}), in which all the time steps are incorporated into an enlarged matrix $\tilde A$ by using suitable truncated Taylor expansion, and we only solve the enlarged system once. 
On classical computers, usually the time-sequential strategy is used because it is not efficient to directly deal with extremely large matrices, while the one-shot strategy seems to improve the precision at the cost of using more quantum qubits on quantum computers. 

Next, we discuss the theoretical complexity of the QSVT-based QLSA. Recall that in the QVST-based method, one assumes a basic oracle of the block encoding of $A/s$ where $s$ is a scaling constant depending on the matrix $A$ itself and the employed encoding method. Since the quantum gates are unitary, there is a theoretical lower bound of $s\ge \|A\|$. The conventional way in the previous works is to assume the gate implementation of the oracle and discuss the query complexity (number of queries to the oracle). The best order is given by \cite{Costa.2022}, in which the authors obtained the optimal query complexity $O(\kappa_A\log (1/\varepsilon))$. Here and henceforth, $\kappa_A$ denotes the condition number of the matrix $A$, and $\kappa_{\tilde A}$ is the condition number of $\tilde A$. 

For the time-sequential strategy, we should take a similar time step as the AAPITE algorithm $\Delta\tau = O(\varepsilon)$ (see \ref{subsec:F-2}) such that the error from the time discretization is also controlled by the error bound $\varepsilon$. 
Let $\lambda_{\text{max}}$ denote the maximal eigenvalue (or singular value if $H_N$ is not Hermitian) of the matrix $H_N$. Since the minimal eigenvalue/singular value is a constant independent of $N$, we have
$$
\kappa_A = 1+O(\varepsilon\lambda_{\text{max}}).
$$ 
Recall that for the (second-order) PDEs discussed in this paper, we have $\lambda_{\text{max}} = O(dN^{2/d})$. Henceforth, we omit the minor dependence on $d$, and we conclude that the query complexity for the time-sequential strategy is 
$$
O\left((1/\varepsilon)\log(1/\varepsilon)\right) + O\left(N^{2/d}\log(1/\varepsilon)\right) \le \tilde O\left(N^{2/d} (1/\varepsilon)\right).
$$
Here, the $1/\varepsilon$ factor comes from the iterations of $K=O(1/\Delta\tau)$ time steps, and the notation $\tilde O$ is similar to $O$ but allows an additional $\mathrm{polylog}$ factor. 

For the one-shot strategy, we use the enlarged matrix in \cite{Krovi2023} and denote it by $\tilde A$. To construct the block encoding of $\tilde A$, the time step is chosen to be $\Delta\tau = O(\|H_N\|^{-1})=O(\lambda_{\text{max}}^{-1})$. According to Theorem 4 in \cite{Krovi2023}, and noting that we should take the parameters in \cite{Krovi2023} as $m=p=O(1/\Delta\tau)$, $k=\log(\lambda_{\text{max}}^3/\varepsilon)+O(1)$, the condition number of
the enlarged matrix satisfies 
$$
\kappa_{\tilde A} = O\left(N^{2/d}\sqrt{\log(N^{6/d}/\varepsilon)}\right). 
$$
Therefore, the query complexity of the one-shot strategy is 
$$
\tilde O\left(N^{2/d}\mathrm{polylog}(1/\varepsilon)\right).
$$
Focusing only on the asymptotic order of the key parameters $N$ and $\varepsilon$ (without comparing the prefactor), we find that the one-shot strategy is better than the time-sequential strategy in precision.  

To obtain the gate complexity of the QSVT-based QLSA, one has to consider also the explicit gate implementation of the basic oracle (i.e., the block encoding of $A/s$ or $\tilde A/s$). 
It was believed that there exist efficient gate implementations (up to $O(\mathrm{polylog} N)$) for any sparse matrix, but it was clarified recently that this is true only for well-structured (possibly non-sparse) matrices including the Toeplitz matrix, extended binary tree matrix, etc. with limited numbers of distinct values \cite{Camps.2024, Sunderhauf.2024}.
For this moment, we assume the simple cases (e.g. constant coefficients without potential) such that the gate complexity of the oracle is $O(\mathrm{polylog}\, N)$. Recalling that the gate complexity of the AAPITE algorithm is $O(\varepsilon^{-1-\alpha}\mathrm{polylog}\, N)$ (see Table \ref{sec3:tab1}) where $\alpha\in (0,1]$ comes from the polynomial approximation of the potential term, we conclude that compared to the QSVT-based QLSA, the AAPITE algorithm has an exponential improvement in the matrix size at the cost of a worse behavior regarding the error bound $\varepsilon$. The above discussion is summarized in Table \ref{sec1:tab2} in Sect.~\ref{sec:1}. 
\begin{remark}
We give a key remark on our exponential improvement in the matrix size $N$. Considering a simple case without the potential function so that there is no Suzuki-Trotter error, we can employ the exact PITE circuit to implement the ITE operator. This does not introduce any approximation error, and the gate complexity/circuit depth is $O(dN^{2/d})$, which is the same order as the best-order QLSAs (in 1D case, the gate complexity for the exact PITE is $O(N)$ which is better than $O(N^2)$ by the QLSAs). To reduce the gate complexity regarding $N$, we introduce an approximation, which yields large errors for the large Fourier coefficients. Thanks to the underlying structure of the PDEs, the quantum statevector corresponds to the solution which is usually smooth enough so that the $k$-th Fourier coefficient decreases rapidly as $k\to \infty$. 
This guarantees that our approximation algorithm provides a good approximate solution for a fixed simulation time. 
Such an idea cannot be readily adopted by the QLSAs. The essential reason is the singularity of the inverse operator near the zero point, which means the approximation will lead to larger errors for the smaller eigenvalues. 
Thus, it seems that the dependence on the condition number (i.e., polynomial dependence on $N$) is not avoidable unless the matrix $A$ or $\tilde A$ itself is preconditioned to a modified one with eigenvalues bounded by $\mathrm{polylog}\, N$. 
Therefore, the preconditioning technique is indispensable for the QLSAs, and one has to consider the additional classical and quantum computational costs, which are in general $O(\mathrm{poly}\, N)$, but may be reduced in some specific cases. Some ideas were discussed in \cite{Bagherimehrab.2023pre} and the references therein. 
\end{remark}

\begin{remark}
We have compared the detailed gate count of the QVST-based QLSA with our algorithm for the advection-diffusion equation without the potential function. In this case, the block encoding of $A$ is possible by the explicit quantum circuits in \cite{Camps.2024, Sunderhauf.2024}. Applying the QSPPACK \cite{LinQSP} to calculate the phase factors in the QSVT, we find that the detailed gate count of the time-sequential QSVT-based QLSA is much larger than our AAPITE algorithm for small $\varepsilon\approx 10^{-3}$. 
As for the one-shot QSVT-based QLSA, the condition number for the enlarged matrix is larger than $3000$ even for relatively small $N=64$. Although the theoretical gate complexity suggests a superiority of the QLSA for sufficiently small $\varepsilon$, the derivation of the phase factors is impossible \cite{Lapworth2024} since the classical computational cost for the phase factors scales quadratically regarding the condition number. 
Although many QLSAs give a good theoretical order regarding the error bound, unfortunately, the algorithms are not applicable in practice due to the large condition numbers (for large matrices). 
\end{remark}

\subsection{Comments}
\label{subsec:5-5}

We end up with the following two comments. 

The first comment is on some other quantum algorithms solving the time evolution equation beyond the QLSAs. 
In the QLSAs, usually the basic oracle is the block encoding of normalized discretized differential operators, which is further implemented by the oracles that obtain the indices and values of the entries (see e.g. \cite{Krovi2023, Lin2022}).   
Recently, there are other methods to avoid such a matrix oracle, which use the access to the real-time evolution operation of the (modified) matrix \cite{Sato.2024pre, Bharadwaj.2024pre, Sato.2024, Jin.2022, Jin.2023}. Based on the finite difference method (FDM) for simple heat and wave equations with constant coefficients, such RTE operations can be efficiently constructed by explicit quantum gates. However, for general time evolution equations, e.g. there is a potential function in the PDEs, the RTE operations are not easy to implement and one suggestion is to call the matrix oracles again. 
In particular, we mention the method of the Schr\"odingerisation proposed in \cite{Jin.2022}, the algorithm can be applied to a class of time evolution equations including the diffusion equation with a function potential. 
In general, this algorithm also requires the matrix oracle, and its explicit implementation is discussed in \cite{Guseynov.2024}, which has a quite large prefactor. On the other hand, for the specific problem of the diffusion equation with a function potential, Jin et al. used the Suzuki-Trotter formula to obtain an explicit quantum circuit using two diagonal unitary matrices \cite{Jin.2023}. By using the explicit implementation of the diagonal unitary matrices in \cite{Huang.2024p}, this algorithm seems to give a similar gate complexity as our AAPITE algorithm. 
The rigorous error estimates on the discretization error, Suzuki-Trotter error, and the approximation error were not provided in \cite{Jin.2023}, so that the precise parameter dependence in the prefactor is not yet clarified (the accumulation of the errors and a possible hidden dependence on $N$ in the choice of time step). Therefore, we need to discuss in more detail to find the strong and weak points of both algorithms in a future work.  

The second comment is about the end-to-end quantum speedup including the readout of the quantum state. In the problems of solving PDEs, what we need are the solutions themselves or sometimes certain physical quantities that can be integrated/calculated in terms of the solutions. For all the mentioned quantum algorithms, the normalized solution is prepared in the amplitudes of the corresponding basis, and such an amplitude encoding is indispensable to achieve the quantum speedup regarding the matrix size $N$. After the quantum circuit prepares the quantum state corresponding to the solution, we have to readout the information of the normalized solution at grid points and the norm of the solution, or at least the expectations of some physical observables using the prepared quantum state (i.e., the normalized solution). In either case, this can be done by repeated measurements of the quantum state, and it boils down to the quantum amplitude estimation (QAE) problem. Owing to the stochastic error in readout, we need $\Omega(1/\varepsilon)$ repetitions to estimate the amplitudes with the desired error bound $\varepsilon$, see e.g., \cite{Manzano.2023, Maronese.2024} and the references therein (there is also a prefactor of $O(N)$ in the naive readout if the solution at every grid point is needed). This implies that if the matrix size $N$ is fixed in advance, then the end-to-end quantum algorithms (the initial state preparation is efficient if the underlying function is good enough) do NOT outperform any classical methods with logarithmic dependence on $1/\varepsilon$ (in the asymptotic behavior of the gate complexity as $\varepsilon\to 0$). 
Considering the end-to-end quantum algorithms, the quantum speedup appears in the following two settings:

\noindent (1) speedup regarding the matrix size $N$ when the error bound is relatively large and fixed in advance; 

\noindent (2) total speedup after balancing the improvement in $N$ and the drawback in $\varepsilon$. 

\noindent The first setting refers to the applications that require the treatment of large matrices with moderate precision, and the quantum speedup appears when $N\ge N_{\varepsilon_0}$. Here, $N_{\varepsilon_0}$ is a large constant depending on the fixed error bound $\varepsilon_0>0$. 
In the second setting, we use the speedup in the matrix size $N$ to recover the drawback in the error bound $\varepsilon$ by choosing a suitable $N=N(\varepsilon)$, depending on the error bound, and consider the total complexity with only one parameter $\varepsilon$. This direction includes the applications to high-precision solvers and was discussed in \cite{Childs.2021, Montanaro.2016} for the QLSAs. In particular, Montanaro and Pallister \cite{Montanaro.2016} investigated the conventional FEM, and demonstrated a polynomial speedup for sufficiently large dimension $d>d_0$ where $d_0=p+1$ depends on the polynomial degree $p$ of the FEM basis functions (since $N=h^{-1}=\Omega(\varepsilon^{-d/(p+1)})$ due to the discretization error of the FEM). For the QLSAs, this figures out that the theoretical quantum speedup (i.e., the asymptotic performance as $\varepsilon\to 0$) is polynomial for large $d$ if we need relatively large $N=\Omega(\varepsilon^{-\gamma})$ for some $\gamma>0$, while the speedup would vanish if we only need relatively small $N=O(\mathrm{polylog}(1/\varepsilon))$. Therefore, whether there exists a theoretical quantum speedup or not heavily relies on the relation between the matrix size, depending on the grid/mesh parameter, and the error bound. For the conventional second-order PDEs (Poisson equations, diffusion equations, and wave equations), such a relation can be theoretically estimated by establishing the error estimate for the discretization error (i.e., the difference between the continuous solution and the discretized solution). The error estimate depends on the choice of the discretization methods. However, it is difficult to rigorously estimate the discretization error for the complicated coupled PDE system with nonlinearity. Therefore, in the second setting, even the theoretical quantum speedup is still not clarified for many complex systems in practical applications. 

\section{Extension to system of diffusion equations}
\label{sec:6}

We also consider a coupled system of linear diffusion equations, which can be regarded as the linearized reaction-diffusion system. 
The governing equations are as follows:
\begin{align}
\label{sec5:eq-gov}
\partial_t u^{(j)}(\mathbf{x},t) = a_j \nabla^2 u^{(j)}(\mathbf{x},t) - \sum_{\ell=1}^J p_{j\ell}(\mathbf{x},t) u^{(\ell)}(\mathbf{x},t), \quad \mathbf{x}\in [0,L]^d,\ t>0,
\end{align}
for $j=1,\ldots,J$, where $a_j\in \mathbb{R}_{>0}$, $j=1,\ldots,J$ are positive diffusion coefficients and $(p_{j\ell})_{j,\ell=1,\ldots,J}\in \mathbb{R}^{J\times J}$ denotes the coefficient matrix of reaction. 
Moreover, we include the initial conditions: 
\begin{align}
\label{sec5:eq-init}
u^{(j)}(\mathbf{x},0) = u_0^{(j)}(\mathbf{x}), \quad j=1,\ldots,J.
\end{align}
Let $n\in \mathbb{N}$, $N_0=2^n$, $N=N_0^d$, $[N_0] := \{-N_0/2,-N_0/2+1,\ldots,N_0/2-1\} = [\tilde N_0]-N_0/2$, and introduce the basis
$\phi_{\mathbf{k}}(\mathbf{x}) = \frac{1}{L^{d/2}} \mathrm{exp}\left(\mathrm{i}\frac{2\pi}{L}k_1 x_1\right) \ldots \mathrm{exp}\left(\mathrm{i}\frac{2\pi}{L}k_{d} x_{d}\right)$, $\mathbf{k}\in [N_0]^d$. Let the following linear combination of the bases:
$$
u_{N}^{(j)}(\mathbf{x},t) = \sum_{\mathbf{k}\in [N_0]^d} \hat u_{\mathbf{k}}^{(j)}(t)\phi_{\mathbf{k}}(\mathbf{x}),
$$
solve the governing equations \eqref{sec5:eq-gov} with the initial conditions: 
\begin{align*}
u_N^{(j)}(\mathbf{x},0) = \sum_{\mathbf{k}\in [N_0]^d} (\phi_{\mathbf{k}}, u_0^{(j)}) \phi_{\mathbf{k}}(\mathbf{x}), \quad j=1,\ldots,J,
\end{align*}
which is derived from \eqref{sec5:eq-init}. 
By taking $\mathbf{x} = \mathbf{p}_{\mathbf{l}}$ (see Eq.~\eqref{sec3:eq-x}) for all $\mathbf{l}\in [\tilde N_0]^d$, we apply the real-space grid method (i.e., a special FSM) to obtain a matrix equation:
\begin{align*}
\partial_t 
\begin{pmatrix}
F_N &  & \\
 & \ddots & \\
 &  & F_N
\end{pmatrix}
\begin{pmatrix}
\mathbf{\hat u}^{(1)} \\
\vdots \\
\mathbf{\hat u}^{(J)}
\end{pmatrix}
= & - \begin{pmatrix}
F_N &  & \\
 & \ddots & \\
 &  & F_N
\end{pmatrix}
\begin{pmatrix}
D_N^{(1)} &  & \\
 & \ddots & \\
 &  & D_N^{(J)}
\end{pmatrix}
\begin{pmatrix}
\mathbf{\hat u}^{(1)} \\
\vdots \\
\mathbf{\hat u}^{(J)}
\end{pmatrix} \\
& - \begin{pmatrix}
P_N^{(11)} & \cdots & P_N^{(1J)}\\
\vdots & \ddots & \vdots \\
P_N^{(J1)} & \cdots & P_N^{(JJ)}
\end{pmatrix}
\begin{pmatrix}
F_N &  & \\
 & \ddots & \\
 &  & F_N
\end{pmatrix}
\begin{pmatrix}
\mathbf{\hat u}^{(1)} \\
\vdots \\
\mathbf{\hat u}^{(J)}
\end{pmatrix},
\end{align*}
where 
\begin{align*}
&\mathbf{\hat u}^{(j)} = \sum_{\mathbf{k}\in [\tilde N_0]^d} \hat u_{\mathbf{k}-N_0/2}^{(j)}(t) \ket{\mathbf{k}}, \quad j=1,\ldots,J, \\
&D_N^{(j)} = \sum_{\mathbf{k}\in [\tilde N_0]^d} a_j (2\pi/L)^2\left|\mathbf{k}-\frac{N_0}{2}\right|^2 \ket{\mathbf{k}} \bra{\mathbf{k}}, \quad j=1,\ldots,J, \\
&P_N^{(j\tilde j)}(t) = \sum_{\mathbf{l}\in [\tilde N_0]^d} p_{j\tilde j}(\mathbf{p}_{\mathbf{l}},t) \ket{\mathbf{l}} \bra{\mathbf{l}}
= \sum_{\mathbf{l}\in [\tilde N_0]^d} p_{j\tilde j}\left(\frac{L}{N_0}\mathbf{l}, t\right) \ket{\mathbf{l}} \bra{\mathbf{l}}, \quad j,\tilde j=1,\ldots,J.
\end{align*}
Since we intend to solve the equations on the quantum computers, we should assume $J=2^{j_0}$ for some integer $j_0\in \mathbb{N}$. Otherwise, we need to introduce some dummy equations. 
Henceforth, we consider the simplest case of $J=2$. Then, the above matrix equation yields
\begin{align*}
\partial_t \mathbf{y} &= - \left(I\otimes F_N\right) \left(\ket{0}\bra{0}\otimes D_N^{(1)}+\ket{1}\bra{1}\otimes D_N^{(2)}\right) \left(I\otimes F_N^\dag\right) \mathbf{y} - 
\begin{pmatrix}
P_N^{(11)} & P_N^{(12)}\\
P_N^{(21)} & P_N^{(22)}
\end{pmatrix}
\mathbf{y} =: - H_{2N}(t) \mathbf{y},
\end{align*}
where $\mathbf{y} = \left(I\otimes F_N\right) \mathbf{\hat u}$ with $\mathbf{\hat u} = (\mathbf{\hat u}^{(1)};\mathbf{\hat u}^{(2)})$. 
For a given target time $T>0$, we intend to derive $\mathbf{y}(T)$. 
By choosing sufficiently large $K\in\mathbb{N}$ and $\Delta\tau = T/K$, we integrate the above system and obtain the approximation
\begin{align*}
\mathbf{y}(T) \approx \prod_{j=0}^{K-1} \mathrm{exp}\left(-\Delta\tau H_{2N}(j\Delta\tau)\right) \mathbf{y}(0),
\end{align*}
with $\mathbf{y} = (\mathbf{y}^{(1)};\mathbf{y}^{(2)})$, and 
\begin{align*}
\mathbf{y}^{(j)}(0) = \frac{1}{N_0^{d/2}} \sum_{\mathbf{l}\in [\tilde N_0]^d} 
\left(\sum_{\mathbf{k}\in [N_0]^d} (\phi_{\mathbf{k}}, u_0^{(j)})\mathrm{exp}\left(\mathrm{i}\frac{2\pi}{N_0}\mathbf{k}\cdot\mathbf{l}\right) \right)\ket{\mathbf{l}} 
\approx \left(\frac{L}{N_0}\right)^{d/2} \sum_{\mathbf{l}\in [\tilde N_0]^d} u_0^{(j)}(\mathbf{x}_\mathbf{l}) \ket{\mathbf{l}},
\end{align*}
for $j=1,2$. Now, we are interested in the ITE operator $\mathrm{exp}\left(-\Delta\tau H_{2N}(j\Delta\tau)\right)$ for each $j=1,\ldots,K$. 
By the first-order Suzuki-Trotter formula, we have
\begin{align*}
\mathrm{exp}\left(-\Delta\tau H_{2N}(j\Delta\tau)\right) &\approx \mathrm{exp}\left(-\Delta\tau \left(I\otimes F_N\right) \left(\ket{0}\bra{0}\otimes D_N^{(1)}+\ket{1}\bra{1}\otimes D_N^{(2)}\right) \left(I\otimes F_N^\dag\right)\right) \\
&\quad \mathrm{exp}\left(-\Delta\tau \left(\ket{0}\bra{0}\otimes P_N^{(11)}(j\Delta\tau)+\ket{1}\bra{1}\otimes P_N^{(22)}(j\Delta\tau)\right)\right) \\
&\quad \mathrm{exp}\left(-\frac12\Delta\tau X\otimes \left(P_N^{(12)}(j\Delta\tau)+P_N^{(21)}(j\Delta\tau)\right)\right) \\
&\quad \mathrm{exp}\left(-\frac12\mathrm{i}\Delta\tau Y\otimes \left(P_N^{(12)}(j\Delta\tau)-P_N^{(21)}(j\Delta\tau)\right)\right).
\end{align*}
Here and henceforth, $X$, $Y$, and $Z$ denote Pauli matrices. More precisely, we rewrite the four operations as follows:
\begin{align*}
&\quad \mathrm{exp}\left(-\Delta\tau \left(I\otimes F_N\right) \left(\ket{0}\bra{0}\otimes D_N^{(1)}+\ket{1}\bra{1}\otimes D_N^{(2)}\right) \left(I\otimes F_N^\dag\right)\right) \\
&= \left(\ket{0}\bra{0}\otimes \mathrm{exp}\left(-\Delta\tau F_ND_N^{(1)}F_N^\dag\right)\right)+\left(\ket{1}\bra{1}\otimes \mathrm{exp}\left(-\Delta\tau F_ND_N^{(2)}F_N^\dag\right)\right) \\
&= (I\otimes F_N)\mathrm{exp}\left(-\Delta\tau \left(\ket{0}\bra{0}\otimes D_N^{(1)}+\ket{1}\bra{1}\otimes D_N^{(2)}\right)\right)(I\otimes F_N^\dag), 
\end{align*}
\begin{align*}
&\quad \mathrm{exp}\left(-\Delta\tau \left(\ket{0}\bra{0}\otimes P_N^{(11)}(j\Delta\tau)+\ket{1}\bra{1}\otimes P_N^{(22)}(j\Delta\tau)\right)\right) \\ 
&= \mathrm{exp}\left(\Delta\tau P_{0}\right)\mathrm{exp}\left(-\Delta\tau \left(\ket{0}\bra{0}\otimes P_N^{(11)}(j\Delta\tau)+\ket{1}\bra{1}\otimes P_N^{(22)}(j\Delta\tau)+P_{0}I\right)\right), 
\end{align*}
\begin{align*}
&\quad \mathrm{exp}\left(-\frac12\Delta\tau X\otimes \left(P_N^{(12)}(j\Delta\tau)+P_N^{(21)}(j\Delta\tau)\right)\right) \\
&= (H\otimes I)\mathrm{exp}\left(-\frac12\Delta\tau Z\otimes \left(P_N^{(12)}(j\Delta\tau)+P_N^{(21)}(j\Delta\tau)\right)\right) (H\otimes I) \\
&= \mathrm{exp}\left(\frac12\Delta\tau P_{1}\right) (H\otimes I) \mathrm{exp}\left(-\frac12\Delta\tau \left(Z\otimes \left(P_N^{(12)}(j\Delta\tau)+P_N^{(21)}(j\Delta\tau)\right)+P_{1}I\right)\right) (H\otimes I), 
\end{align*}
\begin{align*}
&\quad \mathrm{exp}\left(-\frac12\mathrm{i}\Delta\tau Y\otimes \left(P_N^{(12)}(j\Delta\tau)-P_N^{(21)}(j\Delta\tau)\right)\right) \\
&= (W\otimes I) \mathrm{exp}\left(-\frac12\mathrm{i}\Delta\tau Z\otimes \left(P_N^{(12)}(j\Delta\tau)-P_N^{(21)}(j\Delta\tau)\right)\right) (W^\dag\otimes I),
\end{align*}
where 
\begin{align*}
P_{0}(j\Delta\tau) &= -\min\left\{\min_{\mathbf{x}\in [0,L]^d} p_{11}(\mathbf{x},j\Delta\tau), \min_{\mathbf{x}\in [0,L]^d} p_{22}(\mathbf{x},j\Delta\tau)\right\}, \\
P_{1}(j\Delta\tau) &= \max_{\mathbf{x}\in [0,L]^d} \left|p_{12}(\mathbf{x},j\Delta\tau)+p_{21}(\mathbf{x},j\Delta\tau)\right|,
\end{align*}
and 
$$
H = \frac{1}{\sqrt{2}}\begin{pmatrix}
1 & 1 \\
1 & -1
\end{pmatrix}, \quad
W = \frac{1}{\sqrt{2}}\begin{pmatrix}
1 & \mathrm{i} \\
\mathrm{i} & 1
\end{pmatrix} = R_x(-\pi/2).
$$
Recalling that $D_N^{(j)}$, $j=1,2$ and $P_N^{(j\tilde j)}$, $j,\tilde j=1,2$ are both diagonal matrices, the above calculations imply 
\begin{align*}
\mathrm{exp}\left(-\Delta\tau H_{2N}(j\Delta\tau)\right) 
&\approx \mathrm{exp}\left(\Delta\tau(P_0+P_1/2)\right) (I\otimes F_N) \mathrm{exp}\left(-D_1\right)(I\otimes F_N^\dag) \mathrm{exp}\left(-D_2\right) \\
&\quad (H\otimes I) \mathrm{exp}\left(-D_3\right) (HW\otimes I)\mathrm{exp}\left(-\mathrm{i}D_4\right) (W^\dag\otimes I),
\end{align*}
where $D_1,D_2,D_3$ and $D_4$ are all diagonal matrices with non-negative entries. Therefore, we can apply the proposed AAPITE operator for three times and a RTE operator once, all for real-valued diagonal matrices, to approximately implement the ITE operator in each time step. The quantum circuit in one time step is demonstrated in Fig.~\ref{sec5:fig3}. 
\begin{figure}
\centering
\resizebox{15cm}{!}{
\includegraphics[keepaspectratio]{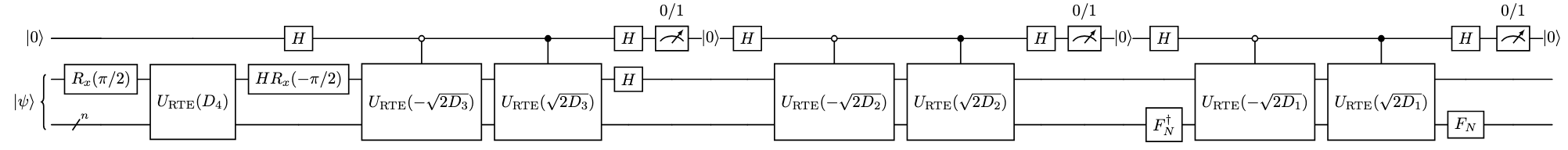}
}
\caption{Quantum circuit in each time step for solving a system of diffusion equations using the AAPITE. Here, three AAPITE circuits are employed. The first one is for the diffusion, and the other two are for the diagonal components and the symmetric part of the non-diagonal components of the reaction matrix, respectively. }
\label{sec5:fig3}
\end{figure}

If we can retrieve the solutions after each time step by measurements, then we can employ the solutions to construct the Hamiltonian in the next time step, which means we can deal with the nonlinear reaction-diffusion equations. 
For example, letting $J=2$ and $d=2$ in Eq.~\eqref{sec5:eq-gov}, we have
\begin{align*}
\partial_t u^{(1)}(\mathbf{x},t) = a_1 \nabla^2 u^{(1)}(\mathbf{x},t) - p_{11}(\mathbf{x},t) u^{(1)}(\mathbf{x},t) - p_{12}(\mathbf{x},t) u^{(2)}(\mathbf{x},t), \ \mathbf{x}\in [0,L]^2,\ t>0, \\
\partial_t u^{(2)}(\mathbf{x},t) = a_2 \nabla^2 u^{(2)}(\mathbf{x},t) - p_{21}(\mathbf{x},t) u^{(1)}(\mathbf{x},t) - p_{22}(\mathbf{x},t) u^{(2)}(\mathbf{x},t), \ \mathbf{x}\in [0,L]^2,\ t>0. 
\end{align*}
Let $p_{ij}$, $i,j=1,2$ be some functions of $u^{(1)}$ and $u^{(2)}$. 
Then, we obtain the nonlinear reaction-diffusion equations. 
In the following context, we discuss two numerical examples. 

\noindent \underline{\bf Example 1} Turing Pattern formulation

Let $a_1=0.005$, $a_2=0.1$, $L=2\pi$, $p_{11}=(u^{(1)})^2-0.6$, $p_{12}=1$, $p_{21}=-1.5$, $p_{22}=2$. 
We consider the initial conditions of four point sources described by the linear combination of the Gaussian functions:
$$
u_0^{(1)}(\mathbf{x}) = u_0^{(2)}(\mathbf{x}) = A\sum_{j=1}^4 \mathrm{exp}\left(-\frac{|\mathbf{x}-\mathbf{x}_0^{(j)}|^2}{2\sigma^2}\right).
$$
Here, $\sigma^2=0.05$, $\mathbf{x}_0^{(1)}=(\pi/2,\pi/2)^\mathrm{T}$, $\mathbf{x}_0^{(2)}=(\pi/2,3\pi/2)^\mathrm{T}$, $\mathbf{x}_0^{(3)}=(3\pi/2,\pi/2)^\mathrm{T}$, $\mathbf{x}_0^{(4)}=(3\pi/2,3\pi/2)^\mathrm{T}$, and $A$ is a normalized parameter such that the $\ell^2$-norm of the function at $N=N_0^2$ grid points is normalized to $1$. 
We provide a simulation result with $N_0=2^4$ in Fig.~\ref{appE:Fig1}. 
\begin{figure}[htb]
\centering
\resizebox{14cm}{!}{
\includegraphics[keepaspectratio]{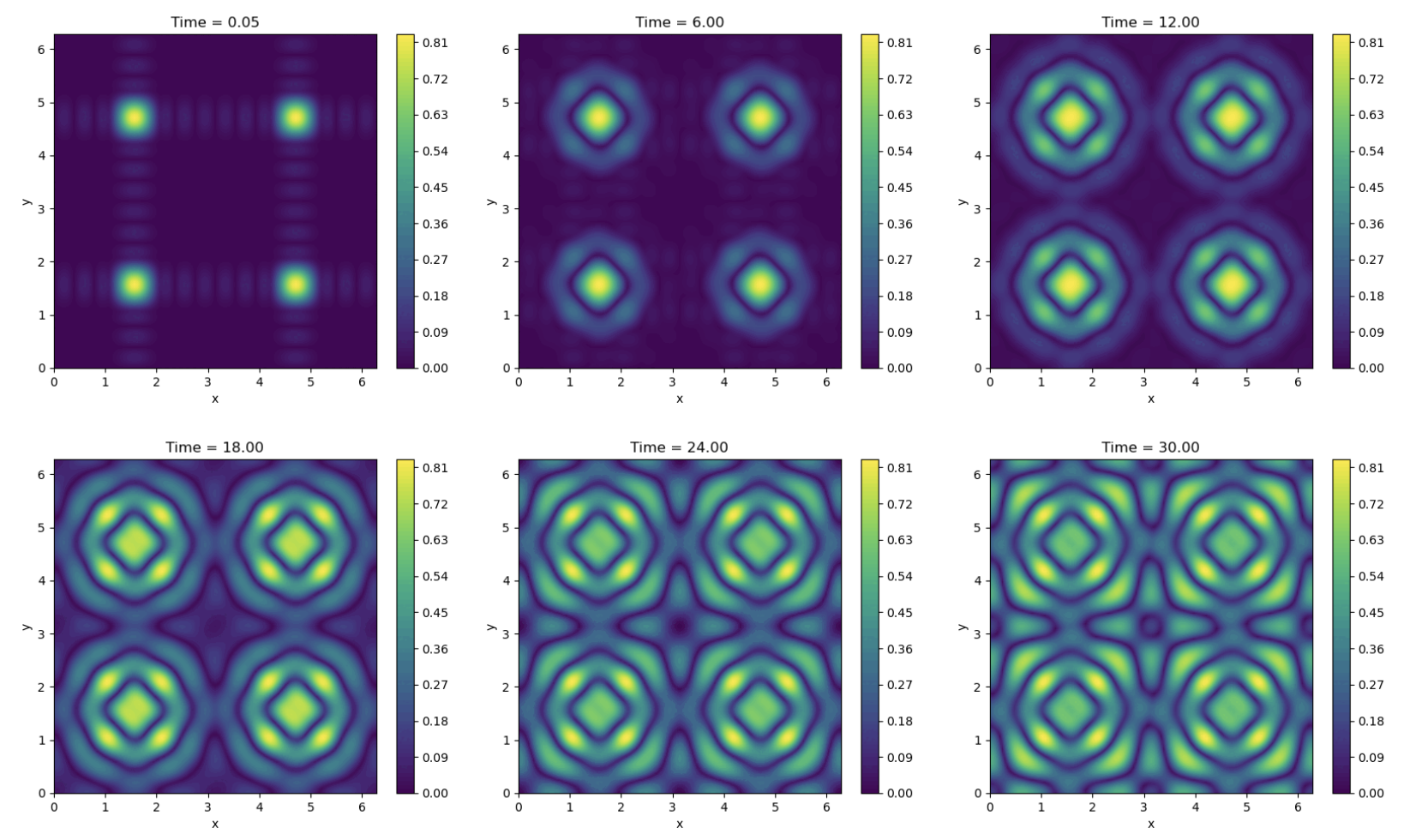}
}
\caption{Time evolution of the $(2^4,2^7)$-quantum solution using the AAPITE operator with $\Delta\tau=0.05$ for the first component of a reaction-diffusion system. The quantum solutions at $t=0.05, 6, 12, 18, 24, 30$ are visualized using the colormap plots. In each subplot, the $x$ and $y$ axes describe the location in $x_1$-direction and $x_2$-direction, respectively. }
\label{appE:Fig1}
\end{figure}
It demonstrates the formulation of ring pattern around the initial point sources. The simulation is done by Qiskit emulator, and we directly obtain the success probability by the statevector after each PITE step. 
In this example, the total success probability is extremely small ($\approx 10^{-21}$). 
According to Sect.~\ref{subsec:3-2}, the success probability for this example is proportional to $\mathrm{exp}\left(-2T(P_0+P_1/2)\right)=\mathrm{exp}\left(-1.7T\right)$, this explains the extremely small success probability for this simulation. 

\noindent \underline{\bf Example 2} Burgers' equation

Let $a_1=a_2=0.05$, $L=2\pi$, $p_{11}=\partial_{x_1} u^{(1)}$, $p_{12}=\partial_{x_2} u^{(1)}$, $p_{21}=\partial_{x_1} u^{(2)}$, $p_{22}=\partial_{x_2} u^{(2)}$. Then, Eq.~\eqref{sec5:eq-gov} becomes (viscous) Burgers' equation with viscosity $\nu = 0.05$: 
$$
\partial_t \mathbf{u}(\mathbf{x},t) = \nu \nabla^2 \mathbf{u}(\mathbf{x},t) - \mathbf{u}\cdot\nabla \mathbf{u}(\mathbf{x},t), \quad \mathbf{x}\in [0,L]^2,\ t>0. 
$$
Here, the spatial derivatives of the solutions are calculated from the solutions using the fast Fourier transform on a classical computer. 
As a simple numerical example, we consider the initial conditions:
\begin{align*}
& u_0^{(1)}(x_1,x_2) = \sin(2x_1), && x_1,x_2\in [0,2\pi],\\
& u_0^{(2)}(x_1,x_2) = 0.5, && x_1,x_2\in [0,2\pi].
\end{align*}
Under these initial conditions, the 2D Burgers' equation reduces to a 1D Burgers' equation, and the simulation result is shown in Fig.~\ref{appE:Fig2}. 
\begin{figure}[htb]
\centering
\resizebox{14cm}{!}{
\includegraphics[keepaspectratio]{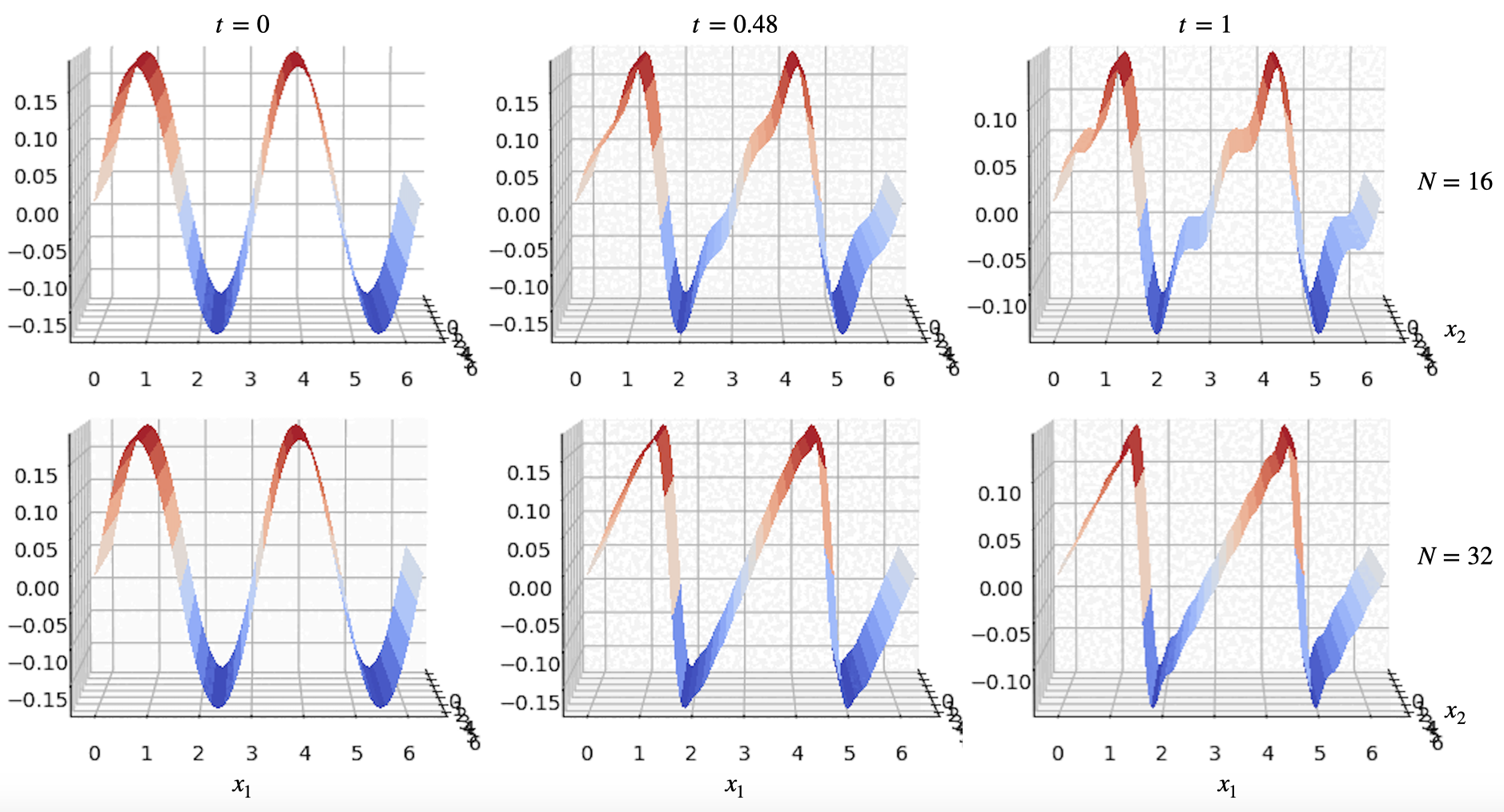}
}
\caption{Time evolution of the $(N_0,2^7)$-quantum solution using the AAPITE operator with $\Delta\tau=0.04$ for the first component of the viscous Burgers' equation. The quantum solutions at $t=0, 0.48, 1$ are visualized using 3D plots from left to right. We take $N_0=16$ in the first row and $N_0=32$ in the second row. Despite of $N_0=16$ or $N_0=32$, the success probability is about $2\times 10^{-4}$. }
\label{appE:Fig2}
\end{figure}
We find that the initial sine wave becomes a sawtooth wave as time goes by. 
Moreover, the quantum solution has better performance for a larger grid parameter $N$ as we observe less numerical oscillation in the $(2^5,2^7)$-quantum solution than that in the $(2^4,2^7)$-quantum solution. 
The success probability is about $2\times 10^{-4}$, which is relatively large compared to the extremely small success probability in the previous example. 
This implies nonlinearity does not necessarily lead to a small success probability. According to the theoretical analysis in Sect.~\ref{subsec:3-2}, we find that the success probability mainly depends on the change of the $\ell^2$-norm of the solution compared to the initial condition. Therefore, quantum solvers may be more expensive (large number of repetitions and measurements) than their classical counterparts in the simulations of highly dissipative systems, whatever they are linear or nonlinear. 

In the simulation of Burgers' equation, a further observation is that smaller viscosity $\nu$ yields severer numerical oscillation (this is also a crucial problem in the classical simulations), and hence, we have to take larger grid parameter $N$ and smaller time step $\Delta\tau$ to relieve the oscillation.  

\section{Summary and concluding remarks}

In this paper, we discuss a new approximate PITE and its application to solving the time evolution differential equations using examples of the advection-diffusion-reaction equations. 

First, we proposed an alternative approximate PITE (AAPITE) operator, which approximates the ITE operator $m_0 \mathrm{exp}\left(-\Delta\tau \mathcal{H}\right)$ with $m_0=1$. We obtained analytical estimations for both the $\ell^2$-error and its success probability. Compared to the original approximate PITE in \cite{Kosugi.2022} where $m_0$ is strictly smaller than one, the AAPITE relieves the restriction in the time step $\Delta\tau$. Moreover, it avoids the problem of the rapid vanishing of the success probability. 
For the convenience of applications, the $\ell^2$-error and the success probability in the case of the Hamiltonian coming from the advection-diffusion-reaction equations were also estimated. 

Second, as for the gate implementation of the AAPITE, we need to employ a basic oracle of the RTE operator for the squared root Hamiltonian, that is, $U_{\text{RTE}}(\sqrt{2\Delta\tau \mathcal{H}})$. Although the squared root Hamiltonian is usually more complicated than the original one in general, one can use the general Hamiltonian simulation techniques in e.g. \cite{Gilyen.2019, BCK15}. 
In this paper, for the advection-diffusion-reaction equation, we provided an explicit gate implementation for the basic oracle based on the real-space grid method using the efficient quantum circuits for the diagonal unitary matrices \cite{Kosugi.2023, Huang.2024p}. Applying the techniques in \cite{Huang.2024p}, we can implement the AAPITE operator by $O(\mathrm{poly}\, n)$ two-qubit gates with only $2$ ancillary qubits (one for the AAPITE and the other for the diagonal unitary matrices) in 1D cases and $2n+2+\lceil\log_2 d\rceil$ ancillary qubits (another $2n+\lceil\log_2 d\rceil$ qubits for the distance register) in multi-dimensional cases. The required number of ancillary qubits is the least one among all the existing quantum algorithms to our best knowledge. 

Third, we provided 1D/2D numerical examples for solving the advection-diffusion-reaction equations by the proposed algorithm. The simulations were executed by Qiskit (a quantum gate-based emulator). The numerical results demonstrated good agreement with the true/reference solutions. Moreover, based on the mathematical theory for the real-space grid method (\ref{appB}), we suggested pre-processing/post-processing steps on classical computers to derive the solutions on a finer grid (without loss of precision) when the (quantum) computational grid is relatively small due to the limited capacity of the near-term quantum computers. Furthermore, we divided the error into the discretization error, the Suzuki-Trotter error and the (PITE) approximation error, and confirmed the important dependence on the grid parameter $N$ and the time step $\Delta\tau$ numerically (see \ref{appE}), which coincides with the theoretical estimations. 

Fourth, we compared our algorithm based on the AAPITE with the ones in the previous works \cite{Kosugi.2022, Nishi.2023}. The result implies that the original approximate PITE \cite{Kosugi.2022} and even the improved version \cite{Nishi.2023} with varying time steps are not suitable for the task of solving differential equations since they undergo the rapid decrease of the success probability owing to $m_0<1$. 
Besides, we also compared the AAPITE with a (time-sequential) QLSA (the HHL algorithm) and a VQA in \cite{Ingelmann.2024} for a specified advection-diffusion equation, and concluded that our quantum algorithm based on the real-space grid method (a specific FSM) outperforms the (time-sequential) QLSAs using the sparse discretized matrices by the backward Euler FDM. If the boundary condition is not important or is periodic, then the discretization based on the FSM is recommended because it has better precision than many other discretization methods (FDM, FEMs, etc.) as $N$ increases and can be explicitly and efficiently implemented on quantum computers. 
In fact, the AAPITE circuit can also be applied to other discretized matrices deriving from the FDMs, FEMs, etc. In such cases, we need an oracle implementing the RTE operator for the square root of the discretized Hamiltonian, which is possible by the general techniques in \cite{Gilyen.2019, BCK15} although the explicit gate count/circuit depth of the oracle is not yet clarified in general cases. 
Moreover, we discussed the implementation with the QSVT-based QLSA, one of the best-order QLSAs, and found that our proposed algorithm achieved an exponential improvement regarding the matrix size at the cost that the dependence on the reciprocal of the error bound is polynomial. Furthermore, the practical implementation of the QSVT-based QLSA is not applicable because the classical computation of the phase factors in the QSVT becomes unstable and intractable for a large condition number. The QLSAs require an efficient preconditioning that will introduce additional classical/quantum costs, while the preconditioning is not necessary for our algorithm.

Finally, we extended our algorithm to the (linearized) diffusion-reaction systems. For the quantum algorithm, we need only $dn+\log_2 J$ qubits to simulate the system consisting of $J$ equations, which seems more efficient than its classical counterpart. In the simplest case of $J=2$, we provided the detailed implementation of the system using three AAPITE circuits in each time step. Moreover, if we allow the measurement of the solution after each time step, we can apply such information to update the Hamiltonian in the next time step, so that nonlinear systems can be treated. 
Simulations of Turing Pattern formulation and Burgers' equation were provided for potential applications if efficient statevector preparation/readout techniques would be developed in the future. 

One future topic is the application to the problem of ground state/eigenstate preparation. We know that the ITE operator of a given Hamiltonian $\mathcal{H}$ is useful in deriving its ground state since the overlaps regarding the excited states tend to zero exponentially fast as the time-like parameter increases. 
Different from solving the numerical solutions to differential equations, for the ground state preparation problem, the evolution of the excited states is not necessarily to obey the ITE operator, and hence, one can use an energy shift technique to avoid the rapid decay of the success probability \cite{Nishi.2023, Kosugi.2023} in the original approximate PITE. Although the AAPITE works for such a problem, it remains open whether it would outperform the original approximate PITE \cite{Kosugi.2022} and its derivative \cite{Nishi.2023} in such an application. 

Another future work is the insight comparisons to the QSVT-based QLSA \cite{Krovi2023}, the LCU-based QLSA \cite{Childs.2017}, and the recently proposed Schr\"odingerisation method \cite{Jin.2022, Jin.2023}. The explicit gate count and circuit depth should be compared in practical settings. This is indispensable if we intend to find which is the best quantum solver for the linear/nonlinear PDEs or how to choose a suitable one for a given problem.  


\section*{Acknowledgements}
This work was supported by Japan Society for the Promotion of Science (JSPS) KAKENHI under Grant-in-Aid for Scientific Research No.21H04553, No.20H00340, and No.22H01517. 
This work was partially supported by the Center of Innovations for Sustainable Quantum AI (JST Grant number JPMJPF2221).

\appendix
\section{Previous approximate PITE [Kosugi et al. 2022]}
\label{appA}

We provide detailed analysis for an approximate PITE that is equivalent to the one in \cite{Kosugi.2022}. 
Take $\Theta := \arccos(m_0 \mathrm{exp}\left(-\Delta\tau \mathcal{H}\right))$ with an arbitrarily fixed parameter $m_0\in (0,1)$. We denote the real-valued function $g_0(x) := \arccos(m_0\mathrm{exp}\left(-x\right))$, $x\ge 0$. 
By the Taylor expansion around $x_0\ge 0$ up to first order, we obtain 
\begin{align*}
g_0(x) = g_0(x_0) + g_0^\prime(x_0)(x-x_0) + O(|x-x_0|^2). 
\end{align*}
By noting 
$$
\lim_{x\to 0} g_0(x) = \arccos(m_0) =: \theta_0, \quad \lim_{x\to 0} g_0^\prime(x) = \frac{m_0}{\sqrt{1-m_0^2}} =: s_0,
$$
and taking $x_0=0$, we reach the approximation
\begin{align*}
\Theta = g_0(\Delta\tau \mathcal{H}) = \theta_0 I + s_0\Delta\tau \mathcal{H} + O((\Delta\tau)^2) \approx \theta_0 I + s_0\Delta\tau \mathcal{H}.
\end{align*}
In other words, 
\begin{align*}
m_0\mathrm{exp}\left(-\Delta\tau \mathcal{H}\right) = \cos(g_0(\Delta\tau \mathcal{H}))\approx \cos(\theta_0 I + s_0\Delta\tau \mathcal{H}). 
\end{align*}
Here, we mention that $m_0\not= 1$, otherwise, the above Taylor expansion does not converge. 
Since 
\begin{align*}
U_{\text{RTE}}(\pm (\theta_0 I + s_0\Delta\tau \mathcal{H})) = U_{\text{RTE}}(\pm\theta_0 I) U_{\text{RTE}}(\pm s_0\Delta\tau \mathcal{H}), 
\end{align*}
and the four RTE operators $U_{\text{RTE}}(\pm\theta_0 I)$, $U_{\text{RTE}}(\pm s_0\Delta\tau \mathcal{H})$ are commutable, from Fig.~\ref{sec2:fig1}, we obtain the quantum circuit in Fig.~\ref{appC2:fig3}, in which we used the relation
$$
\mathrm{exp}\left(\mathrm{i}\theta_0\right)\ket{0}\bra{0} + \mathrm{exp}\left(-\mathrm{i}\theta_0\right)\ket{1}\bra{1} = R_z(-2\theta_0).
$$
\begin{figure}
\centering
\resizebox{12cm}{!}{
\includegraphics[keepaspectratio]{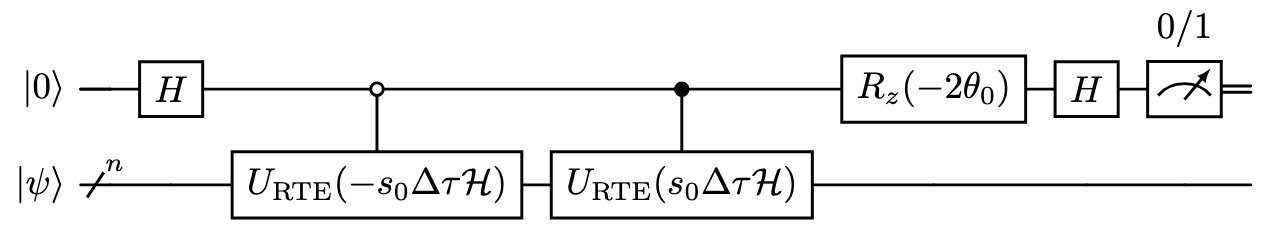}
}
\caption{A quantum circuit for approximate PITE. This quantum circuit is equivalent to Fig. 1(b) in \cite{Kosugi.2022}. }
\label{appC2:fig3}
\end{figure}
In the sense of possibly applying two single-qubit gates $Q$ and $Q^\dag$ to the ancillary qubit, the proposed circuit in Fig.~\ref{appC2:fig3} is equivalent to Fig. 1(b) in \cite{Kosugi.2022}. 

Before estimating the error and the success probability, we mention the quantum resources used in a single PITE step. By the construction, it is clear that the main request of quantum resources comes from the part of controlled RTE operator of a given Hermite operator $\mathcal{H}$. In the practical applications with the first-quantized formulation, admitting the error from the Suzuki-Trotter expansion, it is well-known that such controlled RTE operators can be implemented efficiently in gate complexity $O(\mathrm{polylog}\, N)$ \cite{Kosugi.2022}. Considering the dependence of $N$, this is an exponential improvement compared to the exact implementation mentioned in Sect.~\ref{sec:2}. 

\subsection{Error estimate}
\label{subsec:A-1}

For simplicity, we consider a $K$-step approximate PITE for $K\in \mathbb{N}$ and assume that the eigenvalues of $\mathcal{H}$: $\{\lambda_k\}_{k=0}^{N-1}$ are non-negative without loss of generality. 
As long as we know the smallest eigenvalue of $\mathcal{H}$: $\lambda_{min}$, we can shift the operator $\mathcal{H}$ by $\lambda_{\ast}I$ for any $\lambda_{\ast}\le \lambda_{min}$ and consider the shifted operator $\widetilde{\mathcal{H}} := \mathcal{H} - \lambda_{\ast} I$. Thus, 
$$
\mathrm{exp}\left(-\Delta\tau \mathcal{H}\right) = \mathrm{exp}\left(-\Delta\tau \lambda_{\ast}\right) \mathrm{exp}\left(-\Delta\tau \widetilde{\mathcal{H}}\right),
$$
and we can moreover assume that the smallest eigenvalue of the shifted operator is exactly zero, provided that we know exactly the smallest eigenvalue of the original operator $\mathcal{H}$. 

For an approximate PITE circuit, we have error due to the approximation. Here, we introduce the error after normalization, which is defined by
\begin{align*}
\mathrm{Err}_2 := \left\|\frac{\cos^K(\theta_0+s_0\Delta\tau \mathcal{H})\ket{\psi}}{\|\cos^K(\theta_0+s_0\Delta\tau \mathcal{H})\ket{\psi}\|}-\frac{\mathrm{exp}\left(-K\Delta\tau \mathcal{H}\right)\ket{\psi}}{\|\mathrm{exp}\left(-K\Delta\tau \mathcal{H}\right)\ket{\psi}\|}\right\|,
\end{align*}
where $\ket{\psi}$ is an initial state. 
Moreover, we estimate by the triangle inequality $\|a-b\|\le \|a\|+\|b\|$ that
\begin{align*}
\mathrm{Err}_2 &\le \left\|\frac{\cos^K(\theta_0+s_0\Delta\tau \mathcal{H})\ket{\psi}}{\|\cos^K(\theta_0+s_0\Delta\tau \mathcal{H})\ket{\psi}\|}-\frac{m_0^K\mathrm{exp}\left(-K\Delta\tau \mathcal{H}\right)\ket{\psi}}{\|\cos^K(\theta_0+s_0\Delta\tau \mathcal{H})\ket{\psi}\|}\right\| \\
&\quad + \left\|\frac{m_0^K\mathrm{exp}\left(-K\Delta\tau \mathcal{H}\right)\ket{\psi}}{\|\cos^K(\theta_0+s_0\Delta\tau \mathcal{H})\ket{\psi}\|}-\frac{\mathrm{exp}\left(-K\Delta\tau \mathcal{H}\right)\ket{\psi}}{\|\mathrm{exp}\left(-K\Delta\tau \mathcal{H}\right)\ket{\psi}\|}\right\| \\
&\le 2\frac{\|\cos^K(\theta_0+s_0\Delta\tau \mathcal{H})\ket{\psi}-m_0^K\mathrm{exp}\left(-K\Delta\tau \mathcal{H}\right)\ket{\psi}\|}{\|\cos^K(\theta_0+s_0\Delta\tau \mathcal{H})\ket{\psi}\|}.
\end{align*}
The denominator is the square root of the success probability: 
$$
\mathbb{P}_2(\ket{0}) := \|\cos^K(\theta_0+s_0\Delta\tau \mathcal{H})\ket{\psi}\|^2.  
$$
If we denote the numerator by 
$$
\widetilde{\mathrm{Err}_2} := \|\cos^K(\theta_0+s_0\Delta\tau \mathcal{H})\ket{\psi}-m_0^K\mathrm{exp}\left(-K\Delta\tau \mathcal{H}\right)\ket{\psi}\|,
$$
then again by the triangle inequality $\|a-b\|\ge \|a\|-\|b\|$, we obtain
\begin{align}
\label{appC2:eq1}
\|m_0^K\mathrm{exp}\left(-K\Delta\tau \mathcal{H}\right)\ket{\psi}\| - \widetilde{\mathrm{Err}_2} 
\le \|\cos^K(\theta_0+s_0\Delta\tau \mathcal{H})\ket{\psi}\| 
\le \|m_0^K\mathrm{exp}\left(-K\Delta\tau \mathcal{H}\right)\ket{\psi}\| + \widetilde{\mathrm{Err}_2}.
\end{align}
Thus, we find that the denominator goes to $m_0^{K}\|\mathrm{exp}\left(-K\Delta\tau \mathcal{H}\right)\ket{\psi}\|$ as $\widetilde{\mathrm{Err}_2}\to 0$, and hence the success probability 
\begin{align*}
\mathbb{P}_2(\ket{0}) \to m_0^{2K}\|\mathrm{exp}\left(-K\Delta\tau \mathcal{H}\right)\ket{\psi}\|^2 \quad \text{as }\ \widetilde{\mathrm{Err}_2}\to 0. 
\end{align*}
Therefore, estimating the upper bound of the numerator $\widetilde{\mathrm{Err}_2}$ is essential. 
Here, we adopt the non-decreasing ordering $\lambda_0\le \lambda_1 \le \cdots \le \lambda_{N-1}$ of the eigenvalues, and we denote the corresponding eigenfunctions by $\{\psi_k\}_{k=0}^{N-1}$, which forms a complete orthonormal basis of $\mathbb{C}^N$. 
Let $N_\delta(\psi)$ and $\mathcal{K}_\delta(\psi)$ be defined by Eqs.~\eqref{sec2:eq1a} and \eqref{sec2:eq1b}, respectively. 
Then, we can calculate
\begin{align*}
&\quad \left\|\cos^K(\theta_0+s_0\Delta\tau \mathcal{H})\ket{\psi}-m_0^K\mathrm{exp}\left(-K\Delta\tau \mathcal{H}\right)\ket{\psi}\right\|^2 \\
&= \sum_{k=0}^{N-1} |\langle \psi_k|\psi\rangle|^2 \left|\cos^K(\theta_0+s_0\Delta\tau \lambda_k)-(m_0\mathrm{exp}\left(-\Delta\tau \lambda_k\right))^K\right|^2 \\
&= m_0^{2K}\left(\sum_{k\in \mathcal{K}_{\delta}} + \sum_{k\in [\tilde N]\setminus \mathcal{K}_{\delta}}\right) |\langle \psi_k|\psi\rangle|^2 \left|(\cos(\theta_0+s_0\Delta\tau \lambda_k)/m_0)^K-(\mathrm{exp}\left(-\Delta\tau \lambda_k\right))^K\right|^2 \\
&\le m_0^{2K}\Bigg(\sum_{k\in \mathcal{K}_{\delta}} |\langle \psi_k|\psi\rangle|^2 \Bigg|\left(\cos(\theta_0+s_0\Delta\tau \lambda_k)/m_0-\mathrm{exp}\left(-\Delta\tau \lambda_k\right)\right) \\
&\quad \left(\sum_{j=0}^{K-1}(\cos(\theta_0+s_0\Delta\tau \lambda_k)/m_0)^j(\mathrm{exp}\left(-\Delta\tau \lambda_k\right))^{K-1-j}\right)\Bigg|^2 + 4\delta^2\Bigg)\\
&\le m_0^{2K}\left(K^2 \sum_{k\in \mathcal{K}_{\delta}} |\langle \psi_k|\psi\rangle|^2 \left|\cos(\theta_0+s_0\Delta\tau \lambda_k)/m_0-\mathrm{exp}\left(-\Delta\tau \lambda_k\right)\right|^2 + 4\delta^2\right)\\
&= m_0^{2K}\left(K^2 \sum_{k\in \mathcal{K}_{\delta}} |\langle \psi_k|\psi\rangle|^2 \left|\cos(s_0\Delta\tau\lambda_k)-m_0^{-1}\sqrt{1-m_0^2}\sin(s_0\Delta\tau\lambda_k)-\mathrm{exp}\left(-\Delta\tau \lambda_k\right)\right|^2 + 4\delta^2\right).
\end{align*}
Here, in the second last line, we used an assumption that $\Delta\tau$ is small enough such that $s_0\Delta\tau\lambda_{N_{\delta}}\in [0,\pi/2]$, which indicates a restriction: 
$$
|\cos(\theta_0+s_0\Delta\tau\lambda_k)/m_0|=\left|\cos(s_0\Delta\tau\lambda_k)-\frac{\sqrt{1-m_0^2}}{m_0}\sin(s_0\Delta\tau\lambda_k)\right|\le 1, \quad \text{for all } k=0,\ldots,N_{\delta}.
$$
Furthermore, in the first inequality, we used an assumption that $|\cos(\theta_0+s\Delta\tau \lambda_k)/m_0|\le 1$ for $k\in [\tilde N]\setminus \mathcal{K}_{\delta}$. Such an assumption yields a hidden constraint among the choices of $m_0$, $\Delta\tau$ and the spectrum of the Hermite operator.  
Note that the trigonometric functions and the exponential functions are smooth in $\mathbb{R}$. By employing the Taylor's theorem with the Lagrange form of the remainder, we obtain
\begin{align*}
&\quad \cos(s_0\Delta\tau\lambda_k)-\frac{\sqrt{1-m_0^2}}{m_0}\sin(s_0\Delta\tau\lambda_k)-\mathrm{exp}\left(-\Delta\tau \lambda_k\right) \\
&= \lambda_k^2 (\Delta\tau)^2 \left(-s_0^2 \cos(s_0\lambda_k\xi)+s_0\sin(s_0\lambda_k\xi)-\mathrm{exp}\left(-\lambda_k\xi\right)\right)/2,
\end{align*}
for some $0\le\xi\le \Delta\tau$. 
Therefore, we obtain
\begin{align*}
&\quad \left\|\cos^K(\theta_0+s_0\Delta\tau \mathcal{H})\ket{\psi}-m_0^K\mathrm{exp}\left(-K\Delta\tau \mathcal{H}\right)\ket{\psi}\right\|^2 \\
&\le m_0^{2K}\left(\left(\frac{K (s_0^2+1)(\Delta\tau)^2}{2}\right)^2 \sum_{k\in \mathcal{K}_{\delta}} \lambda_k^4 |\langle \psi_k|\psi\rangle|^2 + 4\delta^2\right).
\end{align*}
Here, we used that $\cos(s_0\lambda_k\xi),\sin(s_0\lambda_k\xi)>0$ and that $\sin(s_0\lambda_k\xi)$ is close to zero for all $k=0,\ldots,N_{\delta}$. By setting $T=K\Delta\tau$ and denoting
$$
C_{m_0}(\ket{\psi},\delta) := (1-m_0^2)^{-1}\left(\sum_{k\in \mathcal{K}_{\delta}} \lambda_k^4 |\langle \psi_k|\psi\rangle|^2\right)^{1/2}, 
$$
we reach the estimate
\begin{align}
\label{appC2:eq-err0}
\mathrm{Err}_2 \le m_0^K(C_{m_0}(\ket{\psi},\delta) T\Delta\tau +4\delta)/\sqrt{\mathbb{P}_2(\ket{0})}
\le m_0^K\left((1-m_0^2)^{-1}\lambda_{N_{\delta}}^2 T\Delta\tau +4\delta\right)/\sqrt{\mathbb{P}_2(\ket{0})}, 
\end{align}
for any $\delta\in [0,1)$. By Eq.~\eqref{appC2:eq1}, we have
\begin{align}
\label{appC2:eq-err1}
\mathrm{Err}_2 \le \frac{(1-m_0^2)^{-1}\lambda_{N_{\delta}}^2 T\Delta\tau +4\delta}{\|\mathrm{exp}\left(-T \mathcal{H}\right)\ket{\psi}\|-\left((1-m_0^2)^{-1}\lambda_{N_{\delta}}^2 T\Delta\tau +4\delta\right)}, 
\end{align}
for any $\delta\in [0,1)$. 
For any $\varepsilon>0$, we can choose
\begin{align}
\label{appC2:eq2}
\Delta\tau \le \frac{4(1-m_0^2)\delta}{\lambda_{N_\delta}^2 T}, \quad 
\delta \le \frac{\|\mathrm{exp}\left(-T \mathcal{H}\right)\ket{\psi}\|\varepsilon}{4(2+\varepsilon)},
\end{align}
and obtain $\mathrm{Err}_2 \le \varepsilon$. 
If we focus on the order with respect to $\varepsilon$ and $m_0$, then we derive
$\delta = O(\varepsilon)$, and thus, 
\begin{equation}
\label{appC2:eq3}
\Delta\tau = O\left(\varepsilon(1-m_0^2)\lambda_{N_{\delta(\varepsilon)}}^{-2}\right), \quad \text{or equivalently,}\quad K=\Omega\left(\varepsilon^{-1}(1-m_0^2)^{-1}\lambda_{N_{\delta(\varepsilon)}}^2\right).
\end{equation}
As we can see by the definition, $\lambda_{N_{\delta(\varepsilon)}}$ tends to $\lambda_{N-1}(=\lambda_{\text{max}})$ as $\varepsilon$ goes to $0$. 

\subsection{Success probability}
\label{subsec:A-2}

According to Eq.~\eqref{appC2:eq2}, we obtain
$$
\widetilde{\mathrm{Err}_2} \le 4\delta \le \frac{\varepsilon}{2+\varepsilon}\|\mathrm{exp}\left(-T \mathcal{H}\right) \ket{\psi}\|. 
$$
Together with Eq.~\eqref{appC2:eq1}, we estimate the success probability as follows:
$$
\left(\frac{2}{2+\varepsilon}\right)^2 m_0^{2K}\|\mathrm{exp}\left(-T \mathcal{H}\right) \ket{\psi}\|^2 \le \mathbb{P}_2(\ket{0}) \le \left(\frac{2+2\varepsilon}{2+\varepsilon}\right)^2  m_0^{2K}\|\mathrm{exp}\left(-T \mathcal{H}\right) \ket{\psi}\|^2. 
$$
By noting that $\|\mathrm{exp}\left(-T \mathcal{H}\right) \ket{\psi}\|^2$ is independent of $\Delta\tau$ and $K$, and $2/(2+\varepsilon)$ and $(2+2\varepsilon)/(2+\varepsilon)$ are both of order $\Omega(1)$ with respect to $\varepsilon$, we find $\mathbb{P}_2(\ket{0}) = O(m_0^{2K})$ regarding $m_0$ and $K$.  
Since $m_0\in (0,1)$, this is an exponential decay with respect to $K$. Moreover, we have 
$$
m_0^{2K} \le m_0^{\frac{2T^2\lambda_{N_{\delta(\varepsilon)}}^2(2+\varepsilon)}{\varepsilon (1-m_0^2)\|\mathrm{exp}\left(-T \mathcal{H}\right)\ket{\psi}\|}} =: p_0(m_0,\varepsilon).
$$
By substituting $m_0^2 = 1-\delta_0$, we have
$$
p_0(m_0,\varepsilon) = (1-\delta_0)^{\frac{T^2\lambda_{N_{\delta(\varepsilon)}}^2(2+\varepsilon)}{\varepsilon \delta_0\|\mathrm{exp}\left(-T \mathcal{H}\right)\ket{\psi}\|}} \le \mathrm{exp}\left(-\frac{T^2\lambda_{N_{\delta(\varepsilon)}}^2(2+\varepsilon)}{\varepsilon\|\mathrm{exp}\left(-T \mathcal{H}\right)\ket{\psi}\|}\right) = O\left(\mathrm{exp}\left(-\varepsilon^{-1}\right)\right).
$$
Here, we used $\sup_{0<\delta_0<1} (1-\delta_0)^{1/\delta_0} = 1/e$. 
This implies that even if we scale $m_0$ according to $K$, the success probability admits an exponential decay for the parameter $\varepsilon^{-1}$. 

This rapid vanishing of the success probability prevents us from obtaining an efficient quantum circuit for small $\varepsilon$ because we need to repeat the approximate PITE circuit for more and more times as $\varepsilon$ becomes smaller. 
The essential reason why the success probability vanishes rapidly lies in the restriction that $m_0$ is strictly smaller than $1$, and $K$ goes to infinity as $m_0$ tends to $1$ (see Eq.~\eqref{appC2:eq3}). 
We remark that although the restriction on the time step $\Delta\tau$ is relieved by using the VS-APITE in \cite{Nishi.2023}, 
the success probability is still small according to the choice $m_0<1$, which is confirmed numerically in Sect.~\ref{subsec:5-3}. Besides, Nishi et al. \cite{Nishi.2023} used a constant energy shift to improve the success probability to derive the ground state. 
Rigorously speaking, the shifted operator realizes a cosine function peaked at the ground state energy instead of the original ITE operator. So such a shifted operator can not be applied for the problem of approximating the ITE operator itself.  
In Sect.~\ref{subsec:2-1}, we proposed another approximate PITE (AAPITE) circuit to avoid this crucial problem. 


\section{Mathematical theory of the real-space grid method}
\label{appB}

We provide the derivation of the $N$-discretized matrix from a Hamiltonian operator in an infinite function space. This idea is not new, but we do rigorous estimates and provide it for the sake of completeness. Since the discussions for general time evolution equations are similar, for simplicity, we illustrate the discretization of the real-space grid method using the following Schr\"odinger equation: 
$$
\mathrm{i} \partial_t u(\mathbf{x},t) = -a \nabla^2 u(\mathbf{x},t) + V(\mathbf{x},t)u(\mathbf{x},t), \quad u(\mathbf{x},0) = u_0(\mathbf{x}), \quad \mathbf{x}\in [0,L]^{Md}, \ t>0,
$$
with periodic boundary condition where $a>0$ is a scaling parameter, 
$L$ is the length of the simulation cell, $d$ is the spatial dimension, 
$M$ is the number of particles and the potential $V$ is assumed to be sufficiently smooth. 
Usually we denote the Hamiltonian operator $\hat H =-a \nabla^2 + V$ and $u$ is the wave function for some one-/multi-particle systems. 

Although one can discretize the space by finite difference method (or apply finite element method if the spatial domain is irregular), here we consider the well-known method based on the Fourier series expansion and the so-called Galerkin approximation. 
The idea of Galerkin approximation is to project the solutions to differential equations from an infinite function space into an $N$ finite subspace, and then show the projected solution tends to the exact solution in some suitable function space as the parameter $N$ tends to infinity. 
Let the solution to the differential equation lie in a function space $X$ (we usually assume $X$ is a Hilbert space for the convenience of discussion), 
which admits a countable orthonormal basis $\{\phi_n(x)\}_{n=1}^\infty$. 
Then, we can represent the solution to the partial differential equation by 
an infinite linear combination of the basis with time-varying coefficients. 
That is, we have $\{\hat{u}_k(t)\}_{k=1}^\infty$ such that 
$$
u(x,t) = \sum_{k=1}^\infty \hat{u}_k(t) \phi_k(x).  
$$
In this section, we consider the Schr\"odinger equation with periodic boundary condition, for which a natural basis is $\{\phi_{\mathbf{k}}(\mathbf{x}) = \frac{1}{L^{D/2}} e^{\mathrm{i}\frac{2\pi}{L}k_1 x_1} \ldots e^{\mathrm{i}\frac{2\pi}{L}k_{D} x_{D}}\}_{-\infty}^\infty \subset L^2([0, L]^{D}; \mathbb{C})$. 
Here and henceforth, we let $D = Md$ for simplicity. 
In other words, we apply the Fourier series expansion of the solution as
\begin{equation}
\label{appB:eq1}
{u}(\mathbf{x},t) = \sum_{k_1,\ldots,k_{D}=-\infty}^\infty {\hat u}_{\mathbf{k}}(t) \phi_{\mathbf{k}}(\mathbf{x}).
\end{equation}
For an arbitrarily fixed even integer $N_0\in \mathbb{N}$, we let $N=N_0^D$ and define $\mathcal{V}_N := \mathrm{span}\{\phi_{\mathbf{k}}, \mathbf{k}\in [N_0]^D = \{-N_0/2, \ldots, N_0/2-1\}^D\}$, which is an $N$-dimensional subspace of $L^2([0,L]^D; \mathbb{C})$. 
Moreover, we denote a projection operator $P_{\mathcal{V}_N}: L^2([0,L]^D; \mathbb{C}) \to \mathcal{V}_N$ defined by
$$
P_{\mathcal{V}_N}(f) := \sum_{\mathbf{k}\in [N_0]^D} (\phi_{\mathbf{k}}, f) \phi_{\mathbf{k}}, \quad f\in L^2([0,L]^D; \mathbb{C})
$$
where $(\cdot,\cdot)$ denotes the inner product in $L^2([0,L]^D; \mathbb{C})$: 
$$
(f,g) := \int_{[0,L]^D} \bar f(x) g(x) dx.
$$ 
Note that $\{\phi_{\mathbf{k}}(\mathbf{x})\}_{\mathbf{k}\in [N]^D}$ constructs an orthonormal basis, and thus $(\phi_{\mathbf{k}}, \phi_{\mathbf{k^\prime}}) = \delta_{\mathbf{k}\mathbf{k^\prime}} := \delta_{k_1 k_1^\prime}\ldots \delta_{k_D k_D^\prime}$ 
where $\delta_{ij}$ is the Kronecker delta notation. 
By inserting Eq.~\eqref{appB:eq1} into the Schr\"odinger equation, we obtain
\begin{equation}
\label{appB:eq2}
\mathrm{i} \sum_{k_1,\ldots,k_D=-\infty}^\infty \partial_t {\hat u}_{\mathbf{k}}(t)\phi_{\mathbf{k}}(\mathbf{x}) = \sum_{k_1,\ldots,k_D=-\infty}^\infty a(2\pi/L)^2|\mathbf{k}|^2 {\hat u}_\mathbf{k}(t)\phi_{\mathbf{k}}(\mathbf{x}) + V(\mathbf{x},t)\sum_{k_1,\ldots,k_D=-\infty}^\infty {\hat u}_{\mathbf{k}}(t)\phi_{\mathbf{k}}(\mathbf{x}).
\end{equation}
Then, we can multiply Eq.~\eqref{appB:eq2} by $\phi_{\mathbf{k}}$ for each $\mathbf{k}\in [N_0]^D$ separately and integrate over $[0,L]^D$ to obtain
\begin{align}
\label{appB:eq4}
\mathrm{i} \partial_t {\hat u}_{\mathbf{k}}(t) 
&= a (2\pi/L)^2|\mathbf{k}|^2 {\hat u}_\mathbf{k}(t) + \sum_{k_1^\prime,\ldots,k_D^\prime=-\infty}^\infty (\phi_{\mathbf{k}}, V(t) \phi_{\mathbf{k^\prime}}) {\hat u}_{\mathbf{k}^\prime}(t) 
\quad \text{for any }\mathbf{k}\in [N_0]^D,
\end{align}
with the initial condition
\begin{align*}
{\hat u}_{\mathbf{k}}(0) = (\phi_{\mathbf{k}}, {u}_0).
\end{align*}
Multiplying Eq.~\eqref{appB:eq4} by $\phi_{\mathbf{k}}$ and taking summation over $\mathbf{k}\in [N_0]^D$, we find that the approximate solution defined by
$$
u_N(\mathbf{x},t) := P_{\mathcal{V}_N}(u)(\mathbf{x},t) = \sum_{\mathbf{k}\in [N_0]^D} {\hat u}_{\textbf{k}}(t) \phi_{\mathbf{k}}(\mathbf{x})
$$
solves
\begin{align}
\label{appB:eq5}
\mathrm{i} \partial_t {u}_{N}(\mathbf{x}, t) 
&= -a \nabla^2 {u}_N(\mathbf{x}, t) + P_{\mathcal{V}_N}(V{u})(\mathbf{x},t),
\end{align}
with the initial condition
\begin{align*}
{u}_{N}(\mathbf{x}, 0) = P_{\mathcal{V}_N}(u_0)(\mathbf{x}).
\end{align*}
By introducing the residue 
$$
R_N(\mathbf{x},t) := P_{\mathcal{V}_N}(V u)(\mathbf{x},t)-V(\mathbf{x},t)P_{\mathcal{V}_N}(u)(\mathbf{x},t), 
$$ 
we rewrite Eq.~\eqref{appB:eq5} by
\begin{align}
\label{appB:eq5b}
\mathrm{i} \partial_t {u}_{N}(\mathbf{x}, t) 
&= -a \nabla^2 {u}_N(\mathbf{x}, t) + V(\mathbf{x},t) u_N(\mathbf{x},t) + R_N(\mathbf{x},t).
\end{align}
It is readily to see that if the residue $R_N$ vanishes, then the approximate solution exactly solves the same Schr\"odinger equation with a truncated initial condition. Then, the error between the exact solution ${u}$ and the approximate solution ${u}_N$ only depends on the truncation error of the Fourier series, which is a well-discussed classical issue \cite{Lasser.2020}. 
However, this is not true for most cases in practical applications. 
Another important remark is that the residue does not lie in $\mathcal{V}_N^\perp$ in general
due to the presence of the space-dependent potential $V$. 

Eq.~\eqref{appB:eq5b} is still a partial differential equation, which is not easy to solve numerically. Indeed, by Galerkin method, 
one intends to solve an approximate ordinary differential equation instead. 
Here, we review two methods to numerically solve the above equation. 

\subsection{Method 1: Standard spectral method}

We go back to Eq.~\eqref{appB:eq4} and truncate the infinite series in the second term on the right-hand side, that is, we keep only the summation over $\mathbf{k}\in [N_0]^D$. 
Thus, we construct an approximate problem to Eq.~\eqref{appB:eq4}: 
\begin{align}
\label{appB:eq6}
\mathrm{i} \partial_t c_{\mathbf{k}}(t) 
&= a (2\pi/L)^2|\mathbf{k}|^2 c_\mathbf{k}(t) + \sum_{\mathbf{k}^\prime\in [N_0]^D} (\phi_{\mathbf{k}}, V(t) \phi_{\mathbf{k^\prime}}) c_{\mathbf{k}^\prime}(t) 
\quad \text{for any }\mathbf{k}\in [N_0]^D,
\end{align}
with 
\begin{align*}
c_{\mathbf{k}}(0) = (\phi_{\mathbf{k}}, {u}_0).
\end{align*}
By introducing the matrices: 
\begin{align*}
&T_N = a(2\pi/L)^2 \sum_{\mathbf{k}\in [N_0]^D} |\mathbf{k}|^2 \ket{\tilde k_1\ldots \tilde k_D} \bra{\tilde k_1\ldots \tilde k_D} 
= a(2\pi/L)^2 \sum_{\mathbf{k}\in [N_0]^D} |\mathbf{k}|^2 \ket{\mathbf{\tilde k}} \bra{\mathbf{\tilde k}}, \\
&\hspace{10cm} \quad \tilde k = k+N_0/2, \\
&V_N = \sum_{\mathbf{k}, \mathbf{k}^\prime\in [N_0]^D} (\phi_{\mathbf{k}}, V(t) \phi_{\mathbf{k^\prime}}) \ket{\mathbf{\tilde k}} \bra{\mathbf{\tilde k}^\prime}, 
\end{align*}
and denoting the vector $\mathbf{c} = \{c_{\mathbf{k}}\} \subset \mathbb{C}^{N}$, we arrive at a matrix equation: 
\begin{align}
\label{appB:eq7}
\mathrm{i} \partial_t \mathbf{c}(t) = T_N \mathbf{c}(t) + V_N(t) \mathbf{c}(t) =: H_N(t) \mathbf{c}(t).
\end{align}
Therefore, we can write the solution to Eq.~\eqref{appB:eq6} explicitly by
$$
\mathbf{c}(t) = e^{-\mathrm{i}\int_0^t H_N(s)ds} \mathbf{c}(0),
$$
and by 
$$
\mathbf{c}(t) = e^{-\mathrm{i} H_N t} \mathbf{c}(0),
$$
provided that $H_N$ is $t$-independent. 
Hence we obtain a numerical solution:
$$
y_N(\mathbf{x},t) := \sum_{\mathbf{k}\in [N_0]^D} c_{\mathbf{k}}(t) \phi_{\mathbf{k}}(\mathbf{x}),
$$
whose error depends on $(\phi_{\mathbf{k}}, R_N(t))$, $\mathbf{k}\in [N_0]^D$. More precisely, if $H_N$ is t-independent, then we have
\begin{align*}
\|y_N(t)-u_N(t)\|_{L^2}^2 &= \sum_{\mathbf{k}\in [N_0]^D} |c_{\mathbf{k}}(t)-{\hat u}_{\mathbf{k}}(t)|^2 
\le \int_0^t \sum_{\mathbf{k}\in [N_0]^D}\left| \sum_{\mathbf{k}^\prime\notin [N_0]^D} (\phi_{\mathbf{k}}, V\phi_{\mathbf{k}^\prime}) {\hat u}_{\mathbf{k}^\prime}(s)\right|^2 ds \\
&= \int_0^t \|V (u(s)-u_N(s))\|_{L^2}^2 ds.  
\end{align*}
Thus, 
$$
\|y_N(t)-u_N(t)\|_{L^2} \le \sqrt{t} |V|_\infty \max_{0\le s\le t} \left\|u(s)-P_{\mathcal{V}_N}(u)(s)\right\|_{L^2}. 
$$
Therefore, roughly speaking the error depends linearly on the $\ell^2$-norm of the Fourier coefficients of $u(s)$ beyond $\mathbf{k}\in [N_0]^D$ (the truncation error of the Fourier series), whose estimate can be found in e.g. Lasser et al. \cite{Lasser.2020}, on the maximum norm of the given potential $V$ and on the square root of the simulation time $t^{1/2}$. 

\subsection{Method 2: Real-space grid method}
\label{subsec:appB-2}

An alternative way is to take a sufficient number of nodes in Eq.~\eqref{appB:eq5b}, which is also called the FSM in the literature \cite{Childs.2022}. 
One applies test functions of delta functions concerning some grid points $\mathbf{x}\in [0, L]^D$ on both sides of Eq.~\eqref{appB:eq5b}. 
By introducing a uniform grid in each spatial dimension: 
$$
p_j = j L/N_0,\quad j = 0, \ldots, N_0, 
$$
and denoting a set $P_{N_0} = \{p_j;\ j=0,\ldots,N_0-1\}$, we take $\mathbf{x}=\mathbf{p}_{\mathbf{l}}\in (P_{N_0})^D$ and obtain
\begin{align}
\label{appB:eq8a}
\mathrm{i} \partial_t {u}_{N}(\mathbf{p}_{\mathbf{l}}, t) 
&= -a \nabla^2 {u}_N(\mathbf{p}_{\mathbf{l}}, t) + V(\mathbf{p}_{\mathbf{l}},t) u_N(\mathbf{p}_{\mathbf{l}},t) + R_N(\mathbf{p}_{\mathbf{l}},t),
\end{align}
where $\mathbf{l} = (l_1, \ldots, l_D)^\mathrm{T}\in \{0,\ldots,N_0-1\}^D := [\tilde N_0]^D$ and $\mathbf{p}_{\mathbf{l}} = (p_{l_1}, \ldots, p_{l_D})^\mathrm{T}$. Now we construct an approximate problem to the above equations:
\begin{align}
\label{appB:eq8}
\mathrm{i} \sum_{\mathbf{k}\in [N_0]^D} \partial_t c_{\mathbf{k}}(t) \phi_{\mathbf{k}}(\mathbf{p}_{\mathbf{l}}) 
&= \sum_{\mathbf{k}\in [N_0]^D} a (2\pi/L)^2|\mathbf{k}|^2 c_{\mathbf{k}}(t) \phi_{\mathbf{k}}(\mathbf{p}_{\mathbf{l}}) + V(\mathbf{p}_{\mathbf{l}},t) \sum_{\mathbf{k}\in [N_0]^D} c_{\mathbf{k}}(t) \phi_{\mathbf{k}}(\mathbf{p}_{\mathbf{l}}),
\end{align}
with 
\begin{align*}
c_{\mathbf{k}}(0) = (\phi_{\mathbf{k}}, {u}_0).
\end{align*}
Again we introduce some notations in the fashion of quantum mechanics:
\begin{align*}
&\mathbf{c}(t) = \sum_{\mathbf{k}\in [N_0]^D} c_{\mathbf{k}}(t) \ket{\mathbf{\tilde k}},\\
&\Phi_N = \sum_{\mathbf{k}\in [N_0]^D}\sum_{\mathbf{l}\in [\tilde N_0]^D} \phi_{\mathbf{k}}(\mathbf{p}_{\mathbf{l}}) \ket{\mathbf{l}} \bra{\mathbf{\tilde k}}
= \frac{1}{L^{D/2}} \sum_{\mathbf{k}\in [N_0]^D}\sum_{\mathbf{l}\in [\tilde N_0]^D} \mathrm{exp}\left(\mathrm{i}\frac{2\pi}{N_0}\mathbf{k}\cdot\mathbf{l}\right) \ket{\mathbf{l}} \bra{\mathbf{\tilde k}},\\
&D_N = \sum_{\mathbf{k}\in [N_0]^D} a (2\pi/L)^2|\mathbf{k}|^2 \ket{\mathbf{\tilde k}} \bra{\mathbf{\tilde k}}, \\
&V_N(t) = \sum_{\mathbf{l}\in [\tilde N_0]^D} V(\mathbf{p}_{\mathbf{l}},t) \ket{\mathbf{l}} \bra{\mathbf{l}}
= \sum_{\mathbf{l}\in [\tilde N_0]^D} V\left(\frac{L}{N_0}\mathbf{l},t\right) \ket{\mathbf{l}} \bra{\mathbf{l}}.
\end{align*}
Taking all $\mathbf{l}\in [\tilde N_0]^D$ in Eq.~\eqref{appB:eq8}, we arrive at a matrix (evolution) equation:
\begin{align*}
\mathrm{i} \partial_t (\Phi_N \mathbf{c})(t)
&= \Phi_N D_N \mathbf{c}(t) + V_N(t)\Phi_N \mathbf{c}(t),
\end{align*}
with 
\begin{align*}
\mathbf{c}(0) = \sum_{\mathbf{k}\in [N_0]^D}(\phi_{\mathbf{k}}, {u}_0) \ket{\mathbf{\tilde k}}.
\end{align*}
Note that $\Phi_N$ is invertible, and we can construct a modified matrix from $\Phi_N$ as follows:
$$
F_N := \frac{L^{D/2}}{N_0^{D/2}} \Phi_N = \frac{1}{N_0^{D/2}} \sum_{\mathbf{k}\in [N_0]^D}\sum_{\mathbf{l}\in [\tilde N_0]^D} \mathrm{exp}\left(\mathrm{i}\frac{2\pi}{N_0}\mathbf{k}\cdot\mathbf{l}\right) \ket{\mathbf{l}} \bra{\mathbf{\tilde k}}.
$$
Immediately we have 
$$
F_N^\dag = \frac{1}{N_0^{D/2}} \sum_{\mathbf{k}\in [N_0]^D}\sum_{\mathbf{l}\in [\tilde N_0]^D} \mathrm{exp}\left(-\mathrm{i}\frac{2\pi}{N_0}\mathbf{k}\cdot\mathbf{l}\right) \ket{\mathbf{\tilde k}} \bra{\mathbf{l}},
$$
and $F_N F_N^\dag = F_N^\dag F_N = I$, which implies that $F_N$ is unitary. Then, we rewrite the above matrix equation by
\begin{align}
\label{appB:eq9}
\mathrm{i} \partial_t (F_N \mathbf{c})(t)
&= F_N D_N \mathbf{c}(t) + V_N(t)F_N \mathbf{c}(t)
= (F_N D_N F_N^\dag + V_N(t))(F_N \mathbf{c})(t)
=: H_N(t) (F_N \mathbf{c})(t),
\end{align}
with 
\begin{align}
\label{appB:eq-init}
F_N \mathbf{c}(0) = \sum_{\mathbf{l}\in [\tilde N_0]^D} 
\left(\frac{1}{N_0^{D/2}}\sum_{\mathbf{k}\in [N_0]^D}(\phi_{\mathbf{k}}, {u}_0) \mathrm{exp}\left(\mathrm{i}\frac{2\pi}{N_0}\mathbf{k}\cdot\mathbf{l}\right)\right)\ket{\mathbf{l}}.
\end{align}
Similarly, we derive the explicit solution:
$$
F_N \mathbf{c}(t) = e^{-\mathrm{i}\int_0^t H_N(s)ds} F_N \mathbf{c}(0), 
$$
and 
$$
F_N \mathbf{c}(t) = e^{-\mathrm{i} H_N t} F_N \mathbf{c}(0), 
$$
provided that $H_N$ is $t$-independent. 
$F_N\mathbf{c}(t)$ is the discrete wave function at a fixed time $t$, which we denote by $\sum_{\mathbf{l}\in [\tilde N_0]^D} \psi_{\mathbf{l}}(t) \ket{\mathbf{l}}$. 
Therefore, we construct the approximate (continuous) wave function by 
\begin{align*}
y_N(\mathbf{x},t) &:= \sum_{\mathbf{k}\in [N_0]^D} c_{\mathbf{k}}(t) \phi_{\mathbf{k}}(\mathbf{x}) 
= \frac{1}{N_0^{D/2}} \sum_{\mathbf{k}\in [N_0]^D}\sum_{\mathbf{l}\in [\tilde N_0]^D} \mathrm{exp}\left(-\mathrm{i}\frac{2\pi}{N_0}\mathbf{k}\cdot\mathbf{l}\right) \psi_{\mathbf{l}}(t) \phi_{\mathbf{k}}(\mathbf{x})\\
&= \sum_{\mathbf{l}\in [\tilde N_0]^D} \psi_{\mathbf{l}}(t) \left(\frac{1}{(N_0L)^{D/2}} \sum_{\mathbf{k}\in [N_0]^D} \mathrm{exp}\left(\mathrm{i}\frac{2\pi}{L}(\mathbf{k}\cdot\mathbf{x}-(L/N_0)\mathbf{k}\cdot\mathbf{l})\right)\right)\\
&= \sum_{\mathbf{l}\in [\tilde N_0]^D} \psi_{\mathbf{l}}(t) \left(\frac{1}{(N_0L)^{D/2}} \sum_{\mathbf{k}\in [N_0]^D} \mathrm{exp}\left(\mathrm{i}\frac{2\pi}{L}\mathbf{k}\cdot(\mathbf{x}-\mathbf{p}_{\mathbf{l}})\right)\right)
=: \sum_{\mathbf{l}\in [\tilde N_0]^D} \psi_{\mathbf{l}}(t) g_N(\mathbf{x}; \mathbf{p}_{\mathbf{l}}).
\end{align*}
Moreover, we can rewrite $g_N(\mathbf{x};\mathbf{p}_{\mathbf{l}})$ as the product of $D$ functions: 
\begin{align}
\nonumber
&\quad g_N(\mathbf{x}; \mathbf{p}_{\mathbf{l}}) \\
\nonumber
&= \frac{1}{(N_0L)^{D/2}} \sum_{\mathbf{k}\in [N_0]^D} \mathrm{exp}\left(\mathrm{i}\frac{2\pi}{L}\mathbf{k}\cdot(\mathbf{x}-\mathbf{p}_{\mathbf{l}})\right) \\
\nonumber
&= \frac{1}{(N_0L)^{D/2}} \sum_{k_1,\ldots, k_D\in [N_0]} e^{\mathrm{i}\frac{2\pi}{L}k_1(x_1-p_{l_1})} \ldots e^{\mathrm{i}\frac{2\pi}{L}k_D(x_D-p_{l_D})}\\
\nonumber
&= \prod_{j=1}^D \frac{1}{(N_0L)^{1/2}} \sum_{k_j=-N_0/2}^{N_0/2-1} e^{\mathrm{i}\frac{2\pi}{L}k_j(x_j-p_{l_j})} \\
\nonumber
&= \prod_{j=1}^D \frac{1}{(N_0L)^{1/2}} e^{-\mathrm{i}\frac{\pi}{L}(x_j-p_{l_j})} \sum_{k_j=-N_0/2}^{N_0/2-1} e^{\mathrm{i}\frac{\pi}{L}(2k_j+1)(x_j-p_{l_j})} \\
\nonumber
&= \prod_{j=1}^D \frac{1}{(N_0L)^{1/2}} e^{-\mathrm{i}\frac{\pi}{L}(x_j-p_{l_j})} \sum_{k_j=0}^{N_0/2-1} \left(e^{\mathrm{i}\frac{\pi}{L}(2k_j+1)(x_j-p_{l_j})} + e^{-\mathrm{i}\frac{\pi}{L}(2k_j+1)(x_j-p_{l_j})}\right) \\
\label{appB:eq10}
&= \prod_{j=1}^D \frac{2}{(N_0L)^{1/2}} e^{-\mathrm{i}\frac{\pi}{L}(x_j-p_{l_j})} \sum_{k_j=0}^{N_0/2-1} \cos{\frac{\pi}{L}(2k_j+1)(x_j-p_{l_j})} =: \prod_{j=1}^D h_{N_0}(x_j;p_{l_j}).
\end{align}
The function $h_{N_0}(x;p)$ is called a ``pixel function" in some previous paper, e.g. Chan et al. \cite{Chan.2023}, which serves an approximation of the delta function $\delta(x-p)$ as $N_0\to \infty$. 
If $N_0$ is large enough, then $\psi_\mathbf{l}(t)$ itself is close to $y_N(\mathbf{x}_{\mathbf{l}},t)$, and hence gives a good approximation to $u(\mathbf{x}_{\mathbf{l}},t)$. 
On the contrary, the approximate solution should be $\sum_{\mathbf{\tilde l}\in [\tilde N_0]^D} \psi_{\mathbf{\tilde l}}(t) g_N(\mathbf{x}_{\mathbf{l}}; \mathbf{p}_{\mathbf{\tilde l}})$ instead of $\psi_\mathbf{l}(t)$ for any $\mathbf{l}\in [\tilde N_0]^D$.  

Next, we investigate the error of Method 2. 
Let us compare $c_{\mathbf{k}}(t)$ with ${\hat u}_{\mathbf{k}}(t)$ by using Eqs.~\eqref{appB:eq8a} and \eqref{appB:eq8}. In the case of $t$-independent potential $V$, we can write down the difference explicitly
\begin{align*}
{\hat u}_{\mathbf{k}}(t)-c_{\mathbf{k}}(t) = \frac{L^{D/2}}{N_0^{D/2}} \int_0^t e^{-\mathrm{i}H_N(t-s)} \mathbf{R}_N(s) ds,
\end{align*}
where $\mathbf{R}_N(s) = \{R_N(\mathbf{p}_{\mathbf{l}},t)\}_{\mathbf{l}\in [\tilde N_0]^D}$.
Therefore, we estimate
\begin{align*}
&\|y_N(t)-u_N(t)\|_{L^2}^2 = \sum_{\mathbf{k}\in [N_0]^D} |c_{\mathbf{k}}(t)-{\hat u}_{\mathbf{k}}(t)|^2 
\le \int_0^t \sum_{\mathbf{k}\in [N_0]^D} \frac{L^D}{N_0^D} \sum_{\mathbf{l}\in [\tilde N_0]^D} \left| R_N(\mathbf{p}_{\mathbf{l}},s)\right|^2 ds \\
&\quad \le 2 \int_0^t \sum_{\mathbf{k}\in [N_0]^D} \frac{L^D}{N_0^D} \sum_{\mathbf{l}\in [\tilde N_0]^D} \left|(P_{\mathcal{V}_N}(V u)-V u)(\mathbf{p}_{\mathbf{l}},s)\right|^2 + \left|V(\mathbf{p}_{\mathbf{l}})(u - P_{\mathcal{V}_N}(u))(\mathbf{p}_{\mathbf{l}},s)\right|^2 ds.
\end{align*}
Thus, 
\begin{align*}
\|y_N(t)-u_N(t)\|_{L^2} &\le L^{D/2} \sqrt{2t} \bigg(|V|_\infty \max_{0\le s\le t} \left\|u(s)-P_{V_N}(u)(s)\right\|_{L^\infty} \\
&\quad + \max_{0\le s\le t} \left\|V u(s)-P_{V_N}(V u)(s)\right\|_{L^\infty} \bigg). 
\end{align*}
As a result, roughly speaking the error depends linearly on the pointwise truncation error of the Fourier series expansions of both $V u$ and $u$, which decays exponentially as $N_0$ increases (e.g. Childs et al. \cite{Childs.2022}), and on the square root of the simulation time $t^{1/2}$. 
Compared to Method 1, the error bound here needs additionally a stronger norm and linear dependence on $L^{D/2}$. 
However, we do not know whether this upper bound is sharp or not, 
and actually by noting that $(\frac{L}{N_0})^D \sum_{\mathbf{l}\in [\tilde N_0]^D} |f(\mathbf{p}_{\mathbf{l}})|^2 \to \|f\|_{L^2}^2$ as $N_0 \to \infty$, we find that the ratio of the upper bounds for Method 2 and Method 1 should go to a constant as $N_0\to \infty$. 

\subsection{Comparison between the methods}

We compare the two methods by discussing their possible implementations on quantum computers. 
To investigate the Hamiltonian simulation, we refer to Eqs.~\eqref{appB:eq7} and \eqref{appB:eq9} for Method 1 and Method 2, respectively. The common key point is to implement the unitary operation $e^{-\mathrm{i}H_N \Delta t}$. Henceforth, we use the first-order Suzuki-Trotter formula for simplicity.

As for Method 1, we have 
\begin{align*}
e^{-\mathrm{i}H_N t} \approx (e^{-\mathrm{i}T_N t/K} e^{-\mathrm{i}V_N t/K})^K.
\end{align*}
Here, $T_N$ is a diagonal matrix so that $e^{-\mathrm{i}T_N t/K}$ can be implemented directly by the circuit in \cite{Zhang.2024} or the subroutine for quadratic functions in the phase, both based on $\{\mathrm{CNOT, R_z}\}$ gates. 
On the other hand, $V_N$ is not diagonal, and first one needs to carry out numerical/analytical calculations to obtain all the entries which are integrals over the space domain. This is tough work for large $N_0$ (at least to the authors' best knowledge) and a general potential $V$. 
Moreover, the derivation of the circuit to implement $e^{-\mathrm{i}V_N t/K}$ can be a hard task since $V_N$ is a dense $N\times N$ matrix whose encoding needs the complexity $O(N^2)$ in general. 
If we assume the quantum circuit is prepared, then we arrive at a quantum state that encodes the Fourier coefficient $\mathbf{c}(t)$. 
One way to obtain a discrete wave function as the quantum state is to simply apply a shifted quantum Fourier transform $F_N$ to the input registers after the quantum circuit of RTE. 

As for Method 2, we have
\begin{align*}
e^{-\mathrm{i}H_N t} \approx (e^{-\mathrm{i}F_N D_N F_N^\dag t/K} e^{-\mathrm{i}V_N t/K})^K.
\end{align*}
Since $D_N$ is a diagonal matrix and $F_N$ is unitary, we have $e^{-\mathrm{i}F_N D_N F_N^\dag t/K} = F_N e^{-\mathrm{i} D_N t/K} F_N^\dag$. 
Moreover, we note that $V_N$ in Method 2 is also a diagonal matrix, then we can employ either the subroutine for unitary diagonal matrices or the one for piecewise polynomials in the phase to deal with both $e^{-\mathrm{i} D_N t/K}$ and $e^{-\mathrm{i} V_N t/K}$. 
Therefore, with the help of shifted quantum Fourier transform $F_N$, Method 2 seems to provide a more direct and simple quantum circuit of the RTE operator than Method 1. Furthermore, the quantum state has already encoded a discrete wave function as we discussed above.

\section{Explicit gate implementation for RTE operator related to Laplace operator}
\label{appG}

In this section, we provide the explicit gate implementation of the RTE operator for $\sqrt{2\Delta\tau D_N^{(2)}}$ in Fig.~\ref{sec3:Fig1}, coming from the Laplace operator. We divide it into two cases: (1) $d=1$; (2) $d\ge 2$. 

\noindent \underline{\bf 1D case} By the definition, we find
$$
\sqrt{2\Delta\tau D_N^{(2)}} = \sum_{k=0}^{N-1} \frac{2\pi}{L}\sqrt{2a\Delta\tau}\left|k-\frac{N}{2}\right| \ket{k}\bra{k},
$$
which yields
$$
U_{\text{RTE}}\left(\sqrt{2\Delta\tau D_N^{(2)}}\right) = \sum_{k=0}^{N-1} \mathrm{exp}\left(-\mathrm{i}\frac{2\pi}{L}\sqrt{2a\Delta\tau}\left|k-\frac{N}{2}\right|\right) \ket{k}\bra{k}.
$$
The phase is a piecewise linear function of $k$ with only two intervals. Then, we can apply the LIU circuit in \cite{Huang.2024p} (Sect.~2.2 therein) and further simplify it as follows. In our case, it is sufficient to take $m=1$, and $M=2^m=2$. By defining 
$$
U_M^{(j)} := \mathrm{exp}\left(-\mathrm{i}\frac{N\pi}{2^j L}\sqrt{2a\Delta\tau}\right)\ket{0}\bra{0} + \ket{1}\bra{1}, \quad j=0, 1,\ldots, n-1, 
$$
we can introduce the phase difference operation $W_j$ (i.e., $z$-axis rotation gate in this case): 
$$
W_j = X U_M^{(j)} X U_M^{(j)\dag} = 
R_z\left(\theta_j\right), \quad j=0, 1,\ldots, n-1,
$$
where $\theta_j=-2\pi N\sqrt{2a\Delta\tau}/(2^j L)$, $ j=0, 1,\ldots, n-1$. Therefore, according to Fig.~2 in \cite{Huang.2024p}, we derive the quantum circuit for the RTE operator using one phase gate and $n-1$ controlled $Rz$ gates up to a global phase, see Fig.~\ref{appG:Fig1}.
\begin{figure}[htb]
\centering
\resizebox{11cm}{!}{
\includegraphics[keepaspectratio]{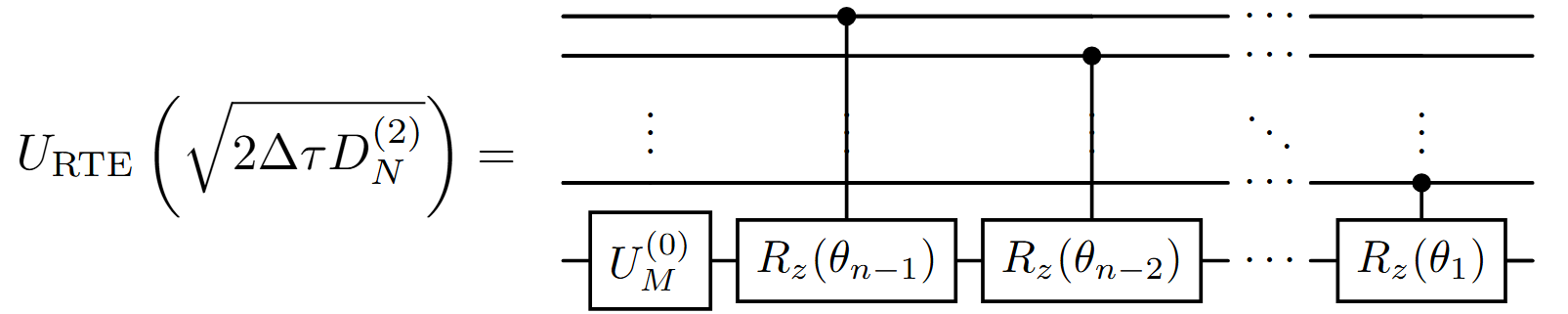}
}
\caption{Quantum circuit for the RTE operator in 1D case. }
\label{appG:Fig1}
\end{figure}

\noindent \underline{\bf Multi-dimensional cases} For $d\ge 2$, the implementation is a bit complicated. We need to introduce a squared distance register as mentioned in Appendix C in \cite{Huang.2024p}. Recalling that
$$
D_N^{(2)} = \sum_{\mathbf{k}\in [\tilde N_0]^d} a (2\pi/L)^2\left|\mathbf{k}-\frac{N_0}{2}\right|^2 \ket{\mathbf{k}} \bra{\mathbf{k}},
$$
we have 
$$
U_{\text{RTE}}\left(\sqrt{2\Delta\tau D_N^{(2)}}\right) = \sum_{\mathbf{k}\in [\tilde N_0]^d} \mathrm{exp}\left(-\mathrm{i}\frac{2\pi}{L}\sqrt{2a\Delta\tau}\left|\mathbf{k}-\frac{N_0}{2}\right|\right) \ket{\mathbf{k}}\bra{\mathbf{k}}.
$$
By using the quantum algorithms for the fixed-point arithmetic operations and the Newton's method, one can construct an approximate distance register $\ket{|\mathbf{k}-N_0/2|}$ \cite{Haner.2018}, and then apply a phase kickback \cite{Jones.2012} to obtain the above RTE operator. However, the required number of ancillary qubits is large (about several hundreds, see Haner \cite{Haner.2018}), and the prefactor of the gate count is also large, which is impossible for the eFTQC.  
Here, we adopt the QFT-based squared distance register in \cite{Huang.2024p}, which needs at most $2n+\lceil\log_2 d\rceil$ ancillary qubits. 
In order to include the set $\{|\mathbf{k}-N_0/2|^2\}_{\mathbf{k}\in [\tilde N_0]}$ whose minimum is $0$ and maximum is $dN_0^2/4$, we need $n_{\text{dis}} = 2n-1+\lceil\log_2 d\rceil$ ancillary qubits for the squared distance register. 
Define the quadratic functions:
$$
f_j(\mathbf{x}) := -\frac{2\pi}{2^j} \sum_{\ell=1}^d\left(x_\ell-\frac{L}{2}\right)^2 \left(\frac{N_0}{L}\right)^2, \quad \mathbf{x}\in [0,L]^d, \ j=1,\ldots, n_{\text{dis}},
$$
and introduce the following polynomial phase gate:
$$
U_{\text{ph}}[f] := \sum_{\mathbf{k}\in [\tilde N_0]} \mathrm{exp}\left(-\mathrm{i} f(k_1 L/N_0, \ldots, k_d L/N_0)\right)\ket{\mathbf{k}}\bra{\mathbf{k}},
$$
for any (multi-dimensional) function $f$ on $[0,L]^d$. According to Fig.~C14 in \cite{Huang.2024p}, the unitary operation:
$$
U_{\text{dis}}(L/2) \ket{0}_{n_{\text{dis}}}\otimes \ket{\mathbf{k}}_{dn} := \ket{|\mathbf{k}-N_0/2|^2}_{n_{\text{dis}}} \otimes \ket{\mathbf{k}}_{dn},
$$
can be constructed by the quantum circuit in Fig.~\ref{appG:Fig2}.
\begin{figure}[htb]
\centering
\resizebox{11cm}{!}{
\includegraphics[keepaspectratio]{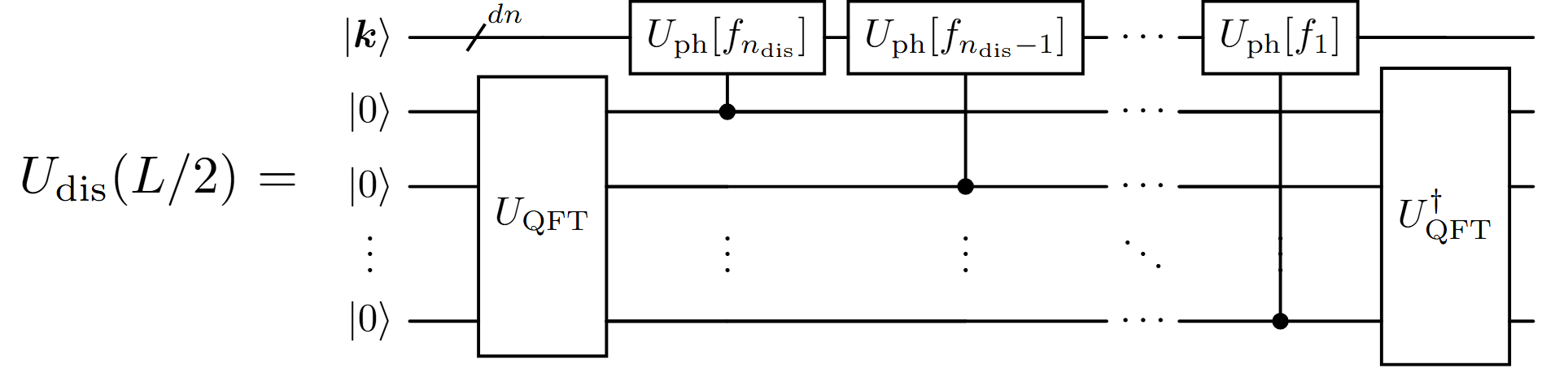}
}
\caption{Quantum circuit for the squared distance register. }
\label{appG:Fig2}
\end{figure}
Here, the gate complexity of a single polynomial phase gate is $O(dn^2)$, and hence the total gate complexity of the squared distance register is $O(dn^3)$. Therefore, with the help of a diagonal unitary operation:
$$
U_{\text{sqrt}} := \sum_{k=0}^{2^{n_{\text{dis}}}-1} \mathrm{exp}\left(-\mathrm{i}\frac{2\pi}{L}\sqrt{2ak\Delta\tau}\right) \ket{k}\bra{k}, 
$$
we can give the quantum circuit for the desired RTE operator as shown in Fig.~\ref{appG:Fig3}. Here, we need the uncomputation of $U_{\text{dis}}(L/2)$, so that the ancillary qubits are returned to the initial zero states and can be reused safely. 
\begin{figure}[htb]
\centering
\resizebox{11cm}{!}{
\includegraphics[keepaspectratio]{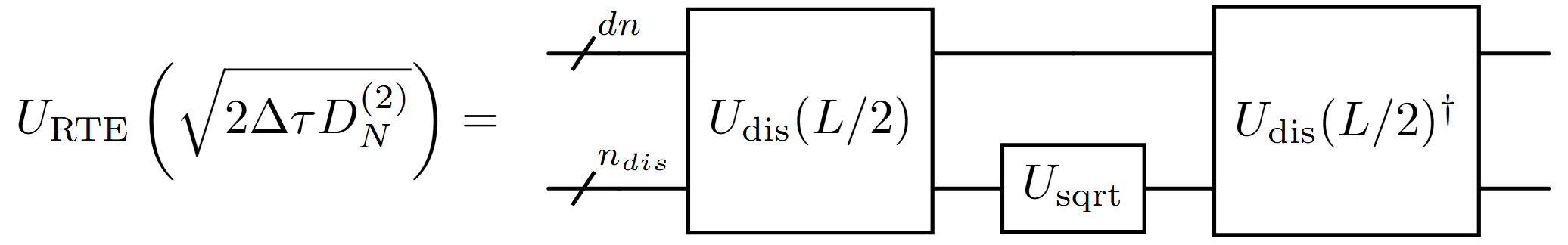}
}
\caption{Quantum circuit for the RTE operator in multi-dimensional cases. }
\label{appG:Fig3}
\end{figure}
The exact implementation of the unitary $U_{\text{sqrt}}$ has a gate complexity $O(dN_0^2)$. Since we consider the eFTQC, we avoid the fixed-point arithmetic operations and adopt the piecewise polynomial approximation in \cite{Huang.2024p}. This gives a gate complexity $O(n^p)$ where $p$ is the polynomial degree. 
In particular, as we find the $k$-th Fourier coefficient of the solution is small for large $k$ in our PDE problem, we can introduce a cut-off of the eigenvalues. For a given positive integer $m<n_{\text{dis}}$, we define the modified diagonal unitary operation:
$$
\tilde U_{\text{sqrt}, m} := \sum_{k=0}^{2^m-1} \mathrm{exp}\left(-\mathrm{i}g_0(k)\right) \ket{k}\bra{k} + \sum_{k=2^m}^{2^{n_{\text{dis}}}-1} \mathrm{exp}\left(-\mathrm{i}g(k)\right) \ket{k}\bra{k}, 
$$
where $g_0(k) := (2\pi/L)\sqrt{2ak\Delta\tau}$, and $g(k)$ is a linear function with respect to $k$ satisfying
$$
g(2^m) = g_0(2^m), \quad \mbox{and} \quad g(2^{n_{\text{dis}}}-1) = g_0(2^{n_{\text{dis}}}-1).
$$
One can also choose $g_0(k) = \arccos(\mathrm{exp}\left(-a(2\pi/L)^2\Delta\tau k\right))$ such that there are no approximation errors for small Fourier coefficients. 
The piecewise function can be implemented by using the quantum comparator with an integer $\mathrm{COMP}(2^m)$ (e.g. Yuan \cite{Yuan.2023}) and combing the exact implementation on the least significant $m$ qubits (e.g. \cite{Zhang.2024}):
$$
U_1 := \sum_{k=0}^{2^m-1} \mathrm{exp}\left(-\mathrm{i}(g_0(k)-g(k))\right) \ket{k}\bra{k},
$$
and the (linear function) phase gate:
$$
U_2 := \sum_{k=0}^{2^{n_{\text{dis}}}-1} \mathrm{exp}\left(-\mathrm{i}g(k)\right) \ket{k}\bra{k},
$$
which is similar to Fig.~3 in \cite{Huang.2024p}. The quantum circuit of $\tilde U_{\text{sqrt}, m}$ requires one ancillary qubit for the comparison and is given in Fig.~\ref{appG:Fig4}. Again the uncomputation of $\mathrm{COMP}(2^m)$ is added, so that the ancillary qubit is detangled from the other qubits and can be reused safely. 
\begin{figure}[htb]
\centering
\resizebox{11cm}{!}{
\includegraphics[keepaspectratio]{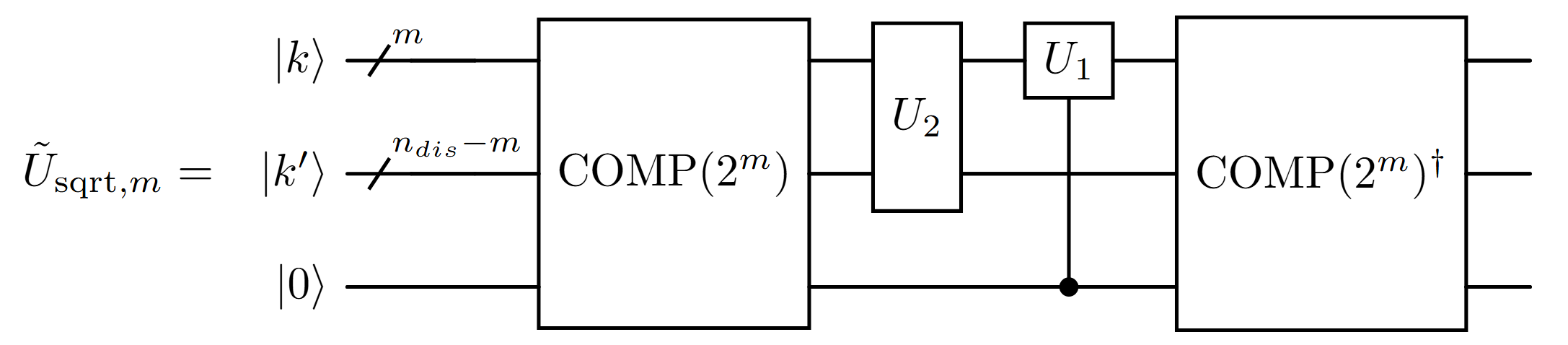}
}
\caption{Quantum circuit for the modified square root operator with a cut-off parameterized by $m$. }
\label{appG:Fig4}
\end{figure}
The gate complexity of the modified implementation $\tilde U_{\text{sqrt}, m}$ is $O(2^m)+O(n^2)$. 

We mention that in the AAPITE circuit, there are two controlled RTE operators. Since the uncomputation part of the first controlled RTE operator cancels the unitary operations $U_{\text{dis}}(L/2)$ and $\mathrm{COMP}(2^m)$ in the second controlled RTE operator, we need the unitary operator for obtaining the squared distance register, the quantum comparator and their inverse operations only once for a single AAPITE circuit. 
For the possible interests of the readers, we provide the detailed CNOT gate count and the circuit depth by Qiskit in a single AAPITE step for the advection-diffusion equation (i.e., 2D case in Sect.~\ref{subsec:4-2} with $V\equiv 0$). Here, we do not include the potential part because its detailed gate count depends on the property of the potential function itself as well as the approximation methods that one uses, and hence is not easy to clarify in several lines. 
The CNOT count and circuit depth are provided in Table \ref{appG:tab1} for $n=3,4,\ldots,11$. The exact implementation has an exponentially increasing gate complexity $O(2^{2n})$, while the modified implementation with a fixed parameter $m=6$ has a polynomial increasing gate complexity $O(n^3)$. 
\begin{table}[htb]
\centering
\caption{CNOT count and circuit depth for a single AAPITE circuit in 2D case by using Qiskit.}
\label{appG:tab1}
\scalebox{0.85}[0.85]{
\begin{tabular}{l|ccccccccc|c}
\hline
& $n=3$ & $n=4$ & $n=5$ & $n=6$ & $n=7$ & $n=8$ & $n=9$ & $n=10$ & $n=11$ & Order \\
\hline
CNOT count (exact) & 988 & 2892 & 8668 & 28804 & - & - & - & - & - & $O(2^{2n})$ \\
CNOT count ($m=6$) & 1268 & 2048 & 3320 & 5192 & 7736 & 11072 & 15272 & 20456 & 26696 & $O(n^3)$ \\
Circuit depth (exact) & 997 & 2797 & 8405 & 28279 & - & - & - & - & - & $O(2^{2n})$ \\
Circuit depth ($m=6$) & 1328 & 2001 & 3090 & 4629 & 6688 & 7301 & 10606 & 14138 & 18819 & $O(n^3)$ \\
\hline
\end{tabular}
}
\end{table}

\section{Derivation of quantum solution}
\label{appC}

In this section, we explain how to derive a quantum solution to Eq.~\eqref{sec4:eq-gov} with Eq.~\eqref{sec4:eq-init}. 
According to the mathematical formulation in \ref{appB} (Method 2), we include a pre-processing step, a quantum solving step, and a post-processing step. Let $n\in \mathbb{N}$, $N_0=2^n$ and $N=N_0^d$. 

\subsection{Pre-processing step}

For the given initial condition $u_0$, we first calculate the coefficients of its Fourier series:
\begin{align*}
c_{\mathbf{k}}(0) = (\phi_{\mathbf{k}}, u_0) = \frac{1}{L^{d/2}}\int_{[0,L]^d} u_0(\mathbf{x}) \mathrm{exp}\left(-\mathrm{i}2\pi \mathbf{k}\cdot\mathbf{x}/L\right) \mathrm{d}\mathbf{x}, \quad \mathbf{k}\in [N_0]^d.
\end{align*}
To prepare the initial quantum state for a quantum solver, we need to normalize the coefficients: 
\begin{align*}
\tilde c_{\mathbf{k}}(0) := \frac{c_{\mathbf{k}}(0)}{\sqrt{\sum_{\mathbf{l}\in [N_0]^d} |c_{\mathbf{l}}(0)|^2}}, \quad \mathbf{k}\in [N_0]^d,
\end{align*}
and we use the amplitude encoding to derive a quantum state:
\begin{align*}
\ket{\mathbf{\tilde c}(0)} := \sum_{\mathbf{k}\in [N_0]^d} \tilde c_{\mathbf{k}}(0) \ket{\mathbf{\tilde k}}, \quad \mathbf{\tilde k} = \mathbf{k} + N_0/2.
\end{align*}

\subsection{Quantum solving step}

\noindent \underline{Preparation of initial quantum state} 
We prepare the initial quantum state $F_N \ket{\mathbf{\tilde c}(0)}$ by an oracle $U_{\mathbf{\tilde c}(0)}$:
$$
U_{\mathbf{\tilde c}(0)} \ket{0}_{dn} := \ket{\mathbf{\tilde c}(0)},
$$
and a shifted QFT. The above oracle is a special case of the quantum state preparation problem and is essentially the same as the vector oracle in the QLSAs. In general, it is known that the exact implementation of the oracle is expensive and requires the gate complexity $O(N)$. In the case of separable variable $u_0$, that is, 
$u_0(\mathbf{x}) = u_{0,1}(x_1)u_{0,2}(x_2)\cdots u_{0,d}(x_d)$, we have
\begin{align*}
&c_{\mathbf{k}}(0) = \prod_{j=1}^d \frac{1}{\sqrt{L}}\int_0^L u_{0,j}(x)\mathrm{exp}\left(-\mathrm{i}2\pi k_jx/L\right) \mathrm{d}x = \prod_{j=1}^d c_{k_j,j}(0), \quad \\
&c_{k,j}(0) := \frac{1}{\sqrt{L}}\int_0^L u_{0,j}(x)\mathrm{exp}\left(-\mathrm{i}2\pi kx/L\right) \mathrm{d}x,
\end{align*}
and
\begin{align*}
\tilde c_{k,j}(0) := \frac{c_{k,j}(0)}{\sqrt{\sum_{l\in [N_0]}|c_{l,j}(0)|^2}}, \quad k\in [N_0],\ j=1,\ldots,d.
\end{align*}
Hence, 
\begin{align*}
\ket{\mathbf{\tilde c}(0)} = \bigotimes_{j=1}^d \left(\sum_{k_j\in [N_0]} \tilde c_{k_j,j}(0) \ket{\tilde k_j}\right),
\end{align*}
and
$$
U_{\mathbf{\tilde c}(0)} \ket{0}_{dn} = \bigotimes_{j=1}^d U_{\mathbf{\tilde c}(0),j} \ket{0}_{n}, \quad U_{\mathbf{\tilde c}(0),j} \ket{0}_{n} := \sum_{k\in [N_0]} \tilde c_{k,j}(0) \ket{\tilde k}_n
$$
In such a case, we can prepare the state for each dimension in parallel, and the oracle requires $O(dN_0)$ quantum gates with depth $O(N_0)=O(N^{1/d})$. Moreover, there exist efficient (approximate) implementations \cite{Moosa.2023, KDNTM23, Grover.2002, Ramos22, Vazquez.2022, Mozafari.2022, Zhang.2022, Bharadwaj.2023} possibly reducing the complexity to $O(\mathrm{polylog}N)$ provided that the desired quantum state has a sufficiently good underlying function. For example, the Fourier series load (FSL) method \cite{Moosa.2023} is efficient for sufficiently smooth functions such that the Fourier coefficients decrease rapidly, and the localized-function method \cite{KDNTM23} works if the function can be well approximated by a finite number of linear combinations of discrete Lorentzian functions.

\noindent \underline{Evolution by PITE} 
By applying the PITE circuit $U_{\text{PITE}}$ to the prepared quantum state, we obtain an approximate quantum state of 
$$
\ket{\psi(t)} \approx \frac{\mathrm{exp}\left(-t H_N\right) F_N \ket{\mathbf{\tilde c}(0)}}{\|\mathrm{exp}\left(-t H_N\right) F_N \ket{\mathbf{\tilde c}(0)}\|},
$$
for an arbitrarily given $t\in (0,T]$ with success probability
$$
\mathbb{P}(\ket{0})(t) = \|U_{\text{PITE}}(t) F_N \ket{\mathbf{\tilde c}(0)}\|^2 \approx \|\mathrm{exp}\left(-t H_N\right) F_N \ket{\mathbf{\tilde c}(0)}\|^2. 
$$
Here, $H_N$ is defined by Eq.~\eqref{sec3:eq-HN}. The square root of the success probability can be obtained up to precision $\varepsilon_0$ by repeating the quantum circuit and measuring the ancillary qubit for $O(1/\varepsilon_0)$ times.

\noindent \underline{Measurement} 
By measuring the derived quantum states for sufficiently many times, we obtain the histogram, from which we retrieve the approximate values of $\{|\psi_{\mathbf{l}}(t)|^2\}_{\mathbf{l}\in [\tilde N_0]^d}\in \mathbb{R}^{N}$. 
The square roots of these values approximate the normalized solutions to the governing equation at each grid point as $N\to \infty$. 
Moreover, to execute a post-processing step for small $N$, we assume that we can do quantum state tomography and obtain the amplitude of the quantum state $\ket{\psi(t)}$: $\{\psi_{\mathbf{l}}(t)\}_{\mathbf{l}\in [\tilde N_0]^d} \in \mathbb{C}^{N}$. 

\subsection{Post-processing step}
\label{subsec:C-3}

According to the mathematical theory for the real-space grid method, we find the normalized approximate solution at point $\mathbf{x}\in [0, L]^d$ is given as 
$$
\sum_{\mathbf{\tilde l}\in [\tilde N_0]^d} \psi_{\mathbf{\tilde l}}(t)g_N(\mathbf{x};\mathbf{p}_{\mathbf{\tilde l}}), 
$$
where $g_N$ is the multi-dimension pixel function defined in \ref{subsec:appB-2}. In particular, we have the approximate solution at the grid points $\{\mathbf{p}_{\mathbf{l}}\}_{\mathbf{l}\in [\tilde N_{\text{f}}]^d}$ as a vector:
\begin{align*}
\tilde u_{N_0,N_f}(t) := \sum_{\mathbf{l}\in [\tilde N_f]^d} \left(\sqrt{\mathbb{P}(\ket{0})(t)}\|\mathbf{c}(0)\| \sum_{\mathbf{\tilde l}\in [\tilde N_0]^d} \psi_{\mathbf{\tilde l}}(t)g_N(\mathbf{p}_{\mathbf{l}};\mathbf{p}_{\mathbf{\tilde l}}) \right)\ket{\mathbf{l}}.
\end{align*}
Here, $\|\mathbf{c}(0)\|$ and $\sqrt{\mathbb{P}(\ket{0})(t)}$ are obtained in the pre-processing step and quantum solving step, respectively, and $N_f = 2^{n_f}$ is a possibly finer grid parameter. 

In fact, we have $F_N \ket{\mathbf{\tilde c}(0)} \to \sum_{\mathbf{l}\in [\tilde N_0]^d} u_0(\mathbf{p}_{\mathbf{l}})/\|u_0(\mathbf{p}_{\mathbf{l}})\|\ket{\mathbf{l}}$ and $g_N(\mathbf{x};\mathbf{p}_{\mathbf{\tilde l}})\to \delta(\mathbf{x}-\mathbf{p}_{\mathbf{\tilde l}})$ as $N\to \infty$. Thus, for sufficiently large $N$, the quantum solution without pre-processing and post-processing steps provides also an approximate solution. 
In the case of relatively small $N$, we can choose $N_f^d>N$ to obtain the solutions at finer grid points. The key idea of the pre-processing and the post-processing steps is to project the initial condition to an $N$-dimensional subspace of $L^2([0, L]^d;\mathbb{C})$: $\mathcal{V}_N = \mathrm{span}\{\phi_{\mathbf{k}}, \mathbf{k}\in [N_0]^d\}$, and then solve the governing equation inside this subspace. With these steps, we can obtain the approximate solution at finer grid points with the same error level as the approximate solution at $N=N_0^d$ grid points. 

\begin{definition}
We call the above approximate solution $\tilde u_{N_0, N_f}$ at $N_f^d$ grid points using both classical computation and quantum computation of only $N_0^d$ grid points, the $(N_0, N_f)$-quantum solution. In particular, if $N_f=N_0$, then we denote it by $\tilde u_{N_0} \equiv \tilde u_{N_0,N_0}$ and simply call it $N_0$-quantum solution. 
\end{definition}

We end up this part with a comment on the computational complexity of the pre-processing step and the post-processing step. 
For the pre-processing step, if the initial condition is good such that we can calculate the Fourier coefficients analytically, then we only need $O(N)$ arithmetic operations to load them into the classical memory. On the other hand, we need the discrete Fourier transform (DFT) to calculate the coefficients. By the best classical algorithm of DFT, this requires $O(N\log (N))$ arithmetic operations. 
For the post-processing step, we need to multiply a $(N_f^d, N)$ matrix to a $(N,1)$ vector, which requires $O(N_f^d N)$ arithmetic operations. If the $N_f^d$ points can be calculated in parallel on multiple computers/GPU, then it requires $O(N)$ arithmetic operations for each point. 
Although the pre-/post-processing step has a dominant order of computational complexity compared to the quantum solving step, it seems reasonable because we are looking for the solution at $N$ points whose preparation itself needs $\Omega(N)$ arithmetic operations. 

\section{Theoretical details}

\subsection{Details of Theorem \ref{thm:aap-err-esti}}
\label{prf:thm1}

\begin{theorem}[Detailed version of Theorem 1]
\label{thm:aap-err-estin}
Let $N, K\in\mathbb{N}$, $\Delta\tau\in \mathbb{R}_{>0}$, and $\mathcal{H}\in \mathbb{C}^{N\times N}$ be a positive semi-definite Hermite matrix. For an arbitrarily given quantum state $\ket{\psi}$, i.e., $\ket{\psi}\in \mathbb{C}^{N}$ satisfying $\|\ket{\psi}\|=1$, and any $\delta\in [0,1)$, there exists a constant $C_1$, depending only on $\mathcal{H}$, $\ket{\psi}$, and $\delta$, such that the following error estimate holds true.
\begin{align}
\label{sec2:eq-err2n}
\mathrm{Err} \le \left(4C_1 K(\Delta\tau)^2/3 +4\delta\right)/\|\mathrm{exp}\left(-K\Delta\tau \mathcal{H}\right)\ket{\psi}\|. 
\end{align}
\end{theorem}
\begin{proof}
Since $\mathcal{H}$ is a positive semi-definite Hermite matrix, we can assume that it has the eigensystem: $\{(\lambda_k,\psi_k)\}_{k=0}^{N-1}$ where $0\le \lambda_0\le \lambda_1 \le \cdots \le \lambda_{N-1}$ are real-valued, and $\{\psi_k\}_{k=0}^{N-1}$ forms a complete orthonormal basis of $\mathbb{C}^N$. 
By the triangle inequality, we have
\begin{align*}
\mathrm{Err}
&\le \left\|\frac{\cos^K(\sqrt{2\Delta\tau \mathcal{H}})\ket{\psi}}{\|\cos^K(\sqrt{2\Delta\tau \mathcal{H}})\ket{\psi}\|}-\frac{\cos^K(\sqrt{2\Delta\tau \mathcal{H}})\ket{\psi}}{\|\mathrm{exp}\left(-K\Delta\tau \mathcal{H}\right)\ket{\psi}\|}\right\| + \frac{\|\cos^K(\sqrt{2\Delta\tau \mathcal{H}})\ket{\psi}-\mathrm{exp}\left(-K\Delta\tau \mathcal{H}\right)\ket{\psi}\|}{\|\mathrm{exp}\left(-K\Delta\tau \mathcal{H}\right)\ket{\psi}\|}\\
&\le \frac{2\|\cos^K(\sqrt{2\Delta\tau \mathcal{H}})\ket{\psi}-\mathrm{exp}\left(-K\Delta\tau \mathcal{H}\right)\ket{\psi}\|}{\|\mathrm{exp}\left(-K\Delta\tau \mathcal{H}\right)\ket{\psi}\|}.
\end{align*}
Then, it is sufficient to estimate the error before the normalization: 
$$
\widetilde{\mathrm{Err}} = \left\|\cos^K(\sqrt{2\Delta\tau \mathcal{H}})\ket{\psi}-\mathrm{exp}\left(-K\Delta\tau \mathcal{H}\right)\ket{\psi}\right\|. 
$$
For any $\delta\in [0,1)$, we introduce
\begin{align}
\label{sec2:eq1a}
N_\delta(\psi) &:= \left\{k\in [\tilde N]:=\{0,\ldots,N-1\}; \sum_{j=0}^{k-1}|\langle\psi_j|\psi\rangle|^2 < 1-\delta^2, \sum_{j=0}^k|\langle\psi_j|\psi\rangle|^2 \ge 1-\delta^2 \right\}, \\
\label{sec2:eq1b}
\mathcal{K}_\delta(\psi) &:= \left\{0,\ldots,N_\delta(\psi)\right\}\subset [\tilde N]. 
\end{align}
Then, we estimate the upper bound of $\widetilde{\mathrm{Err}}^2$ as follows: 
\begin{align}
\nonumber
&\quad \widetilde{\mathrm{Err}}^2 = \sum_{k=0}^{N-1} |\langle \psi_k|\psi\rangle|^2 \left|\cos^K(\sqrt{2\Delta\tau \lambda_k})-(\mathrm{exp}\left(-\Delta\tau \lambda_k\right))^K\right|^2 \\
\nonumber
&= \sum_{k=0}^{N-1} |\langle \psi_k|\psi\rangle|^2 \left|(\cos(\sqrt{2\Delta\tau \lambda_k}))^K-(\mathrm{exp}\left(-\Delta\tau \lambda_k\right))^K\right|^2 \\
\nonumber
&\le \sum_{k\in \mathcal{K}_{\delta}} |\langle \psi_k|\psi\rangle|^2 \left|\cos(\sqrt{2\Delta\tau \lambda_k})-\mathrm{exp}\left(-\Delta\tau \lambda_k\right)\right|^2 \left(\sum_{j=0}^{K-1}|\cos(\sqrt{2\Delta\tau \lambda_k})|^j|\mathrm{exp}\left(-\Delta\tau \lambda_k\right)|^{K-1-j}\right)^2 + 4\delta^2 \\
\label{sec2:eq3}
&\le K^2 \sum_{k\in \mathcal{K}_{\delta}} |\langle \psi_k|\psi\rangle|^2 \left|\cos(\sqrt{2\Delta\tau \lambda_k})-\mathrm{exp}\left(-\Delta\tau \lambda_k\right)\right|^2 + 4\delta^2.
\end{align}
We define
$$
h(x) := \cos(x) - \mathrm{exp}\left(-\frac12 x^2\right), \quad x\in \mathbb{R}. 
$$
By a direct calculation, we have
\begin{align*}
& h^\prime(x) = -\sin(x) + x \mathrm{exp}\left(-\frac12 x^2\right), \quad h^\prime(0) = 0, \\
& h^{(2)}(x) = -\cos(x) + (1-x^2) \mathrm{exp}\left(-\frac12 x^2\right), \quad h^{(2)}(0) = 0, \\
& h^{(3)}(x) = \sin(x) + (x^3-3x) \mathrm{exp}\left(-\frac12 x^2\right), \quad h^{(3)}(0) = 0, \\
& h^{(4)}(x) = \cos(x) + (-x^4+6x^2-3) \mathrm{exp}\left(-\frac12 x^2\right), \quad h^{(4)}(0) = -2,
\end{align*}
where $h^{(k)}$ denotes the $k$-th derivative of $h$. It is clear that $h\in C^\infty(\mathbb{R})$. By employing the Taylor's theorem near the origin with the Lagrange form of the remainder, we have 
\begin{align*}
h(x) &= h(0) + h^\prime(0)x + \frac12 h^{(2)}(0) x^2 + \frac16 h^{(3)}(0) x^3 + \frac{1}{24} h^{(4)}(\xi) x^4 \\
&= \frac{1}{24}\left(\cos\xi + (-\xi^4+6\xi^2-3)\mathrm{exp}\left(-\frac12 \xi^2\right)\right)x^4,
\end{align*}
for any $x>0$ and some $\xi\in [0,x]$. For any $k=0,1,\ldots,N-1$, letting $x=\sqrt{2\Delta\tau \lambda_k}$, we obtain
\begin{align*}
\left|\cos(\sqrt{2\Delta\tau \lambda_k})-\mathrm{exp}\left(-\Delta\tau \lambda_k\right)\right|^2
&= \left|\frac{1}{6}\left(\cos\xi_k + (-\xi_k^4+6\xi_k^2-3)\mathrm{exp}\left(-\frac12 \xi_k^2\right)\right)\right|^2 \lambda_k^4 (\Delta\tau)^4, 
\end{align*}
for some $\xi_k\in [0,\sqrt{2\Delta\tau \lambda_k}]$. By noting
$$
\sup_{x\ge 0} |-x^2+6x-3|\mathrm{exp}\left(-\frac12 x\right) = 3, 
$$
we obtain 
\begin{align*}
\left|\cos(\sqrt{2\Delta\tau \lambda_k})-\mathrm{exp}\left(-\Delta\tau \lambda_k\right)\right|^2
&\le \frac49 \lambda_k^4 (\Delta\tau)^4, 
\end{align*}
for any $k=0,1,\ldots,N-1$. We insert this into Eq.~\eqref{sec2:eq3} and obtain
\begin{align*}
\widetilde{\mathrm{Err}}^2 \le \frac{4K^2}{9}(\Delta\tau)^4 \sum_{k\in \mathcal{K}_{\delta}} \lambda_k^4 |\langle\psi_k|\psi\rangle|^2 + 4\delta^2. 
\end{align*}
By setting
\begin{align*}
C_1(\ket{\psi},\delta) := \left(\sum_{k\in \mathcal{K}_{\delta}} \lambda_k^4 |\langle\psi_k|\psi\rangle|^2\right)^{1/2},
\end{align*}
we obtain 
\begin{align}
\label{sec2:eq4}
\widetilde{\mathrm{Err}} \le \frac{2C_1(\ket{\psi},\delta)K}{3}(\Delta\tau)^2 + 2\delta.
\end{align}
Therefore, we reach the error estimate:
\begin{align*}
\mathrm{Err} \le \left(4C_1(\ket{\psi},\delta) K(\Delta\tau)^2/3 +4\delta\right)/\|\mathrm{exp}\left(-K\Delta\tau \mathcal{H}\right)\ket{\psi}\|, 
\end{align*}
for any $\delta\in [0,1)$. 
\end{proof}
By the definition, $C_1$ depends on the eigensystem of $\mathcal{H}$, and hence has an implicit dependence on $N$. 
In this paper, we consider mainly two operators: a bounded operator $V(x)I$ and an unbounded differential operator $-\nabla^2$. 
To derive an exponential speedup regarding $N$, an $N$-uniform upper bound of $C_1$ is needed. 

\noindent \underline{\bf Bounded operator $V(x)I$} For any bounded operator $V(x)I$ with a non-negative bounded real function $V$, the discretized matrix $\mathcal{H}$ is Hermitian and positive semi-definite. Moreover, we have $\lambda_0=\min_{x}V(x)\ge 0$, and $\lambda_{N-1}=\max_{x}V(x)\le C_0$, where $C_0$ is the upper bound of $V$. Immediately, we have $C_1(\ket{\psi},\delta) \le C_0^2$ according to the definition. By taking $\delta=0$, we have the $N$-uniform bound:
$$
\mathrm{Err} \le 4C_0^2 K(\Delta\tau)^2/(3\|\mathrm{exp}\left(-K\Delta\tau \mathcal{H}\right)\ket{\psi}\|).
$$

\noindent \underline{\bf Differential operator $-\nabla^2$} The discretized matrix of the differential operator is Hermitian and positive semi-definite, but its maximal eigenvalue $\lambda_{N-1}$ scales as $N^2$. If we consider the worst initial state $\ket{\psi}$ such that $|\langle\psi_{N-1}|\psi\rangle|^2=\Omega(1)$, and $\delta=0$, then there is a polynomial dependence on $N$ in $C_1$. Fortunately, in our problem, the initial state has an underlying smooth function, which is the solution to the PDE at each time step. Since the Fourier coefficients of a smooth function decrease rapidly and the truncated error has an exponential decay regarding the number of nodes (e.g. Appendix B in \cite{Childs.2022}), $N_\delta$ defined in \eqref{sec2:eq1a} satisfies $N_\delta=O(\log \delta^{-1})$ provided that $N=\Omega(\log \delta^{-1})$ and the underlying function of the quantum state $\ket{\psi}$ is smooth with the periodic boundary condition. 
If we take $\delta\propto (\Delta\tau)^2$, then we obtain an $N$-uniform bound:
$$
\mathrm{Err} \le \tilde {C_0}K(\Delta\tau)^{2-\epsilon_0}/\|\mathrm{exp}\left(-K\Delta\tau \mathcal{H}\right)\ket{\psi}\|, 
$$
where $\epsilon_0>0$ appears due to the logarithmic dependence of $\delta^{-1}$ in $C_1$, and it can be taken arbitrarily small. Here, $\tilde C_0$ depends on the initial state $\ket{\psi}$, but is $N$-uniformly bounded provided that $N$ is sufficiently large. 
\begin{remark} 
By the current theoretical analysis, the total error has a slightly worse dependence $(\Delta\tau)^{1-\epsilon_0}$ than the linear one ($T=K\Delta\tau$ is a fixed constant in the application). 
On the other hand, the numerical tests indicate that the dependence on $\Delta\tau$ is exactly linear. 
Figures~\ref{appD:Fig5}--\ref{appD:Fig7} in \ref{subsec:4-3-3}, \ref{subsec:4-3-4} imply that for fixed $N$, the errors coming from both the PITE approximation and the Suzuki-Trotter formula are linear regarding $\Delta\tau$, while Fig.~\ref{appD:Fig6n} in \ref{subsec:4-3-2} illustrates that the error will merely increase regarding $N$ as long as $N$ is sufficiently large. 
Thus, we believe there is a more intelligent way to establish the error estimate so that the linear dependence can be proved. For simplicity, we take $\epsilon_0$ close to zero and omit it in the main manuscript. 
\end{remark}

\subsection{Details of the total error estimate}
\label{prf:tot-err-esti}

By using the non-unitary operator $U^K_{\mathrm{AAPITE}}(\Delta\tau)$ introduced in Sect.~\ref{sec:3}, we define the normalized $\ell^2$-error between the ITE operator and the approximate operator based on the AAPITE circuits by
\begin{align*}
E_{\text{tot}}(\Delta\tau; \ket{\psi_{(0)}},K)
&:= \left\|\frac{\mathrm{exp}\left(-K\Delta\tau H_N\right)\ket{\psi_{(0)}}}{\left\|\mathrm{exp}\left(-K\Delta\tau H_N\right)\ket{\psi_{(0)}}\right\|} - U_{\text{AAPITE}}^K(\Delta\tau) \ket{\psi_{(0)}} \right\| \\
&\, \le \frac{2\left\|\mathrm{exp}\left(-K\Delta\tau \tilde H_N\right)\ket{\psi_{(0)}} - P^K \ket{\psi_{(0)}} \right\|}{\left\|\mathrm{exp}\left(-K\Delta\tau \tilde H_N\right)\ket{\psi_{(0)}}\right\|} \\
&= \frac{2\mathrm{exp}\left(-K\Delta\tau V_0\right)}{\left\|\mathrm{exp}\left(-K\Delta\tau H_N\right)\ket{\psi_{(0)}}\right\|}\left\|\mathrm{exp}\left(-K\Delta\tau \tilde H_N\right)\ket{\psi_{(0)}} - P^K \ket{\psi_{(0)}} \right\|,  
\end{align*}
for $K\in \mathbb{N}$. Here, $P := F_N \mathrm{exp}\left(\mathrm{i}\Delta\tau D_N^{(1)}\right)\cos\left(\sqrt{2\Delta\tau D_N^{(2)}}\right)F_N^\dag \cos\left(\sqrt{2\Delta\tau \tilde V_N}\right)$ denotes the approximate PITE operator. Then, it is sufficient to estimate 
\begin{align*}
\tilde E_{\text{tot}}(\Delta\tau; \ket{\psi_{(0)}},K) := \left\|\mathrm{exp}\left(-K\Delta\tau \tilde H_N\right)\ket{\psi_{(0)}} - P^K \ket{\psi_{(0)}} \right\|,
\end{align*}
for $K\in \mathbb{N}$. 
By the definition of $P$, we have $\|P\|\le 1$. Hence, a direct estimate yields
\begin{align*}
\tilde E_{\text{tot}}(\Delta\tau; \ket{\psi_{(0)}},K) 
&= \left\|\sum_{j=0}^{K-1} P^{K-1-j} \left(P - \mathrm{exp}\left(-\Delta\tau \tilde H_N\right)\right) \mathrm{exp}\left(-j\Delta\tau \tilde H_N\right) \ket{\psi_{(0)}} \right\| \\
&\le \sum_{j=0}^{K-1}\left\|\left(P - \mathrm{exp}\left(-\Delta\tau \tilde H_N\right)\right) \mathrm{exp}\left(-j\Delta\tau \tilde H_N\right) \ket{\psi_{(0)}} \right\| \\
&= \sum_{j=0}^{K-1} \tilde E_{\text{tot}}(\Delta\tau; \ket{\psi_{(j)}},1).
\end{align*}
Here and henceforth, we denote $\ket{\psi_{(j)}} = \mathrm{exp}\left(-j\Delta\tau \tilde H_N\right) \ket{\psi_{(0)}}$ for simplicity. 
Moreover, we have
\begin{align*}
&\tilde E_{\text{tot}}(\Delta\tau; \ket{\psi_{(j)}},1) \\
&= \left\|\left(P - \mathrm{exp}\left(-\Delta\tau \tilde H_N\right)\right) \ket{\psi_{(j)}} \right\| \\
&\le \left\|\left(P - F_N \mathrm{exp}\left(\mathrm{i}\Delta\tau D_N^{(1)}\right)\mathrm{exp}\left(-\Delta\tau D_N^{(2)}\right)F_N^\dag \mathrm{exp}\left(-\Delta\tau \tilde V_N\right)\right) \ket{\psi_{(j)}} \right\| \\
&\quad + \left\|\left(\mathrm{exp}\left(-\Delta\tau F_N (D_N^{(2)}-\mathrm{i}D_N^{(1)}) F_N^\dag\right) \mathrm{exp}\left(-\Delta\tau \tilde V_N\right) - \mathrm{exp}\left(-\Delta\tau \tilde H_N\right)\right) \ket{\psi_{(j)}} \right\| \\
&=: \tilde E_{\text{appr}}\left(\Delta\tau; \ket{\psi_{(j)}}\right) + \tilde E_{\text{trot}}\left(\Delta\tau; \ket{\psi_{(j)}}\right), 
\end{align*}
for $j=0,\ldots,K-1$. 
We estimate the Suzuki-Trotter error and the approximation error, respectively. 

First, we consider the Suzuki-Trotter error, which can be estimated by the following lemma.
\begin{lemma}
\label{appC2:lem1}
Let $(X,\|\cdot\|_X)$ be a Banach space. Assume that $-A$ is a generator of a strongly continuous semigroup $e^{-tA}$ on $X$, and $B$ is a bounded linear operator on $X$. Moreover, we assume 
\begin{equation}
\label{sec2:eq-assump1}
\|e^{-tA}\| \le 1, \quad \|e^{-tB}\| \le 1, \quad \|e^{-t(A+B)}\| \le 1,
\end{equation}
for all $t>0$. Then, for any $x\in X$ and $t>0$, we have
\begin{align*}
\left\|\left(e^{-tA}e^{-tB} - e^{-t(A+B)}\right)x\right\|_X \le \left(\|B\|^2\|x\|_X + \frac12 \max_{0\le s\le t}\|[A,B]e^{-s(A+B)}x\|\right) t^2.  
\end{align*}
\end{lemma}
\begin{proof}
We follow the strategy of \cite{Jahnke.2000} in which the error estimate of the Strang splitting (second-order splitting) was obtained. 
According to the variation-for-constant formula, for any $x\in X$ and $t>0$, we have 
$$
e^{-t(A+B)}x = e^{-tA}x - \int_0^t e^{-(t-s)A} B e^{-s(A+B)}x\, \mathrm{d}s, 
$$
On the other hand, by the series expansion of $e^{-tB}$, we find
$$
e^{-tA}e^{-tB}x = e^{-tA}x - t e^{-tA}Bx + R_1x,
$$
where $\|R_1x\|:=\|e^{-tA}(e^{-tB}-I+tB)x\|\le \frac{t^2}{2}\|B\|^2 \|x\|_X$ by Eq.~\eqref{sec2:eq-assump1}. Then, we estimate
\begin{align*}
e^{-tA}e^{-tB}x - e^{-t(A+B)}x = -te^{-tA}Bx + \int_0^t e^{-(t-s)A} B e^{-s(A+B)}x\, \mathrm{d}s + R_1x =: R_2x + R_1x.
\end{align*}
We introduce a series of operator $W(s) := e^{-(t-s)A}Be^{-s(A+B)}$, $s\ge 0$, then we have
\begin{align*}
R_2x = -t W(0)x + \int_0^t W(s)x\, \mathrm{d}s = t^2 \int_{0}^1 (1-\theta)W^\prime(\theta t)x\, \mathrm{d}\theta. 
\end{align*}
Since $W^\prime(s) = e^{-(t-s)A}[A, B]e^{-s(A+B)} - e^{-(t-s)A}B^2e^{-s(A+B)}$, by Eq.~\eqref{sec2:eq-assump1}, we have
\begin{align*}
\|W^\prime(s)x\| \le \|[A,B]e^{-s(A+B)}x\| + \|B\|^2\|x\|_X, \quad 0\le s\le t.
\end{align*}
Therefore, we have 
$$
\|R_2x\| \le \frac{t^2}{2}\left(\max_{0\le s\le t}\|[A,B]e^{-s(A+B)}x\| + \|B\|^2\|x\|_X\right).
$$
Then, the desired estimate follows immediately from the triangle inequality. 
\end{proof}

Letting $X=\mathbb{C}^N$, $A=F_N (D_N^{(2)}-\mathrm{i}D_N^{(1)}) F_N^\dag$, and $B=\tilde V_N$ in Lemma \ref{appC2:lem1}, we obtain 
\begin{align*}
&\quad \tilde E_{\text{trot}}\left(\Delta\tau; \ket{\psi_{(j)}}\right) \\
&\le \left(\|\tilde V_N\|^2 \left\|\ket{\psi_{(j)}}\right\|+\frac12 \max_{0\le s\le \Delta\tau}\|[F_N (D_N^{(2)}-\mathrm{i}D_N^{(1)})F_N^\dag, \tilde V_N]\mathrm{exp}\left(-s\tilde H_N\right)\ket{\psi_{(j)}}\|\right) (\Delta\tau)^2 \le C_2 (\Delta\tau)^2. 
\end{align*}
Here, $C_2$ is a positive constant depending on $V, V^\prime, V^{(2)}$, and the a priori bound of the derivative of the solution to the continuous equation (see e.g. Chapter 5 in \cite{Lin2022}), and hence is uniform regarding $N$. 

Next, we consider the approximation error. We divide it into two parts using the following two operators: 
\begin{align*}
&P_1 := \left(F_N \mathrm{exp}\left(\mathrm{i}\Delta\tau D_N^{(1)}\right)\cos\left(\sqrt{2\Delta\tau D_N^{(2)}}\right)F_N^\dag -F_N \mathrm{exp}\left(\mathrm{i}\Delta\tau D_N^{(1)}\right) \mathrm{exp}\left(-\Delta\tau D_N^{(2)}\right)F_N^\dag\right)\mathrm{exp}\left(-\Delta\tau \tilde V_N\right), \\
&P_2 := \left(F_N \mathrm{exp}\left(\mathrm{i}\Delta\tau D_N^{(1)}\right)\cos\left(\sqrt{2\Delta\tau D_N^{(2)}}\right)F_N^\dag \right)\left(\cos\left(\sqrt{2\Delta\tau \tilde V_N}\right)-\mathrm{exp}\left(-\Delta\tau \tilde V_N\right)\right),
\end{align*}
so that
\begin{align*}
\tilde E_{\text{appr}}\left(\Delta\tau; \ket{\psi_{(j)}}\right) 
\le \left\|P_1 \ket{\psi_{(j)}} \right\| + \left\|P_2 \ket{\psi_{(j)}} \right\|.
\end{align*}
Moreover, we have
\begin{align*}
&\quad F_N \mathrm{exp}\left(\mathrm{i}\Delta\tau D_N^{(1)}\right)\cos\left(\sqrt{2\Delta\tau D_N^{(2)}}\right)F_N^\dag -F_N \mathrm{exp}\left(\mathrm{i}\Delta\tau D_N^{(1)}\right) \mathrm{exp}\left(-\Delta\tau D_N^{(2)}\right)F_N^\dag \\
&= \left(F_N \mathrm{exp}\left(\mathrm{i}\Delta\tau D_N^{(1)}\right) F_N^\dag\right)\left(\cos\left(\sqrt{2\Delta\tau F_N D_N^{(2)}F_N^\dag}\right) - \mathrm{exp}\left(-\Delta\tau F_N D_N^{(2)}F_N^\dag\right)\right).
\end{align*} 
We let $K=1$, $\mathcal{H} = F_N D_N^{(2)}F_N^\dag$ and $\tilde V_N$, respectively, in Theorem \ref{thm:aap-err-esti}. Then, we have
\begin{align*}
&\left\|P_1 \ket{\psi_{(j)}} \right\| \le C_3(\Delta\tau)^2, \quad \left\|P_2 \ket{\psi_{(j)}} \right\| \le C_4(\Delta\tau)^2.
\end{align*}
According to \ref{prf:thm1}, $C_3$ depends on the state $\mathrm{exp}\left(-\Delta\tau \tilde V_N\right)\ket{\psi_{(j)}}$, and it is uniform regarding $N$ provided that the underlying function of the state is sufficiently smooth. On the other hand, $C_4$ depends only on the variation of $V$: $V_1-V_0$, where $V_1 := \max_{\mathbf{x}\in [0,L]^d} V(\mathbf{x}) \ge V_0$. 
To sum up, we find that 
\begin{align*}
\tilde E_{\text{tot}}(\Delta\tau; \ket{\psi_{(0)}},K) \le C_5 K(\Delta\tau)^2,
\end{align*}
for a constant $C_5>0$. By noting that $T=K\Delta\tau$, we have 
\begin{align*}
E_{\text{tot}}(\Delta\tau; \ket{\psi_{(0)}},T/\Delta\tau) 
\le \frac{2C_5 T\mathrm{exp}\left(-T V_0\right)}{\left\|\mathrm{exp}\left(-T H_N\right)\ket{\psi_{(0)}}\right\|}\Delta\tau = O(\Delta\tau).
\end{align*}

\section{Numerical details}
\label{appE}

In this part, we discuss further details on the error of our proposed method using numerical examples. 
According to the theoretical results in \ref{subsec:appB-2} and \ref{prf:tot-err-esti}, we find the total $\ell^2$-error can be divided into the following three parts:

\noindent -- Discretization error: This describes the difference between the solution to the original continuous problem and the exact solution to the discretized problem. Here, we have a grid parameter $N$, indicating the size of the matrix, and the error goes to zero as the grid parameter $N$ goes to infinity. 

\noindent -- Suzuki-Trotter error: This error comes from the Suzuki-Trotter formula when the two (or several) operators under consideration are not commutable. Here, we have a time step parameter $\Delta\tau$, and the error goes to zero as the time step $\Delta\tau$ goes to zero. 

\noindent -- Approximation error: This error describes the difference between the numerical solution using the exact PITE circuit and numerical solution using the proposed approximate PITE circuit. Here, we also have a time step parameter $\Delta\tau$, and the error goes to zero as the time step $\Delta\tau$ goes to zero. 

In the following context, we investigate the three types of errors regarding the grid parameter $N$ and the time step $\Delta\tau$, respectively. 

\subsection{Discretization error: numerical dependence on grid parameter}
\label{subsec:4-3-1}

We consider the 1D example in Sect.~\ref{subsec:4-1} and focus on the discretization error. To do this, we take ``the analytical solution" (more precisely, we cut off the infinite series in Eq.~\eqref{sec4:eq-AS} after $1000$ terms) as the reference solution. 
In the following plot of error, we compare the quantum solutions using the exact PITE operator to highlight the discretization error under different choices of $N=2^4,2^5,2^6$. 
\begin{figure}[htb]
\centering
\resizebox{9cm}{!}{
\includegraphics[keepaspectratio]{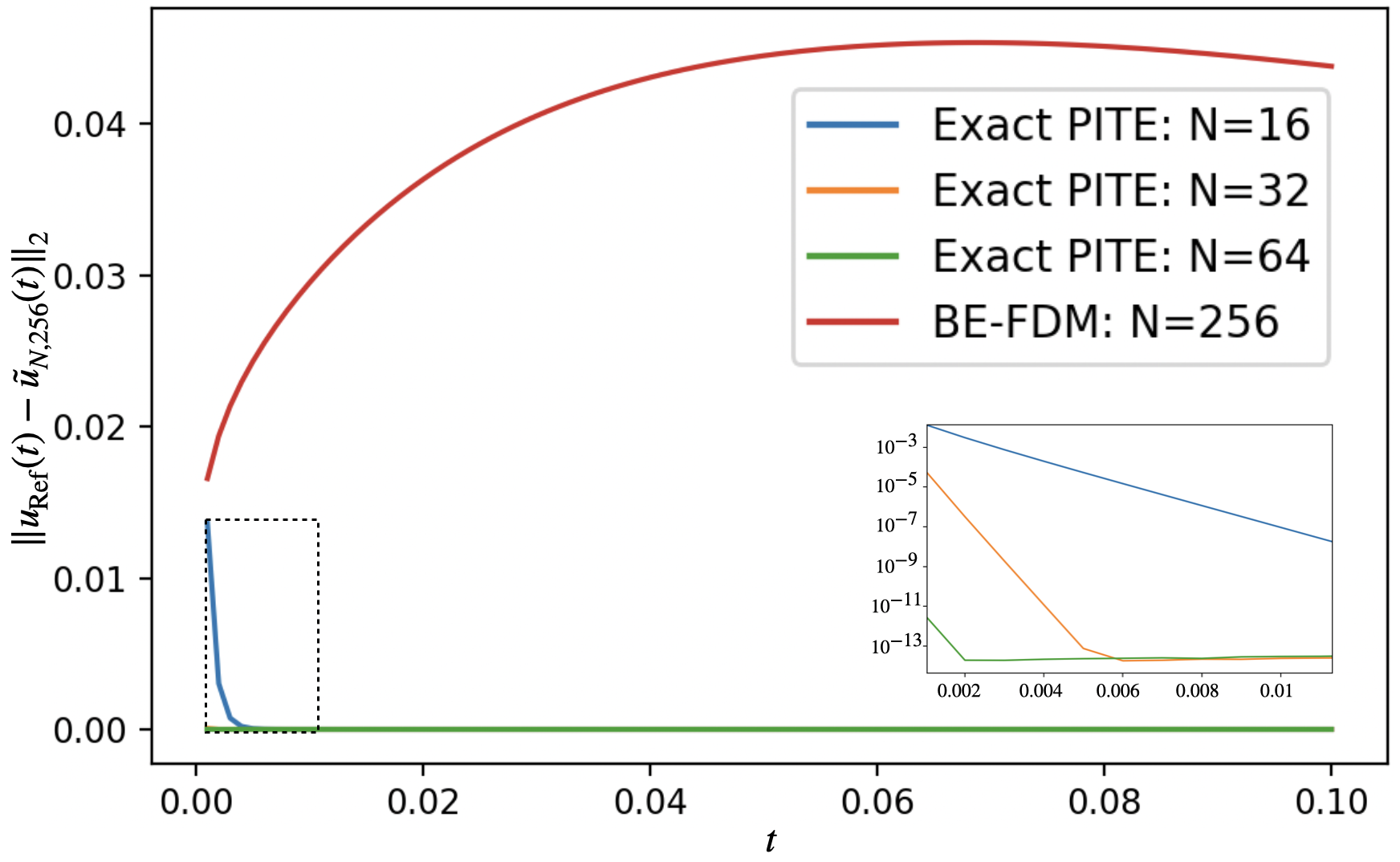}
}
\caption{(1D example) Plot of $\ell^2$-error between the quantum solutions and the reference solution. The blue, orange and green lines denotes the $(N, 2^8)$-quantum solutions using the exact PITE with $N=16,32,64$ and $\Delta\tau=10^{-3}$, respectively. 
The classical solutions by the backward Euler FDM with time step $\Delta\tau = t/10^{3}$ are provided in red. A zoom-in plot in log-scale is added to distinguish the green line from the orange line.} 
\label{appD:Fig2}
\end{figure}
According to Fig.~\ref{appD:Fig2}, we observe that the quantum solutions with small $N=16$ outperform the backward Euler FDM solution with a large $N=256$, which shows that the Fourier spectral approach has better precision than the FDM based approach regarding the grid parameter. 
Besides, in this example, even $4$ qubits ($N=2^4=16$) are enough to derive a sufficiently good solution with error bound $10^{-7}$ after $t\ge 0.01$, and it is clear from the zoom-in plot that the decay of the $\ell^2$-error regarding $N$ is exponentially fast.  

Thanks to the vanishing of the reaction term (i.e., potential $V$), the ITE keeps invariant in any $N$-dimensional spaces $\mathcal{V}_N$ spanned by the Fourier bases. In other words, $\mathrm{exp}\left(-t\mathcal{H}\right)\mathcal{V}_N \subset \mathcal{V}_N$ for any $t>0$ and $N\in \mathbb{N}$. Since the Fourier coefficients of the initial condition given by a sine function decay rapidly as the indices increase, only several Fourier coefficients are enough for deriving a good approximation of the initial condition. This explains why small $N$ is enough for the example in Sect.~\ref{subsec:4-1}. 
On the other hand, with the presence of the reaction term in $\mathcal{H}$, $\mathcal{V}_N$ is no longer invariant regarding $\mathrm{exp}\left(-t\mathcal{H}\right)$. 
This leads to a possible increase of the error as simulation time becomes larger (see e.g. Fig.~\ref{sec4:Fig3} in Sect.~\ref{subsec:4-2}).  

\subsection{Approximation and/or Suzuki-Trotter errors: numerical dependence on grid parameter}
\label{subsec:4-3-2}

Next, we discuss the $N$ dependence in the approximation and the Suzuki-Trotter errors. Recall Eq.~\eqref{sec2:eq-err2n} in \ref{prf:thm1}: 
$$
\text{Err} \le \left(4C_1(\ket{\psi},\delta)T\Delta\tau/3+4\sqrt{\delta}\right)/\|\mathrm{exp}\left(-T\mathcal{H}\right)\ket{\psi}\|,
$$
where 
$$
C_1(\ket{\psi},\delta) = \left(\sum_{k\in \mathcal{K}_\delta}\lambda_k^4|\langle \psi_k|\psi\rangle |^2\right)^{1/2}. 
$$
By the definition, $C_1(\ket{\psi},\delta)$ may depend also on $N$. As we mentioned, it can increase as $N$ becomes larger, and its upper bound is $\lambda_{N-1}^2=O(d^2 N^{4/d})$ in the worst scenario. 
On the other hand, for suitably given initial state $\ket{\psi}$ and an error tolerance $\delta>0$, the $N$ dependence in the approximation error could be relieved to some extent. 
In Fig.~\ref{appD:Fig6n}, the approximation error regarding $N$ is demonstrated under the 1D example without potential in Sect.~\ref{subsec:4-1} and another 1D example with the piecewise constant potential function \eqref{sec4:eq-V} discussed in Sect.~\ref{subsec:5-3}, where the initial states are both sine functions. Here, we take the reference solution as the one calculated by the exact matrix multiplication (i.e., in a similar fashion of the one in Sect.~\ref{subsec:4-2}) to address the errors beyond the discretization error. Our numerical plots imply the $N$-uniform bounds in the case of a good initial state, which means to achieve a given error bound $\delta>0$, we do not need to choose $\Delta\tau\to 0$ as $N\to \infty$. 
Fig.~\ref{appD:Fig6n} illustrates that for fixed time step $\Delta\tau$, the $\ell^2$-error seems to converge as $N$ increases. This gives evidence that the required number of PITE steps $K=t/\Delta\tau$ is uniformly bounded with respect to $N$, and it mainly depends on the initial state and the desired error bound.
Since the gate complexity of the AAPITE circuit is proportional to $K$, Fig.~\ref{appD:Fig6n} hence implies the exponential speedup regarding the grid parameter $N$ in the quantum solving step, compared to the classical calculations using the matrix multiplications or the optimizations (e.g. CG method). 
\begin{figure}[htb]
\centering
\resizebox{15cm}{!}{
\includegraphics[keepaspectratio]{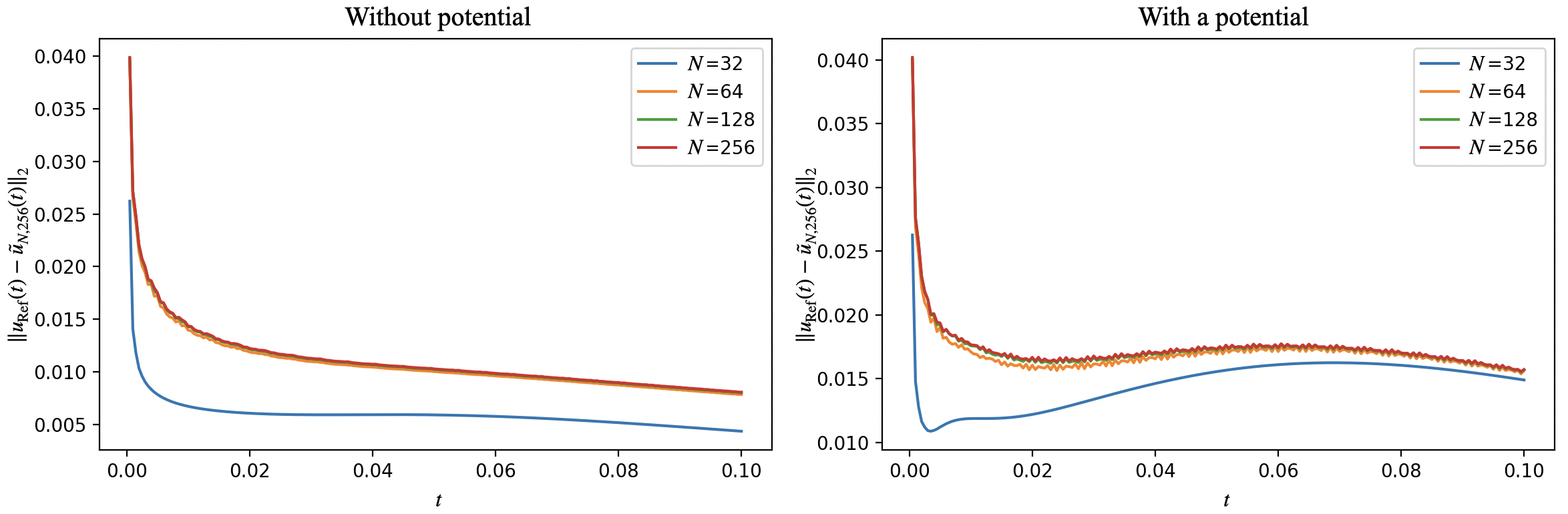}
}
\caption{(1D examples) Plot of $\ell^2$-errors between the $(N,2^8)$-quantum solutions and the reference solution with different $N=32,64,128,256$. Here, we take the same time step $\Delta\tau=0.0005$ in all subplots. The left subplot is for a 1D case without any potential, while the right subplot is for a similar 1D case with an absorption potential in the center part of the domain. }
\label{appD:Fig6n}
\end{figure}
We give a remark on the previous quantum speedup statement. 
Following the LCU-based QLSA in \cite{Childs.2017} using the (best known) Hamiltonian simulation \cite{BCK15, BCC.15} or the QSVT-based QLSA in \cite{Krovi2023} using QSVT \cite{Gilyen.2019}, one can derive an at most polynomial speedup due to the polynomial dependence on $N$ in the condition numbers for the discretized matrices of many PDEs. Although \cite{Clader.2013} suggested a preconditioning technique using SPAI preconditioner to reduce the condition number, there is no guarantee that it works for all discretized matrices. While \cite{Montanaro.2016} showed that exponential speedup is possible only for the quantum solving step provided that the preconditioner works and reduces the condition number to $O(1)$, the exponential speedup in $N$ is not clarified for the existing QLSAs. 
In this work, we confirmed the exponential speedup in $N$ for the advection-diffusion-reaction equation using the AAPITE based on a special FSM for the discretization. Besides, such an exponential speedup in $N$ was confirmed for the constant coefficient hyperbolic equations using another quantum algorithm in \cite{Sato.2024} without discussion of the total error. 

\subsection{Approximation error: numerical dependence on time step}
\label{subsec:4-3-3}

According to the theoretical analysis in Sect.~\ref{sec:3}, we find the overhead of the $\ell^2$-error between the quantum solution and the reference solution (truncated solution) is linear in the time step. Here, the reference solution is chosen as the exact ITE for the $N$-discretized Hamiltonian $H_N$ of the initial state (by the matrix multiplication). In particular, if the potential $V$ vanishes, then the reference solution is the $N$-truncated (analytical) solution which is the projection of the true solution to an $N$-dimensional subspace $\mathcal{V}_N$ spanned by $N$ eigenvectors of the Laplace operator. 

In the numerical example in Sect.~\ref{subsec:4-1}, we plotted the evolution of the $\ell^2$-error for different time steps: $\Delta\tau=0.002,0.001,0.0005$. Here, we provide another plot of the error regarding $\Delta\tau$ at $4$ time points $t=0.002,0.01,0.05,0.1$ in Fig.~\ref{appD:Fig5}.  
\begin{figure}[htb]
\centering
\resizebox{8cm}{!}{
\includegraphics[keepaspectratio]{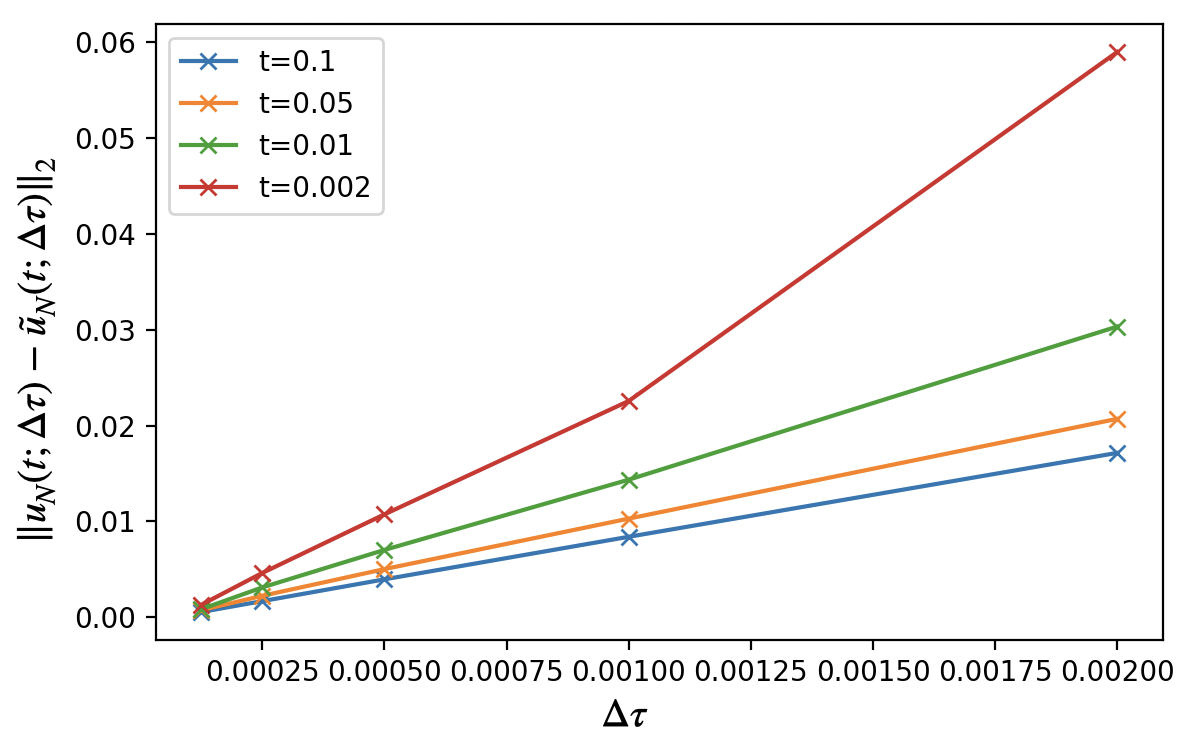}
}
\caption{(1D example) Plot of $\ell^2$-errors between the $2^6$-quantum solutions and the reference solution regarding time steps $\Delta\tau=0.000125,0.00025,0.0005,0.001,0.002$. The results at $4$ different time points $t=0.002,0.01,0.05,0.1$ are shown in different colors. }
\label{appD:Fig5}
\end{figure}
We find that for sufficiently large simulation time $t$ ($t=0.05,0.1$) or sufficiently small time step $\Delta\tau$ ($\Delta\tau\le 0.001$), the linear dependence on the time step is confirmed numerically, which coincides with the theoretical overhead. 

Next, we confirm the numerical dependence on time step in the more involved 2D example in Sect.~\ref{subsec:4-2}. Since the potential $V$ does not vanish in this example, we should split the error into the Suzuki-Trotter error and the approximation error. 
To obtain the approximation error, we introduce an intermediate solution by the quantum circuit in Fig.~\ref{sec3:Fig1}, replacing the approximate PITE circuit by the exact PITE circuit, as the reference solution. The $\ell^2$-errors between the $(2^4,2^7)$-quantum solutions and the reference solution are shown in Fig.~\ref{appD:Fig6}. 
\begin{figure}[htb]
\centering
\resizebox{8cm}{!}{
\includegraphics[keepaspectratio]{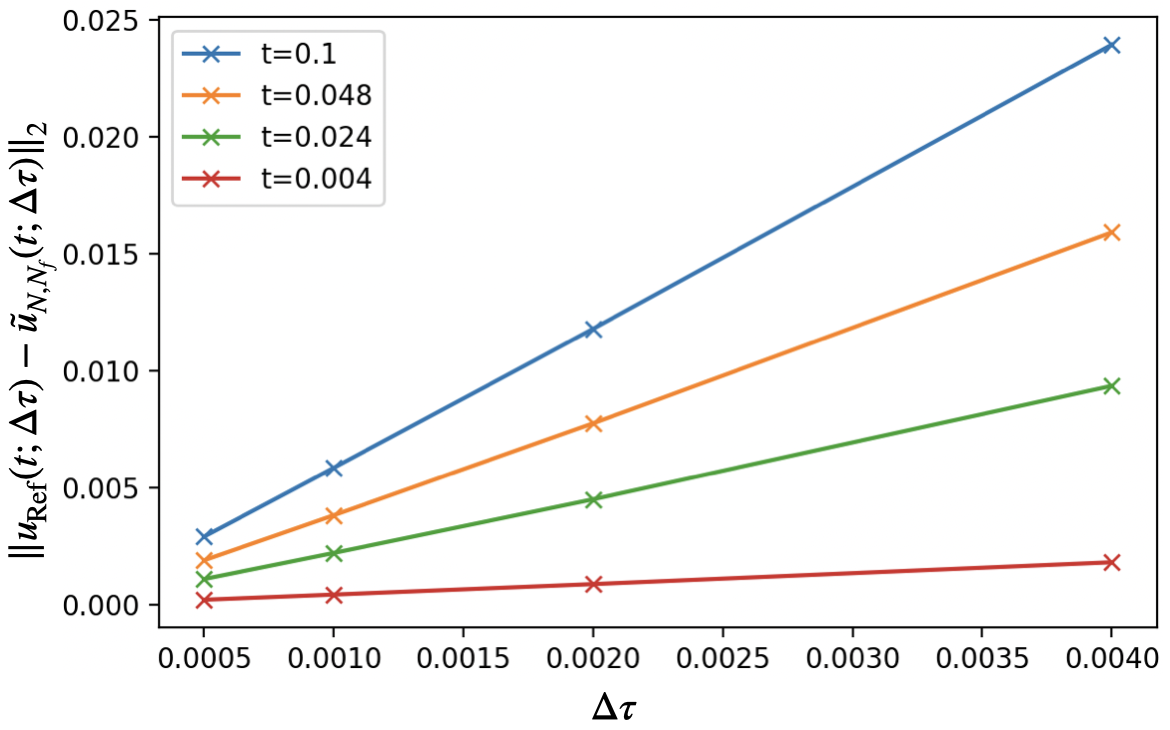}
}
\caption{(2D example) Plot of $\ell^2$-errors between the $(2^4,2^7)$-quantum solutions and the reference solution regarding time steps $\Delta\tau=0.0005,0.001,0.002,0.004$. The results at $4$ different time points $t=0.004,0.024,0.048,0.1$ are shown in different colors. }
\label{appD:Fig6}
\end{figure}
A linear dependence on $\Delta\tau$ is again confirmed, but the slope becomes larger for larger simulation time $t$, which is opposite to the 1D example in Fig.~\ref{appD:Fig5}. According to our comment in the last paragraph of \ref{subsec:4-3-1}, this owes to the presence of a function potential $V$. Without providing the details, we also confirmed another 1D example with a quadratic function potential that has a peak at the center of the domain, and the result shows that the $\ell^2$-error decreases at first and then keeps increasing. This implies that the approximation error depends on $\Delta\tau$ linearly, but its dependence on $t$ is not clear, which is influenced by the initial condition and the form of the potential as shown in Figs.~\ref{sec4:Fig1},\ref{sec4:Fig3}. 

\subsection{Suzuki-Trotter error: numerical dependence on time step}
\label{subsec:4-3-4}

Finally, we check the Suzuki-Trotter error numerically, and pay attention to its dependence on the time step. In the 2D example in Sect.~\ref{subsec:4-2}, we choose the reference solution as the exact ITE for the $N$-discretized Hamiltonian $H_N$ of the initial condition, which is calculated by the matrix multiplication. 
Here, we consider the intermediate solution described in \ref{subsec:4-3-3} that has no approximation error as the quantum solution. More precisely, the intermediate solution is derived by the quantum circuit in Fig.~\ref{sec3:Fig1}, where we substitute the RTE operators: 
$$
U_{\text{RTE}}\left(\pm \sqrt{2\Delta\tau D_N^{(2)}}\right) \quad \text{and}\quad U_{\text{RTE}}\left(\pm \sqrt{2\Delta\tau (V_N-V_0 I))}\right)
$$ 
with 
$$
U_{\text{RTE}}\left(\pm \arccos\left(\mathrm{exp}\left(-\Delta\tau D_N^{(2)}\right)\right)\right) \quad \text{and}\quad U_{\text{RTE}}\left(\pm \arccos\left(\mathrm{exp}\left(-\Delta\tau (V_N-V_0 I)\right)\right)\right),
$$
respectively. In other words, we consider the exact PITE circuit. Since $D_N^{(2)}$ and $V_N-V_0 I$ are both diagonal matrices, for the above RTE operators, we can use the quantum circuit in \cite{Zhang.2024} for the exact implementation. With a first-order Suzuki-Trotter formula:
\begin{align*}
\mathrm{exp}\left(-T H_N\right) &\approx \mathrm{exp}\left(-V_0 T\right)\bigg(F_N \mathrm{exp}\left(\mathrm{i}\Delta\tau D_N^{(1)}\right) \mathrm{exp}\left(-\Delta\tau D_N^{(2)}\right)F_N^\dag \\
&\quad \mathrm{exp}\left(-\Delta\tau (V_N-V_0 I)\right) \bigg)^K, 
\end{align*}
or a second-order Suzuki-Trotter formula: 
\begin{align*}
\mathrm{exp}\left(-T H_N\right) &\approx \mathrm{exp}\left(-V_0 T\right) \bigg(F_N \mathrm{exp}\left(\mathrm{i}\Delta\tau D_N^{(1)}/2\right) \mathrm{exp}\left(-\Delta\tau D_N^{(2)}/2\right) F_N^\dag \\
&\quad \mathrm{exp}\left(-\Delta\tau (V_N-V_0 I)\right) F_N \mathrm{exp}\left(\mathrm{i}\Delta\tau D_N^{(1)}/2\right) \mathrm{exp}\left(-\Delta\tau D_N^{(2)}/2\right) F_N^\dag\bigg)^K, 
\end{align*}
where $V_0, D_N^{(1)}, D_N^{(2)}$, and etc. are defined in Sect.~\ref{subsec:3-1}, we plot the $\ell^2$-errors between the quantum solutions and the reference solution regarding $\Delta\tau$ or $(\Delta\tau)^2$ in Fig.~\ref{appD:Fig7}. 
\begin{figure}[htb]
\centering
\resizebox{15cm}{!}{
\includegraphics[keepaspectratio]{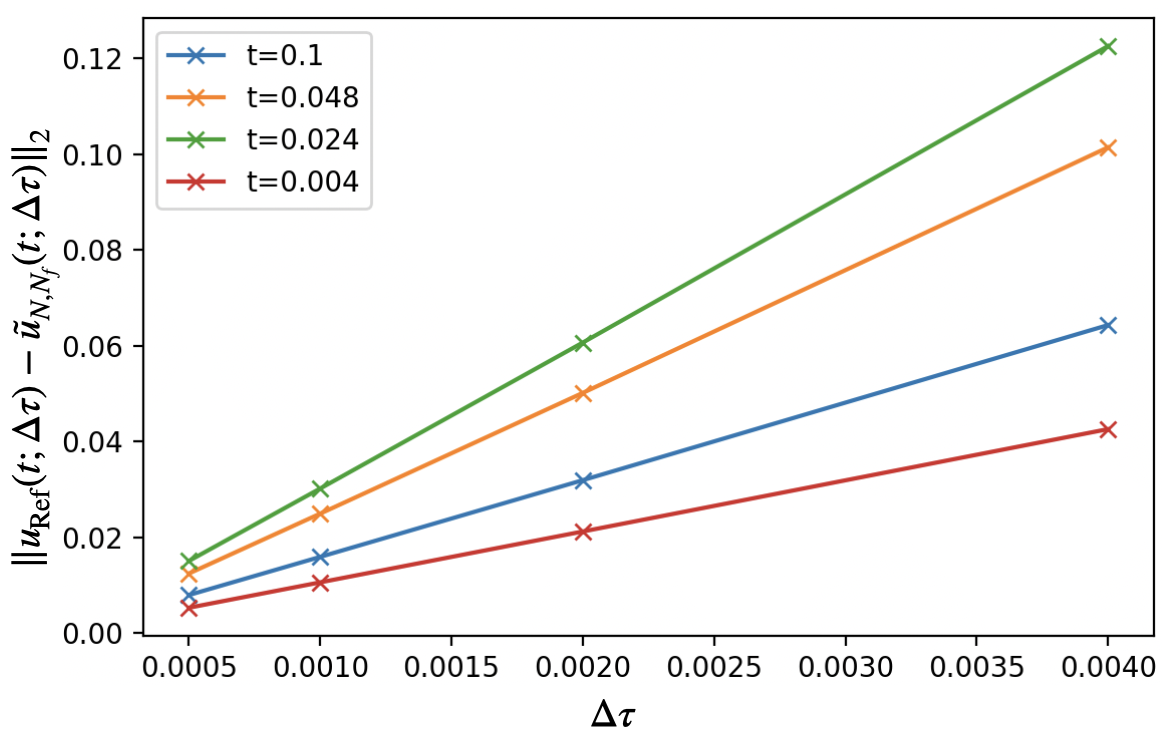}
\includegraphics[keepaspectratio]{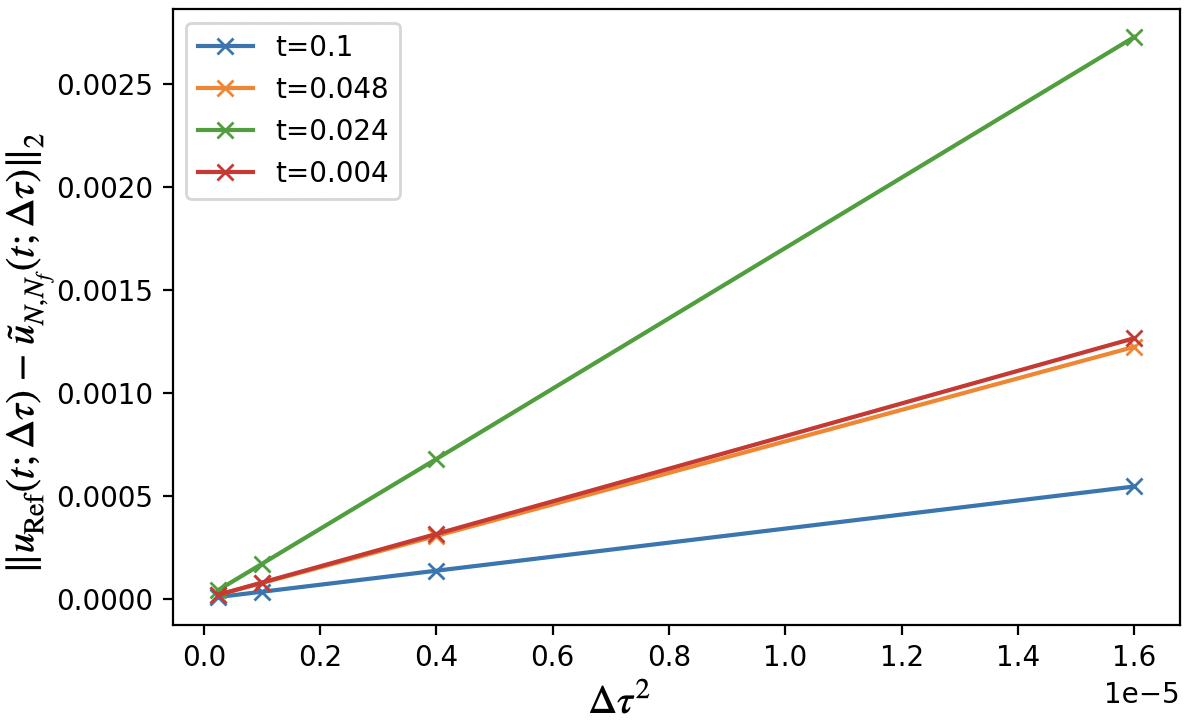}
}
\caption{(2D example) Plot of $\ell^2$-errors between the exact $(2^4,2^7)$-quantum solutions and the reference solution regarding time steps $\Delta\tau=0.0005,0.001,0.002,0.004$. In the left subplot, a first-order Suzuki-Trotter formula is employed to derive the quantum solutions, while in the right subplot, a second-order Suzuki-Trotter formula is employed, and the $x$-axis is set to be $(\Delta\tau)^2$. 
The results at $4$ different time points $t=0.004,0.024,0.048,0.1$ are shown in different colors. }
\label{appD:Fig7}
\end{figure}
For fixed simulation time $t$, we find linear dependence and quadratic dependence on $\Delta\tau$ for the cases using a first-order Suzuki-Trotter formula and a second-order Suzuki-Trotter formula, respectively. This coincides with the theoretical overheads discussed in \ref{prf:tot-err-esti} and \cite{Jahnke.2000}, and it seems that the theoretical linear/quadratic dependence on $\Delta\tau$ is not only an overhead but also a sharp estimate. 
Moreover, we find by focusing on the slopes that the Suzuki-Trotter errors are relatively small for both a small time and a large time. 
In Fig.~\ref{appD:Fig7} for the current 2D example, the Suzuki-Trotter errors are most significant at $t=0.024$, when the initial point source completely enters the zone of absorption (see Fig.~\ref{sec4:Fig2}).

\section{Details on comparison}
\label{appF}

\subsection{Higher order AAPITE}
\label{subsec:F-1}

We introduce the higher order AAPITE operators by keeping more terms in the Taylor expansion of $\arccos(\mathrm{exp}\left(-x^2\right))$ near $x=0+$. For example, the second-order alternative PITE operator is given by 
\begin{equation}
\label{appD:eq-2oPITE}
\cos\left(\sqrt{2\Delta\tau \mathcal{H}}-\sqrt{2\Delta\tau^3\mathcal{H}^3}/6\right) = \mathrm{exp}\left(-\Delta\tau \mathcal{H}\right) + O(\Delta\tau^3),
\end{equation}
and the fourth-order one is 
\begin{equation}
\label{appD:eq-4oPITE}
\cos\left(\sqrt{2\Delta\tau \mathcal{H}}-\sqrt{2\Delta\tau^3\mathcal{H}^3}/6+\sqrt{2\Delta\tau^5\mathcal{H}^5}/120+\sqrt{2\Delta\tau^7\mathcal{H}^7}/336\right) = \mathrm{exp}\left(-\Delta\tau \mathcal{H}\right) + O(\Delta\tau^5).
\end{equation}
The above coefficients come from the Taylor expansion:
$$
\arccos(\mathrm{exp}\left(-x^2\right)) = \sqrt{2}x - \frac{\sqrt{2}}{6}x^3 + \frac{\sqrt{2}}{120}x^5 + \frac{\sqrt{2}}{336}x^7 + \cdots.
$$

\subsection{Comparison to time-sequential QLSA}
\label{subsec:F-2}
 
In this part, we discuss the essential difference between the approximate PITE algorithms, including the original approximate PITE, and the time-sequential QLSA, mainly focusing on the error of the time discretization (or approximation). We consider the spatial discretized equation:
\begin{equation}
\label{appF:eq-gov}
\partial_t \mathbf{v}(t) = -H_N \mathbf{v}(t), \quad t>0,
\end{equation}
where $H_N$ is the discretized matrix in \eqref{sec3:eq-HN}. Then, we apply the backward Euler method in time and obtain the equations for the approximate solution $\mathbf{\tilde v}$: 
$$
\frac{\mathbf{\tilde v}(m\Delta\tau)-\mathbf{\tilde v}((m-1)\Delta\tau)}{\Delta\tau} = -H_N \mathbf{\tilde v}(m\Delta\tau), \quad m=1,2,\ldots,
$$
which yields the updated operation in each time step:
$$
\mathbf{\tilde v}(m\Delta\tau) = (I + \Delta\tau H_N)^{-1} \mathbf{\tilde v}((m-1)\Delta\tau), \quad m=1,2,\ldots.
$$
We call such a multistep method the time-sequential QLSA for the time evolution equations, where a QLSA is used to implement the inverse operator in each time step. 
On the other hand, the AAPITE solves
$$
\mathbf{\hat v}(m\Delta\tau) = \cos\left(\sqrt{2\Delta\tau H_N}\right) \mathbf{\hat v}((m-1)\Delta\tau), \quad m=1,2,\ldots,
$$
in each time step, instead. Recalling that the exact solution to Eq.~\eqref{appF:eq-gov} satisfies
$$
\mathbf{v}(m\Delta\tau) = \mathrm{exp}\left(-\Delta\tau H_N\right) \mathbf{v}((m-1)\Delta\tau), \quad m=1,2,\ldots,
$$
the theoretical difference between the time-sequential QLSA and the AAPITE algorithm lies in the different ways of approximation to the ITE operator in each time step. 
As a result, to compare their performance with small time step $\Delta\tau$, we only need to check the first several Taylor coefficients of the underlying functions: $g_{\text{exa}}(y)=\mathrm{exp}\left(-y\right)$ (exact), $g_{\text{QLSA}}(y)=1/(1+y)$ (QLSA), and $g_{\text{AAP}}(y)=\cos(\sqrt{2y})$ (AAPITE), respectively. We can also discuss the original approximate PITE whose underlying function is 
$$
g_{\text{OAP}}(y) = \frac{1}{m_0}\cos\left(\arccos m_0 + \frac{m_0}{\sqrt{1-m_0^2}}y\right), \quad m_0\in (0,1).
$$
By a direct calculation, we obtain
\begin{align*}
& g_{\text{exa}}(y) = 1 - y + \frac{1}{2} \mathrm{exp}\left(-\xi_0\right)y^2, \quad \xi_0\in [0,y], \\
& g_{\text{QLSA}}(y) = 1 - y + \frac{1}{(1+\xi_1)^3}y^2, \quad \xi_1\in [0,y], \\
& g_{\text{AAP}}(y) = 1 - y + \frac{\sin\sqrt{2\xi_2}-\sqrt{2\xi_2}\cos\sqrt{2\xi_2}}{2(\sqrt{2\xi_2})^3}y^2, \quad \xi_2\in [0,y], \\
& g_{\text{OAP}}(y) = 1 - y - \frac{m_0}{2(1-m_0^2)}\cos\left(\arccos m_0 + \frac{m_0}{\sqrt{1-m_0^2}}\xi_3\right)y^2, \quad \xi_3\in [0,y].
\end{align*}
Moreover, we have
$$
\lim_{\xi_0\to 0} \frac{1}{2}\mathrm{exp}\left(-\xi_0\right) = \frac{1}{2}, \ 
\lim_{\xi_1\to 0} \frac{1}{(1+\xi_1)^3} = 1, \
\lim_{\xi_2\to 0+} \frac{\sin\sqrt{2\xi_2}-\sqrt{2\xi_2}\cos\sqrt{2\xi_2}}{2(\sqrt{2\xi_2})^3} = \frac{1}{6}, 
$$
and 
$$
\lim_{\xi_3\to 0} -\frac{m_0}{2(1-m_0^2)}\cos\left(\arccos m_0 + \frac{m_0}{\sqrt{1-m_0^2}}\xi_3\right) = -\frac{m_0^2}{2(1-m_0^2)}\in (-\infty,0), 
$$
for $m_0\in (0,1)$. This implies that $g_{\text{QLSA}}$, $g_{\text{AAP}}$, and $g_{\text{OAP}}$ are all the first-order approximations of $g_{\text{exa}}$ near the origin $y=0$. Therefore, we find that
\begin{itemize}
\item The QLSA, the algorithms using the original approximate PITE and the AAPITE have comparable orders of precision regarding the time step $\Delta\tau$. 

\item Compared to the QLSA, the AAPITE algorithm has slightly better precision for sufficiently small $\Delta\tau$ since the second-order coefficient is closer to the exact one (i.e., $|1/6-1/2|<|1-1/2|$). 

\item The second-order coefficient of the original approximate PITE is far away from the desired value $1/2$, e.g. the limit is $-81/38$ for $m_0=0.9$ (the sign itself is opposite). Thus, we need to take much smaller time step $\Delta\tau$ than those for the QLSA and the AAPITE algorithm to derive a comparable precision. 

\item One needs a shift of eigenvalues for the AAPITE algorithm if the smallest eigenvalue of $H_N$ is negative, but the QLSA as well as the original approximate PITE is valid for any Hermitian $H_N$ provided that we take a sufficiently small $\Delta\tau$.
\end{itemize}
Furthermore, we plot the underlying functions of the different approximations of the ITE operator to provide an intuitive comparison of the above mentioned methods. Under a relatively large scale $[0,1]$ and a relatively small scale $[0,0.01]$, Fig.~\ref{appF:Fig1} shows that the errors of the different methods do not change a lot if the time step $\Delta\tau$ is sufficiently close to $0$, while the higher order AAPITE operators give good approximation even for relatively large $\Delta\tau$. 
\begin{figure}[htb]
\centering
\resizebox{15cm}{!}{
\includegraphics[keepaspectratio]{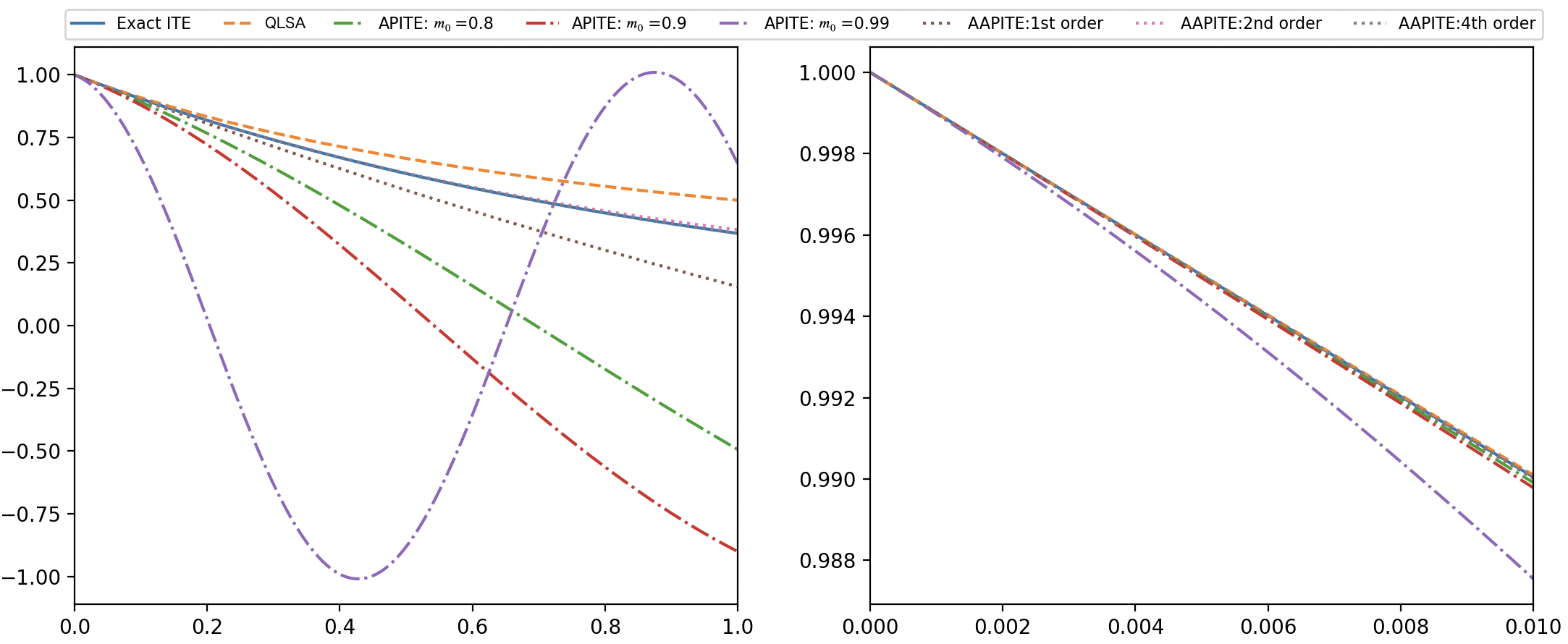}
}
\caption{Plots of the underlying functions of different approximations. The $x$-axis is the variable and the $y$-axis denotes the underlying functions of different approximations. The functions in a relatively large interval $[0,1]$ are plotted in the left subplot, while those near $x=0$ are plotted in the right subplot. }
\label{appF:Fig1}
\end{figure}
Since the gate complexity of the quantum algorithm is proportional to $T/\Delta\tau$ for a given simulation time $T$, the higher order AAPITE operators can be more efficient if their implementation does not require substantially additional gate complexity than the first-order one. 
In practice, we conjecture that there is an optimal order to derive the optimized gate count/circuit depth because of the trade-off between the improvement of the order and the increase of the prefactor. 

Taking into account the vanishing of the success probability of the original approximate PITE, and noting that we can also develop higher order time-sequential QLSA by substituting the backward Euler method with some advanced ones, we conclude:

\noindent (a) The AAPITE algorithm is usually more efficient than the original approximate PITE algorithm as a better approximation of the ITE operator. An exceptional case is that the RTE operator for $\sqrt{H_N}$ is much more expensive than the RTE operator for $H_N$ itself, and simultaneously, the simulation time $T$ is small with a relatively large error bound so that we need only limited PITE steps.   

\noindent (b) The AAPITE algorithm has a comparable (slightly better) precision than the time-sequential QLSA for a fixed time step $\Delta\tau$. 
To determine which algorithm should be used, we have to compare the detailed gate complexity/circuit depth for the basic oracles: $\mathrm{exp}\left(\pm \mathrm{i}\sqrt{2\Delta\tau H_N}\right)$ for the AAPITE and $(I+\Delta\tau H_N)^{-1}$ for the QLSA. This depends on both the structure of $H_N$ and the strategy used to realize the inverse operator (e.g. QPE-based \cite{HHL09}, LCU-based \cite{Childs.2017}, QSVT-based \cite{Gilyen.2019}, etc.). 

\noindent (c) The higher order AAPITE or the higher order time-sequential QLSA can be better than the first-order ones. For fixed matrix $H_N$, simulation time $T$, and desired error bound $\varepsilon$, there exists an optimal order for either the AAPITE algorithm or the time-sequential QLSA so that the gate complexity/circuit depth is optimized.

\end{document}